\documentclass[11pt]{amsart} 
\usepackage{amscd,amssymb,amsxtra}
\usepackage[mathscr]{eucal}
\usepackage{mathabx}
\usepackage{comment}
\usepackage{color}
\usepackage{enumerate}
\makeatletter
\let\@secnumfont\bfseries
\def\section{\@startsection{section}{1}%
  \z@{4\linespacing\@plus\linespacing}{\linespacing}%
  {\bfseries\centering}}
\def\introsection{\@startsection{section}{1}%
  \z@{3\linespacing\@plus\linespacing}{\linespacing}%
  {\bfseries\centering}}
\def\subsection{\@startsection{subsection}{2}%
   \z@{1.25\linespacing\@plus.7\linespacing}{.5\linespacing}%
   {\normalfont\bfseries}}
\def\subsectionsinline{\def\subsection{\@startsection{subsection}{2}%
  \z@{1\linespacing\@plus.7\linespacing}{-.5em}%
  {\normalfont\bfseries}}}

\makeatother

\theoremstyle{definition}

\newtheorem{example}[equation]{Example}

\newtheorem{question}[equation]{Question}
\newtheorem{construction}[equation]{Construction}

\newtheorem*{definition*}{Definition}
\newtheorem*{example*}{Example}
\newtheorem*{problem*}{Problem}
\newtheorem*{exercise*}{Exercise}
\newtheorem*{question*}{Question}
\newtheorem*{construction*}{Construction}

\theoremstyle{remark}

\newtheorem{remark}[equation]{Remark}

\newtheorem*{note*}{Note}
\newtheorem*{notation*}{Notation}
\newtheorem*{remark*}{Remark}
\newtheorem*{data*}{Data}

\theoremstyle{plain}

\newtheorem{proposition}[equation]{Proposition}

\newtheorem{claim}[equation]{Claim}
\newtheorem{proposal}[equation]{Proposal}

\newtheorem{hypothesis}[equation]{Hypothesis}

\newtheorem*{theorem*}{Theorem}
\newtheorem*{corollary*}{Corollary}
\newtheorem*{lemma*}{Lemma}
\newtheorem*{proposition*}{Proposition}
\newtheorem*{conjecture*}{Conjecture}
\newtheorem*{claim*}{Claim}
\newtheorem*{proposal*}{Proposal}
\newtheorem*{conclusion*}{Conclusion}
\newtheorem*{hypothesis*}{Hypothesis}
\newtheorem*{assumption*}{Assumption}

\numberwithin{equation}{section}

\definecolor{refkey}{rgb}{0,.6,.4}

\renewcommand{\:}{\colon}

\newcommand{\Ahat}{{\hat A}}

\newcommand{\CC}{{\mathbb C}}
\newcommand{\CP}{{\mathbb C\mathbb P}}

\newcommand{\FF}{\mathbb F}

\DeclareMathOperator{\Hom}{Hom}
\DeclareMathOperator{\id}{id}
\DeclareMathOperator{\Map}{Map}

\newcommand{\PP}{{\mathbb P}}
\DeclareMathOperator{\pt}{pt}
\newcommand{\QQ}{{\mathbb Q}}
\newcommand{\RP}{{\mathbb R\mathbb P}}
\newcommand{\RR}{{\mathbb R}}
\newcommand{\TT}{\mathbb T}
\DeclareMathOperator{\Spin}{Spin}

\newcommand{\ZZ}{{\mathbb Z}}

\newcommand{\chiup}{\raise.5ex\hbox{$\chi$}}
\newcommand{\cir}{S^1}

\newcommand{\dbar}{{\bar\partial}}

\newcommand{\inv}{^{-1}}
\DeclareRobustCommand{\mstrut}{^{\vphantom{1*\prime y\vee M}}}
\newcommand{\mlstrut}{_{\vphantom{1*\prime y}}}

\newcommand{\temsquare}{\raise3.5pt\hbox{\boxed{ }}}

\newcommand{\zmod}[1]{\ZZ/#1\ZZ}

\newcommand{\zt}{\zmod2}

\usepackage[all,2cell]{xy}\renewcommand{\cir}{\ensuremath{S^1}}
\usepackage{graphicx}
\usepackage{epsf}
\usepackage{ifpdf} 
\ifpdf
\def\epsorpdf{pdf}
\else
\def\epsorpdf{eps}
\fi

\usepackage[colorlinks,citecolor=refkey]{hyperref}

\definecolor{refkey}{rgb}{0,.8,.2}\definecolor{labelkey}{rgb}{1,0,0}

\DeclareMathOperator{\Bord}{Bord}
\DeclareMathOperator{\Cliff}{Cliff}
\DeclareMathOperator{\Euler}{Euler}
\DeclareMathOperator{\Ext}{Ext}
\DeclareMathOperator{\Line}{Line}
\DeclareMathOperator{\Pin}{Pin}
\DeclareMathOperator{\Rep}{Rep}
\DeclareMathOperator{\Sign}{Sign}
\DeclareMathOperator{\Vect}{Vect}
\newcommand{\Anob}{Anom_{\textnormal{bose}}(d,G,\phi )}

\newcommand{\Anof}{Anom_{\textnormal{fermi}}(d,G,\phi )}
\newcommand{\BF}{\Bord_n(\sF)}
\newcommand{\BNG}{B\mstrut _\nabla G}
\newcommand{\CZ}{\CC/\ZZ}
\newcommand{\Cb}{\overline{\CC}}

\newcommand{\Cx}{\CC^\times }
\newcommand{\Fg}{F_{\textnormal{geometric}}}
\newcommand{\Ft}{F_{\textnormal{topological}}}
\newcommand{\ft}{\FF_2}
\newcommand{\ICZ}{I\CZ}

\newcommand{\IZ}{I{\ZZ}}

\newcommand{\RZ}{\RR/\ZZ}
\newcommand{\SPTba}[3]{SPT_{\textnormal{bose}}(#1,#2,#3)}
\newcommand{\SPTb}{SPT_{\textnormal{bose}}(d,G,\phi )}

\newcommand{\SPTf}{SPT_{\textnormal{fermi}}(d,G,\phi )}
\newcommand{\SREba}[3]{SRE_{\textnormal{bose}}(#1,#2,#3)}
\newcommand{\SREb}{SRE_{\textnormal{bose}}(d,G,\phi )}
\newcommand{\SREfa}[3]{SRE_{\textnormal{fermi}}(#1,#2,#3)}

\newcommand{\SREf}{SRE_{\textnormal{fermi}}(d,G,\phi )}

\newcommand{\TZ}{T\ZZ}
\newcommand{\Zf}{\ZZ_\phi }
\newcommand{\Zt}{\widetilde{\ZZ}}
\newcommand{\bc}{\bar{c}}
\newcommand{\bord}[1]{\Bord_{#1}}
\newcommand{\bort}{\Bord_{\langle n-1,n  \rangle}}
\newcommand{\bo}{\mathbf{1}}
\newcommand{\fA}{{\phi \mstrut _{A}}}
\newcommand{\fB}{{\phi \mstrut _{B}}}
\newcommand{\fE}{{\phi \mstrut _{E}}}
\newcommand{\fHZ}{{\phi \mstrut _{H\ZZ}}}
\newcommand{\fIZ}{{\phi \mstrut _{I\ZZ}}}
\newcommand{\fTZ}{{\phi \mstrut _{T\ZZ}}}
\newcommand{\otIZ}{{\bar\tau \mstrut _{I\ZZ}}}
\newcommand{\otTZ}{{\bar\tau \mstrut _{T\ZZ}}}
\newcommand{\otriv}{\underline{\RR}}
\newcommand{\pmo}{\mu _2}

\newcommand{\sCfx}{\sC_{\textnormal{fermi}}^{\times }}
\newcommand{\sCx}{\sC^{\times }}
\newcommand{\sC}{\mathcal{C}}
\newcommand{\sFg}{\sF_{\textnormal{geometric}}}
\newcommand{\sFt}{\sF_{\textnormal{topological}}}
\newcommand{\sF}{\mathcal{F}}
\newcommand{\sH}{\mathscr{H}}
\newcommand{\sX}{\mathscr{X}}
\newcommand{\tF}{\widetilde{F}}
\newcommand{\tIZ}{{\tau \mstrut _{I\ZZ}}}
\newcommand{\tTZ}{{\tau \mstrut _{T\ZZ}}}

\newcommand{\ta}{\tau \mstrut _{\le n}\alpha }

\newcommand{\tia}{\tilde{\alpha }}
\newcommand{\trF}{\tau \mstrut _{\le\,d-1}F}
\newcommand{\tsF}{\widetilde{\sF}}
\newcommand{\vtriv}[1]{\underline{#1}}
\newcommand{\wIZ}{{w \mstrut _{I\ZZ}}}
\newcommand{\wTZ}{{w \mstrut _{T\ZZ}}}

\begin{document}

\abovedisplayskip18pt plus4.5pt minus9pt
\belowdisplayskip \abovedisplayskip
\abovedisplayshortskip0pt plus4.5pt
\belowdisplayshortskip10.5pt plus4.5pt minus6pt
\baselineskip=15 truept
\marginparwidth=55pt

\makeatletter
\renewcommand{\tocsection}[3]{%
  \indentlabel{\@ifempty{#2}{\hskip1.5em}{\ignorespaces#1 #2.\;\;}}#3}
\renewcommand{\tocsubsection}[3]{%
  \indentlabel{\@ifempty{#2}{\hskip 2.5em}{\hskip 2.5em\ignorespaces#1%
    #2.\;\;}}#3} 
\renewcommand{\tocsubsubsection}[3]{%
  \indentlabel{\@ifempty{#2}{\hskip 5.4em}{\hskip 5.4em\ignorespaces#1%
    #2.\;\;}}#3} 
\makeatother

\setcounter{tocdepth}{3}



 \title[Short Range Entanglement and Invertible Field Theories]{Short-Range Entanglement and Invertible Field Theories} 
 \author[D. S. Freed]{Daniel S.~Freed}
 \thanks{The work of D.S.F. is supported by the National Science Foundation
under grant DMS-1207817.  Some of this work was undertaken at the
Mathematical Sciences Research Institute as part of the program on Algebraic
Topology.}  
 \address{Department of Mathematics \\ University of Texas \\ Austin, TX
78712} 
 \email{dafr@math.utexas.edu}
 \date{August 10, 2014}
 \begin{abstract} 
 Quantum field theories with an energy gap can be approximated at long-range
by \emph{topological} quantum field theories.  The same should be true for
suitable condensed matter systems.  For those with \emph{short range
entanglement} (SRE) the effective topological theory is \emph{invertible},
and so amenable to study via stable homotopy theory.  This leads to concrete
topological invariants of gapped SRE phases which are finer than existing
invariants.  Computations in examples demonstrate their effectiveness.
 \end{abstract}
\maketitle



{\small
\def\reftext{References}
\renewcommand{\tocsection}[3]{%
  \begingroup 
   \def\tmp{#3}%
   \ifx\tmp\reftext
  \indentlabel{\phantom{1}\;\;} #3%
  \else\indentlabel{\ignorespaces#1 #2.\;\;}#3%
  \fi\endgroup}
\tableofcontents}

   \section{Introduction}\label{sec:1}

The long-range behavior of gapped systems in condensed matter physics is
accessible via topology.  For noninteracting fermionic systems there is a
classification of topological phases using ideas related to
$K$-theory~\cite{K1}.  Over the past few years the interacting case has been
vigorously studied, for both fermionic systems and bosonic systems, with an
emphasis on \emph{short-range entanglement} (SRE); a small sampling of papers
is~\cite{CGW1, CGW2, CGLW, GW, LV, VS, We, PMN, CFV, WPS, HW, WS, Ka1, Ka2,
WGW, KTTW}.  Particular \emph{symmetry protected topological} (SPT) phases
are captured by group cohomology~\cite{CGLW, GW}, but other investigations
(e.g.~\cite{VS}) reveal the existence of additional SRE phases and raise the
question of a complete classification.  In this paper we propose an invariant
of bosonic and fermionic SRE topological phases constructed from effective
field theory.  Computations and examples demonstrate that it effectively
detects known SRE phases.

Our proposal applies to gapped systems which at low energy (long time) can be
approximated by topological field theories.  Such a system must be
sufficiently local that it can be formulated on arbitrary manifolds, and it
must have a continuum limit which is a field theory, at least at low energy.
We do not investigate microscopic behavior at all in this paper, but rather
simply assume the existence of a long-range topological
theory.\footnote{Rather than a single long-range approximation, we envision a
connected space of long-range theories.}  Short range entanglement, or the
absence of topological order, is a microscopic assumption.
Kitaev~\cite{K2,K3,K4} has been studying SRE phases from first principles
microscopically,\footnote{We remark that there is an alternative microscopic
definition of SRE proposed by Chen-Gu-Wen~\cite{CGW1}.} and has suggested the
macroscopic consequence that the long-range topological theory has a unique
vacuum on any background manifold.  We go further and assume that the
long-range topological field theory describing an SRE phase is \emph{fully
extended} and \emph{invertible}, concepts that we explain below.  From
mathematical investigations it has been known for a long time that fully
extended, invertible, \emph{topological} field theories are equivalent to
maps between spectra in the sense of algebraic topology.  This link with
stable homotopy theory is the basis of our proposal.
 
While our proposal in~\S\ref{subsec:5.2} is specific and precise, we neither
formulate nor prove a mathematical theorem which justifies it.  Also, while
we enumerate groups which should house invariants of SRE phases, we do not
argue either that the effective field theory is a complete invariant or that
every possible effective field theory is realized by a microscopic system.
In place of proof the paper marshals evidence in two stages.
Pre-\S\ref{subsec:5.2} is a long conceptual march leading to the proposal.
Post-\S\ref{subsec:5.2} is a series of experimental checks, including the
relationship to group cohomology, boundary terminations, and specific
computations.  We include a long discussion in~\S\ref{subsec:6.3} about
detecting Kitaev's $E_8$~phase using invertible topological field theories.
There are additional possible effective field theories which are
``$4^{\textnormal{th}}$ roots'', perhaps an indication that not all possible
effective theories are realized by microscopic systems.  Another example, the
``3d bosonic $E_8$~phase with half-quantized surface thermal Hall
effect''~\cite{VS}, \cite{BCFV}, is also treated in detail.

It may be useful to broadly characterize our proposal in field-theoretic
language: whereas the group cohomology captures pure gauge theories, the
additional SRE phases contain couplings to gravity or are purely
gravitational.  The SPT phases have no purely gravitational component.  More
fundamentally, the long-range field theory is envisioned as the low-energy
behavior of the coupling to gravity of the original system.  (And, if there
are global symmetries, we gauge them and so couple to gauge theory too.)
 
Bordism as a tool to classify SPT phases appears differently in the recent
papers of Kapustin~\cite{Ka1}, \cite{Ka2}, \cite{KTTW}.  In unpublished work
Kitaev~\cite{K2,K3,K4} develops a classification of SRE phases based on
microscopic considerations.  Their results and approach differ from ours, and
it will be very interesting to reconcile them.  It may be that there is more
microscopic information in the physics which leads to different or additional
input into the effective field theories.  In particular, there are a few
ingredients in our proposal (choice of tangential structure, choice of target
spectrum) which involve leaps of faith and can easily be adjusted if further
microscopic implications are discovered.

We begin in~\S\ref{sec:2} with an exposition of several formal points in
field theory: extended field theory, invertible field theory, relative field
theory, anomalies, global symmetries and gauging, topological field theory,
unitarity.  While many of these concepts are familiar to physicists, the
mathematical language may be unfamiliar and we hope to provide some bridge
here.  We touch on the cobordism hypothesis, which classifies fully extended
topological theories, in~\S\ref{subsec:2.5}; the precise nature of this
classification is illuminated in~\S\ref{subsec:2.7} with a toy example in
preparation for a later discussion of the ``Kitaev $E_8$~phase''
in~\S\ref{subsec:6.3}.  But the cobordism hypothesis is overkill for
\emph{invertible} theories.  In~\S\ref{subsec:2.6} we describe the link
between fully extended, invertible, topological field theories and stable
homotopy theory in general terms.  That involves particular
\emph{Madsen-Tillmann spectra}, which are unstable analogs of Thom's bordism
spectra; we describe them in~\S\ref{sec:4}.  There we also discuss the
crucial theorem of Galatius-Madsen-Tillmann-Weiss \cite{GMTW} which
identifies these spectra as geometric realizations of bordism categories.
That these \emph{unstable} bordism spectra are appropriate to field theories
is natural since field theories are dimension-specific.  The material
in~\S\ref{sec:2} and~\S\ref{sec:4} is general background not particular to
condensed matter systems.
 
In~\S\ref{sec:3} we give much of the general argument about effective field
theories for SRE topological phases.  We begin in~\S\ref{subsec:3.1} with
elementary thoughts indicating why topology may sufficiently describe the low
energy behavior of gapped systems.  Most of our assumptions are stated
explicitly in~\S\ref{subsec:3.3}.  There is still conceptual work to
translate those assumptions into concrete mathematical statements.
Specifically, once we know that SRE gapped phases give rise to invertible
topological theories, there are still parameters to choose: the tangential
structure on manifolds representing space and the target category for the
field theory.  The latter is discussed in~\S\ref{subsec:5.1} separately for
bosonic and fermionic theories; the bosonic case is a bit surprising and we
settle on a kludge for the target spectrum.  For the tangential structures we
assume without much justification that in bosonic theories the space
manifolds are oriented and in fermionic theories they are spin.  Along the
way we encounter a few tricky issues, for example the gauging of antilinear
symmetries (\S\S\ref{subsubsec:2.4.3}, \ref{subsubsec:5.1.3}) and
implementation of unitarity (\S\ref{subsubsec:4.2.5}).
 
We state our proposal in~\S\ref{subsec:5.2}.  We divide theories into bosonic
and fermionic.  Also, symmetries may be anomalous and we propose a
classification of anomaly theories and anomalous gauged theories as well.
Possible effective invertible topological theories for SRE phases with fixed
symmetry form an abelian group; for SPT phases they form a subgroup which we
also delineate.
 
Our first deduction in~\S\ref{subsec:5.3} from the proposal is that the
phases previously identified using group cohomology are included.
In~\S\ref{subsec:6.1} we show that for bosonic theories in $d=1$ space
dimensions group cohomology provides a complete classification, which agrees
with known results~\cite{CLW}.  Already for fermionic theories in~$d=1$ the
situation is more interesting, as described in~\S\ref{subsec:6.2}: we detect
the Majorana chain~\cite{K6} in our classification.  In~\S\ref{subsec:6.3} we
identify bosonic $d=2$ SRE phases.  These were introduced by
Kitaev~\cite{K5}, \cite{K2} and are related to 2-spacetime\footnote{The
number~$d$ above is the dimension of space; spacetime has dimension~$d+1$.}
dimensional chiral conformal field theories whose chiral central charge is an
integer divisible by~8.  These central charges do not show up in the usual
account of the associated 3-spacetime dimensional topological field theory;
there only the reduction mod~8 is used.  Our explanation of how they fit in
here, and so the role of the chiral central charge as a real number not taken
mod~8, is based on the easier examples discussed in~\S\ref{subsec:2.7}.
In~\S\ref{subsec:6.4} we illustrate a constraint imposed by unitarizability.
In~\S\ref{subsec:6.5} we compute that the abelian group of $d=3$ bosonic
time-reversal symmetric effective SRE field theories is isomorphic
to~$(\zt)^{\times 2}$.  One generator is accounted for by group cohomology,
and we claim the other is the 3d bosonic $E_8$~phase with half-quantized
surface thermal Hall effect mentioned earlier.  Finally, in~\S\ref{sec:7} we
give a general discussion of boundary conditions/terminations/excitations and
use it to justify the aforementioned claim.  An appendix includes topological
computations which are needed in the text.
 
The notion of a non-extended \emph{invertible} field theory arose in joint
work with Greg Moore~\cite[\S5.5]{FM1}.  Fully extended invertible
topological theories, and the relation to stable homotopy theory, has been a
longstanding discussion topic with Mike Hopkins and Constantin Teleman, as
have many other general ideas described in~\S\ref{sec:2}.  In particular, we
used these ideas in~\cite{FHT} to construct a topological field theory based
on the Verlinde ring.  The specific application to SRE phases described here
crystallized during the \emph{Symmetry in Topological Phases} workshop in
Princeton, and I thank the organizers for inviting me.  I had long
conversations with Alexei Kitaev after a first draft of this paper was
complete, and those inspired a significant modification
of~\S\ref{subsubsec:5.1.1} and~\S\ref{subsec:6.3}.  I thank him for sharing
his perspectives.  I also thank Zheng-Cheng Gu, Mike Hopkins, Anton Kapustin,
Constantin Teleman, Ashvin Vishwanath, Kevin Walker, Oscal Randal-Williams,
and Xiao-Gang Wen for very helpful conversations and correspondence.

   \section{Field theories from a bordism point of view}\label{sec:2}

We begin with a formal viewpoint on the structure of a field theory, which is
the lens through which we analyze the long-range effective topological theory
in~\S\ref{sec:3}.  In the mathematics literature this approach was abstracted
in Segal's axioms for two-dimensional conformal field theory~\cite{S1} and in
Atiyah's axioms for topological field theories~\cite{A1}.  The
lectures~\cite{S2} treat general quantum field theories from this
perspective.  Our focus in this paper is on \emph{invertible} field theories,
which we define in~\S\ref{subsec:2.2}.  Other general topics we quickly
review include extended field theories, relative field theories, anomalies,
gauging symmetries, and unitarity.  We then focus on fully extended
topological theories, for which the powerful \emph{cobordism
hypothesis}~\cite{BD,L,F1} provides a classification result.  Invertible
topological theories can be analyzed using homotopy theory, and in this
section we explain why that is true but defer a more precise description
to~\S\ref{sec:4}.  We conclude with a few toy examples which illuminate
subtleties we will encounter in the condensed matter systems
of~\S\ref{sec:6}.  The subtleties discussed there may have broader interest.
There are many expositions of this material, in addition to the ones
referenced earlier in this paragraph, and here we offer another.  A somewhat
different point of view on topological field theories may be found
in~\cite{MW}.  The reader may wish to use this section for reference and skip
on first reading to later parts of the paper.

We remark that this formal viewpoint does not distinguish ``classical'' from
``quantum'', and indeed we will give examples of both types.  Another remark
is that the field theories which arise in condensed matter physics are
usually defined on spacetimes which are products of space and time, whereas
the discussion in this section models theories defined on more general
spacetimes.  We discuss the necessary modification in~\S\ref{subsec:4.2} and
account for it in the proposals of~\S\ref{subsec:5.2}.

  \subsection{Field theories}\label{subsec:2.1}

  \begin{figure}[ht]
  \centering
  \includegraphics[scale=1]{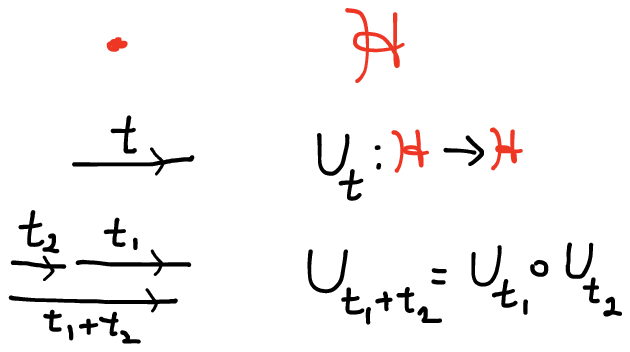}
  \caption{Quantum mechanical evolution}\label{fig:1}
  \end{figure}

Let $n$~be the \emph{spacetime} dimension of a field theory~$F$, which we
simply call the dimension of~$F$.  We write $n=d+1$ where $d$~is the
\emph{space} dimension.\footnote{In the condensed matter literature $d$~is
often called the dimension of the theory, whereas in quantum field theory and
string theory literature it is~$n$ which is the dimension.  We use the terms
`spacetime dimension' and `space dimension' to avoid confusion.}  The case
$n=1$ ($d=0$) is mechanics; there is only time.  A quantum mechanical system
assigns a complex vector space~$\sH$ (the `quantum Hilbert space') to a point
and the time evolution $U_t\:\sH\to\sH$ to a closed interval of length~$t$.
The group law $U_{t_1+t_2}=U_{t_1}\circ U_{t_2}$ is encoded by gluing
intervals, as illustrated in Figure~\ref{fig:1}.  An $n$-dimensional
Euclidean field theory assigns a \emph{partition function}~$F(X)\in \CC$ to a
compact $n$-dimensional manifold~$X$ with no boundary.  There is a complex
vector space~$F(Y)$ for each compact $(n-1)$-manifold, thought of as a
spatial slice, and now this `quantum Hilbert space' may depend on~$Y$.
Roughly speaking, the vector space~$F(S^{n-1})$ attached to a small
sphere~$S^{n-1}$ is the space of local operators, and to a closed
manifold~$X$ with $k$~small open balls removed we attach the
\emph{correlation functions}
  \begin{equation}\label{eq:1}
     F(X\setminus \bigcup\limits_{i=1}^k B^n)\:F(S^{n-1})\otimes
     \cdots\otimes F(S^{n-1})\longrightarrow \CC,
  \end{equation}
as illustrated in Figure~\ref{fig:2}.  Here all boundary components are
``incoming'', whereas each interval in Figure~\ref{fig:1} has an incoming
boundary component and an outgoing boundary component, each a single point.
There are analogous quantum evolution operators in any dimension for
manifolds with both incoming and outgoing components, and the group law of
quantum mechanics has a generalization.

  \begin{figure}[ht]
  \centering
  \includegraphics[scale=.8]{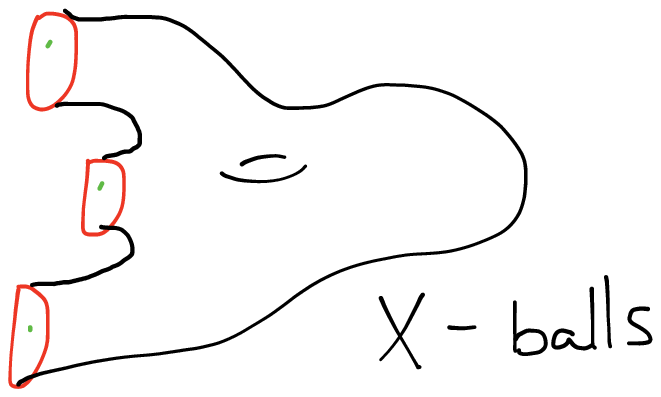}
  \caption{Correlation functions}\label{fig:2}
  \end{figure}

The mathematical expression of this formal structure is the assertion that 
  \begin{equation}\label{eq:2}
     F\:\bort\longrightarrow \Vect 
  \end{equation}
is a homomorphism, or \emph{functor}, between \emph{symmetric monoidal
categories}.  The \emph{bordism} category $\bort$ consists of
closed\footnote{A manifold is \emph{closed} if it is compact without
boundary.}  $(n-1)$-manifolds and bordisms between them.  A bordism
$X\:Y_0\to Y_1$ is a compact $n$-manifold with boundary, the boundary has a
continuous partition $p\:\partial X\to\{0,1\}$ which divides it into incoming
and outgoing components, and there are diffeomorphisms
$Y_0\xrightarrow{\;\cong \;}(\partial X)_0$ and $Y_1\xrightarrow{\;\cong
\;}(\partial X)_1$; see Figure~\ref{fig:3} in which one should view time as
flowing from left to right.\footnote{In that figure the diffeomorphisms map
open collar neighborhoods.  This guarantees that the composition, or gluing,
operation on bordisms yields smooth manifolds.}  The target category has
complex vector spaces as objects and linear maps as morphisms.\footnote{One
should use \emph{topological} vector spaces and \emph{continuous} linear
maps.  The topology on the vector spaces is not relevant for topological
field theories since the vector spaces are finite dimensional so have a
unique linear topology.}
  \begin{figure}[ht]
  \centering
  \includegraphics[scale=.5]{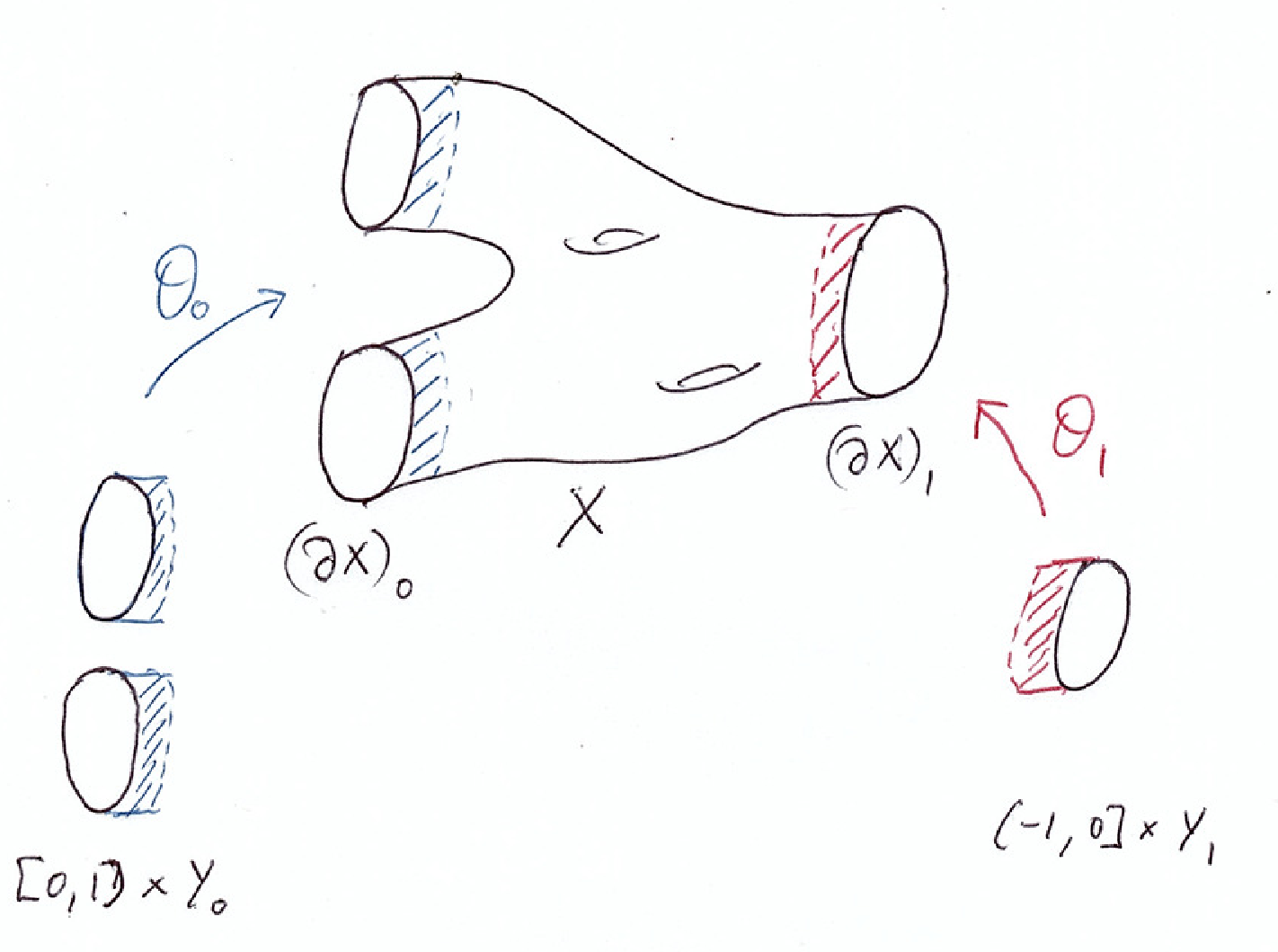}
  \caption{A bordism $X\:Y_0\to Y_1$}\label{fig:3}
  \end{figure}

\noindent 
 There are two kinds of composition.  The internal composition glues
morphisms (Figure~\ref{fig:4}) and the external composition is disjoint
union.  Similarly $\Vect$ ~has two composition laws: the internal composition
is the usual composition of linear maps and the external composition is
tensor product.  The homomorphism~\eqref{eq:2} is required to preserve both
composition laws.  This formulation is rather compact and one must unpack it
to see the usual structures in field theory.

  \begin{figure}[ht]
  \centering
  \includegraphics[scale=.5]{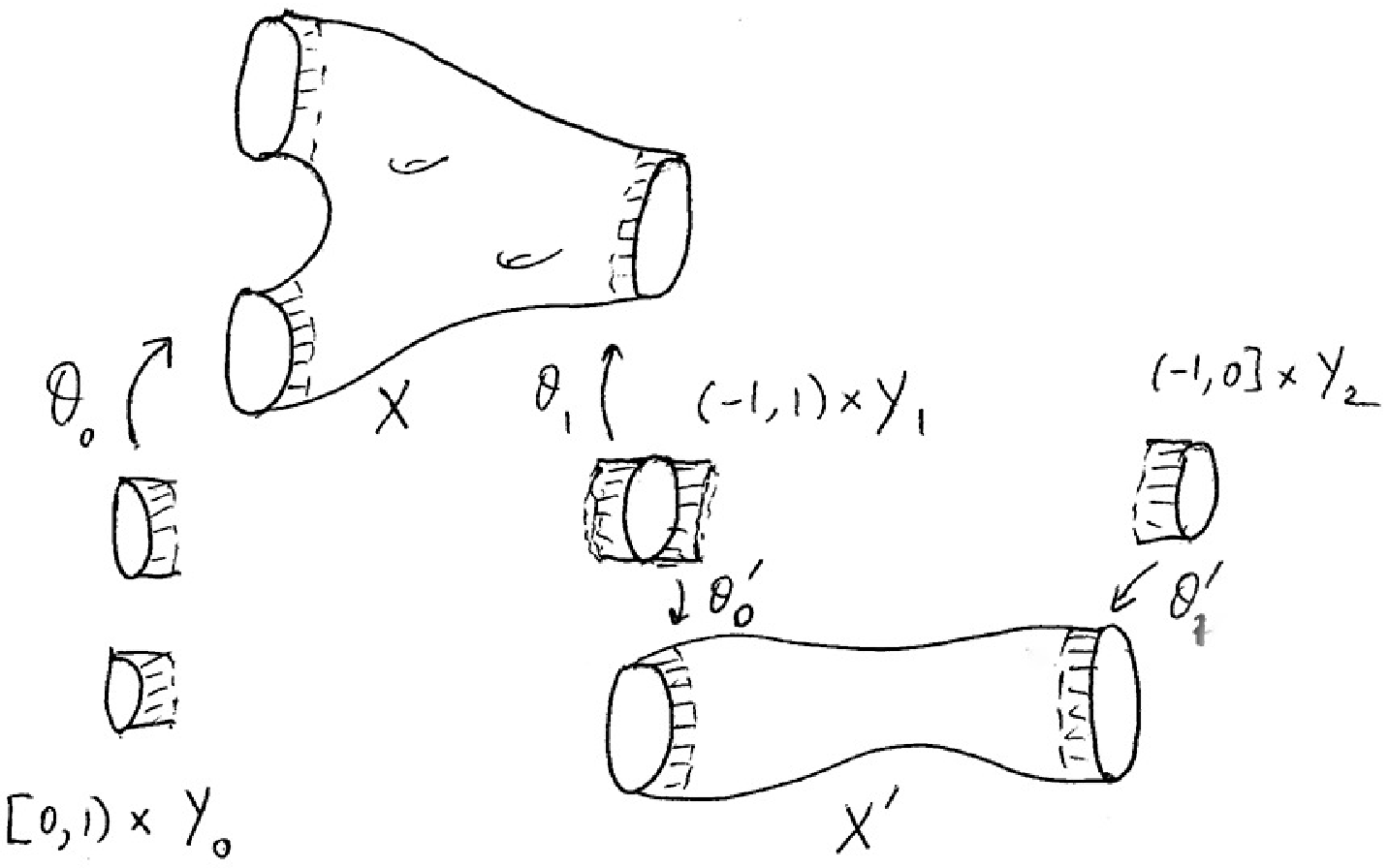}
  \caption{Composition~$X'\circ X$ of morphisms $X\:Y_0\to Y_1$ and
  $X'\:Y_1\to Y_2$}\label{fig:4}
  \end{figure}

Typically one does not have bare $n$-manifolds, but rather each
$n$-manifold~$X$ is endowed with a space\footnote{Some fields, such as gauge
fields, have internal symmetries, and they form a \emph{stack} rather than a
space.}~$\sF(X)$ of fields.  For example, in quantum mechanics
(Figure~\ref{fig:1}) the 1-manifolds have a Riemannian metric: the total
length represents time.  Higher dimensional field theories are often
formulated on Riemannian manifolds, though conformal field theories only
require a conformal structure.  Other possible fields include orientations,
spin structures, scalar fields, spinor fields, etc.  There is a bordism
category~$\bort(\sF)$ of manifolds equipped with a specified collection of
fields, and a functor
  \begin{equation}\label{eq:3}
     F\:\bort(\sF)\longrightarrow \Vect 
  \end{equation}
represents a field theory with $\sF$~as the set of (background) fields. 

  \begin{example}[]\label{thm:8}
 To illustrate the notation, consider $n=3$~spacetime dimensional
Chern-Simons theory with gauge group~$\TT=U(1)$.  There is a classical and a
quantum theory.  The fields~$\sF(X)$ in the classical theory on a
3-manifold~$X$ consist of an orientation~$o$ and a principal $\TT$-bundle
with connection~$A $.  A closed 3-manifold~$X$ appears in the bordism
category~$\Bord_{\langle 2,3 \rangle}$ as a morphism $X\:\emptyset
^2\to\emptyset ^2$ from the empty 2-manifold to itself.  We have $F(\emptyset
^2)=\CC$ and so $F(X)\:\CC\to\CC$ is multiplication by a complex number, and
if we make the fields explicit we denote it as $F(X;o,A)$.  It is given by
the formula
  \begin{equation}\label{eq:10}
     F(X;o,A) = \exp\left(-\frac{i}{2\pi }\int_{X,o}A\wedge dA\right) 
  \end{equation}
where implicitly, since $A$~is a 1-form on the total space of a $\TT$-bundle,
we have used a global section to pull it down to the base~$X$.\footnote{If
a section does not exist, there is a more complicated definition.}  The
orientation is used to define integration of differential forms.  An oriented
2-manifold~$Y$ with $\TT$-connection has an attached Chern-Simons line
$F(Y;o,A)$, and a bordism with fields has a relative Chern-Simons invariant
mapping between the Chern-Simons lines of the boundaries.  This theory is
\emph{invertible} in the sense described in~\S\ref{subsec:2.2}.
 
In the quantum theory~\cite{W1} one integrates over the gauge field~$A$, and
to get a well-defined 3-dimensional theory one needs in addition to the
orientation~$o$ a field~$f$ which is a certain sort of ``framing'' called a
\emph{$p_1$-structure}.  The quantum invariant $F(X;o,f)\in \CC$ of a closed
3-manifold is an arbitrary complex number---for example, it vanishes for some
manifolds---and the quantum vector spaces~$F(Y;o,f)$ do not necessarily have
dimension one.  So the quantum theory is not generally invertible.  
  \end{example}

  \begin{remark}[]\label{thm:1}
 This description of a field theory has an important deficiency: it does not
encode \emph{smooth} dependence on parameters.  For example, the partition
function on a closed $n$-manifold~$X$ must depend smoothly on the background
fields in~$\sF(X)$.  (Example: the smooth dependence of the partition
function of 2-dimensional Yang-Mills theory on the area of a surface.)
Formally, the definitions are enlarged to include fiber bundles $\sX\to S$ of
$n$-manifolds with fields, and these must map to smoothly varying linear maps
of smooth vector bundles over~$S$.  For topological theories, which are our
main concern, instead of families one usually postulates instead that the
morphism sets in the categories~$\bort(\sF)$ and~$\Vect$ have a topology
and all maps are \emph{continuous}.  There are two obvious topologies on the
set of linear maps~$\Hom(V_0,V_1)$ between two finite dimensional vector
spaces, equivalently on the set of $n\times m$ matrices.  We can use the
usual topology induced from the usual topology on the real numbers, or we can
use the discrete topology.  Both are used in the classification scheme
of~\S\ref{subsec:5.2}.  
  \end{remark}

  \begin{remark}[]\label{thm:2}
 One should also allow smooth families of field theories~$F$.  For
topological field theories it is more natural to study \emph{continuous}
families, i.e., to form a \emph{topological space} of topological field
theories.  This is crucial for the classification of gapped topological
phases: if two gapped systems are connected by a continuous path, we expect
their effective topological field theories can also be joined by a continuous
path.
  \end{remark}

Finally, in this paper we will always consider \emph{fully extended} field
theories, usually topological.  An $n$-dimensional fully extended theory
assigns invariants to manifolds of all dimensions~$\le n$.  These compact
manifolds with fields, which are now allowed corners, are organized into an
algebraic structure called a \emph{symmetric monoidal $(\infty
,n)$-category}, denoted~$\Bord_n(\sF)$.  The target for an extended field
theory is a symmetric monoidal $(\infty ,n)$-category~$\sC$, which typically
has its ``$(n-1)^{\textnormal{st}}$ loop space'' isomorphic to~$\Vect$.  If
so, then an extended field theory
  \begin{equation}\label{eq:5}
     F\:\Bord_n(\sF)\longrightarrow \sC
  \end{equation}
restricts on $(n-1)$- and $n$-manifolds to a usual field theory~\eqref{eq:3}.
As already indicated, the invariants attached to manifolds of dimension~$\le
n-2$ tend to be categorical in nature.  A theory which extends in this way is
fully local, and it is natural to make this strong locality hypothesis for
the effective topological theory which comes from a gapped physical theory.
See~\cite{L} for a modern description of fully extended topological field
theories.

  \subsection{Invertible field theories}\label{subsec:2.2}

`Invertibility' refers to the tensor product operation on vector spaces and
linear maps.  First, $\CC$~is a ``unit element'' for tensor product in the
sense that for any vector space~$V$ we have an isomorphism $\CC\otimes
V\xrightarrow{\;\cong \;}V$ which is naturally defined.  Thus we call~$\CC$ a
\emph{tensor unit}.  A complex vector space~$V$ is \emph{invertible} if there
exists a vector space~$V'$ and an isomorphism $V\otimes
V'\xrightarrow{\;\cong \;}\CC$.  Since $\dim (V\otimes V')=(\dim V)(\dim
V')$, it follows immediately that if $V$~is invertible, then $\dim V=1$.
Conversely, if $V$~is 1-dimensional then $V\otimes V^*$ is isomorphic
to~$\CC$.  Thus the invertible vector spaces are precisely the \emph{lines}.

Invertibility of a linear map $T\:V_0\to V_1$ under tensor product is
equivalent to the usual definition of invertibility, namely that there exist
$S\:V_1\to V_0$ such that the compositions $S\circ T$ and~$T\circ S$ are
identity maps.  If $L$~is a line, then a linear map $\lambda \:L\to L$ is
multiplication by a complex number~$\lambda \in \CC$ and it is invertible
if and only if~$\lambda \not= 0$.  We denote the nonzero complex numbers
as~$\Cx$. 
 
Invertible complex vector spaces and invertible linear maps between them form
a subcategory $\Line\subset \Vect$ which by definition has the property that
all morphisms are invertible.  Such a category is called a \emph{groupoid}.
 
An \emph{invertible field theory} $F\:\BF\to\Vect$ is one for which all
vector spaces~$F(Y)$ and linear maps $F(X)$ are
invertible.\footnote{\label{foot}For a fully extended invertible field theory
the value of ~$F$ on \emph{any} manifold of dimension~$\le n$ is invertible
under the symmetric monoidal product of the target.  A theorem of the author
and Constantin Teleman asserts that for oriented theories if the
number~$F(S^n)$ is nonzero and the vector spaces~$F(S^p\times S^{n-1-p})$ are
one-dimensional, then $F$~is invertible.}  We can express that by saying that
$F$~factors through a functor $\BF\longrightarrow \Line$.  Thus all quantum
Hilbert spaces are one-dimensional and all propagations are invertible.  The
tensor product of invertible theories is invertible, so (isomorphism classes
of or deformation classes of) invertible theories form an abelian group.

  \begin{example}[]\label{thm:3}
 Here is a simple example with~$n=1$.  Fix a smooth manifold~$M$ of any
dimension and a smooth complex line bundle $L\to M$ with connection.  We
define a 1-dimensional field theory whose set of fields~$\sF(X)$ on a
1-manifold~$X$ consists of a pair~$(o ,\phi )$ of an orientation~$o $ and a
smooth map $\phi \:X\to M$.  Then to $Y=\pt_+$ a point with the positive
orientation and $\phi (\pt)=m\in M$ we set $F(Y)=L_m$ to be the fiber of the
line bundle $L\to M$ at~$m$.  To $X=[0,1]$ with the usual orientation and a
map $\phi \:[0,1]\to M$ we assign the parallel transport $F(X)\:L_{\phi
(0)}\to L_{\phi (1)}$ along the path~$\phi $.  The reader can easily work out
the values of~$F$ on other manifolds.  We can make a family of such field
theories by varying the line bundle $L\to M$ and its connection.  The path
components of this family of field theories are parametrized by the
topological equivalence classes of line bundles $L\to M$ (without
connection.)  In~\S\ref{sec:4} and~\S\ref{sec:5} we will learn that the set
of path components can be computed by stable homotopy theory.\footnote{The
computation: $[\Sigma ^1MTSO_1\wedge M_+,\Sigma ^2H\ZZ]\cong H^2(M;\ZZ)$.  In
all computations $[X,Y]$ denotes \emph{pointed} homotopy classes of maps
between the pointed spaces~$X,Y$.}
  \end{example}

  \begin{example}[]\label{thm:6}
 Continuing with~$n=1$ we now take the line bundle to be one of the fields,
rather than being pulled back from an external manifold.  Thus let $\sF(X)$
consist of an orientation~$o$ and a complex line bundle $L\to X$ \emph{with
connection}.  The definition of the theory is similar to that in
Example~\ref{thm:3}.  (The two theories are related: choose the universal
line bundle $L\to \CP^{\infty}$ in Example~\ref{thm:3}.)  We continue this
example in~\S\ref{subsec:2.7}.
  \end{example}

  \begin{example}[]\label{thm:4}
 Let~$n=2$ and suppose $\sF$~includes just an orientation.  The line~$F(Y)$
attached to any closed 1-manifold~$Y$ is the trivial line~$\CC$ and the
number attached to any closed 2-manifold~$X$ is 
  \begin{equation}\label{eq:18}
     \lambda ^{\Euler(X)}, 
  \end{equation}
the exponential of the Euler number with base some~$\lambda \in \Cx$.  This
is a connected\footnote{As agrees with the computation $[\Sigma
^2MTSO_2,\Sigma ^3H\ZZ]\cong H^3(BSO_2;\ZZ)=0$.} family of theories with
parameter space\footnote{The computation of the parameter space: $[\Sigma
^2MTSO_2,\Sigma ^2H\CZ]\cong H^2(BSO_2;\CZ)\cong \CZ$.  Aficionados may
relish the following.  If we compose with $\Sigma ^2H\CZ\to \Sigma ^2I\CZ$,
then the theories with parameter~$\lambda \in \Cx$ and~$-\lambda \in \Cx$
become isomorphic.  Here $I\CZ$~is the Brown-Comenetz dual of the sphere
spectrum~(\S\ref{subsubsec:5.1.1}).  Note the numerical invariants of
2-manifolds only depend on~$\lambda ^2$.}~$\Cx$.  If we drop the
orientation, then there is another theory not connected to this
family:\footnote{\label{foot:12}We compute: $[\Sigma ^2MTO_2,\Sigma
^3H\ZZ]\cong H^3(BO_2;\widetilde{\ZZ})\cong \zt$.  Here $\widetilde{\ZZ}$~is
the nontrivial local system on~$BO_2$.}  the invariant of a closed
surface~$X$ is $(-1)^{w_1^2(X)}$, where $w_1^2(X)$~is the characteristic
number associated to the square of the first Stiefel-Whitney class of the
tangent bundle.  (It is nontrivial for the real projective plane, for
example.)
  \end{example}

Classical lagrangian field theories are invertible field theories: the
invariant of a closed $n$-manifold is the \emph{exponentiated} action
$e^{iS(X)}$.  Here `$X$'~includes a choice of background fields. 

  \begin{example}[]\label{thm:5}
 Finite gauge theories provide a typical example of a classical topological
theory in any dimension~$n$~\cite{DW,FQ}.  Let $G$~be a finite group and fix
a cocycle which represents a cohomology class~$\lambda \in H^n(BG;\RZ)$.  The
fields~$\sF(X)$ are an orientation and a principal $G$-bundle $P\to X$.
(As principal $G$-bundles have automorphisms---deck transformations---the
fields in this case form a groupoid, or stack, rather than a space;
see~\cite{FH} for one mathematical treatment.)  Let $\lambda (P)\in \RZ$ be the
pairing of the characteristic class~$\lambda $ of~$P$ in~$H^n(X;\RZ)$ with
the fundamental class of the orientation.  The invariant of the invertible
field theory is~$e^{2\pi i\lambda (P)}$.  Note that no orientation is
required if the cocycle vanishes.  This \emph{classical} Dijkgraaf-Witten
theory is invertible; the \emph{quantum} theory, obtained by a finite sum
over bundles~$P$, is typically not invertible.  The case~$n=3$ is a special
case of Chern-Simons theory (briefly described in Example~\ref{thm:8} with
gauge group~$\TT$).
  \end{example}

  \subsection{Relative field theories and anomalies}\label{subsec:2.3}

A field theory~$F$ as described in~\ref{subsec:2.1} might be termed
\emph{absolute}.  Suppose $\alpha $~is an (absolute) $(n+1)$-dimensional
field theory.  Then we can have an $n$-dimensional theory~$F$ which is
defined \emph{relative} to~$\alpha $.  In fact, it is more precise and
important to realize that we need only the truncation~$\ta$ of~$\alpha $
which remembers the values on manifolds of dimension~$\le n$.  Thus $\alpha
$~need only be defined on such manifolds in the first place.  To a closed
$n$-manifold~$X$ (with fields) the theory~$\alpha $ assigns a vector
space~$\alpha (X)$.  A relative theory~$F$ then assigns either a linear map
  \begin{equation}\label{eq:6}
     F(X)\:\CC\longrightarrow \alpha (X) 
  \end{equation}
or a linear map 
  \begin{equation}\label{eq:7}
     F(X)\:\alpha (X)\longrightarrow \CC. 
  \end{equation}
In the first case we evaluate the map on~$1\in \CC$ to obtain a vector
in~$\alpha (X)$; in the second case $F(X)$~is a covector, an element of the
dual vector space.  There are similar statements for lower dimensional
manifolds.  In the first case we write 
  \begin{equation}\label{eq:8}
     F\:\bo\longrightarrow \ta 
  \end{equation}
and in the second 
  \begin{equation}\label{eq:9}
     F\:\ta\longrightarrow \bo,
  \end{equation}
where $\bo$~is the trivial theory.   
 
If $\alpha $~is invertible, then $F$~is termed an \emph{anomalous} field
theory with \emph{anomaly} theory~$\alpha $. 
 
We refer to~\cite{FT} for more explanations and for nontrivial examples.
Here is an easy one, which illustrates the relationship with boundary
conditions\footnote{Rather than a boundary `condition', a relative theory can
be viewed as a boundary `theory'.}.  Relative theories can be viewed as
boundary conditions in any dimension, an idea we take up in~\S\ref{sec:7}.

  \begin{example}[]\label{thm:7}
 Let $n=0$ and suppose $\alpha $~is a quantum mechanics theory
(Figure~\ref{fig:1}) with Hilbert space~$\sH$ attached to a point.  Then a
relative theory~$F\:\bo\to \tau \mstrut _{\le0}\alpha $ is determined by
evaluating on 0-manifolds, and it is enough to evaluate on a point with each
orientation.  Since $\alpha (\pt_+)=\sH$ and $\alpha (\pt_-)=\sH^*$ we obtain
a vector~~$\Omega _F\in \sH$ and a dual vector~$\theta _F\in \sH^*$.  We can
use the relative theory~$F$ as a boundary condition for~$\alpha $, as
illustrated in Figure~\ref{fig:5}.
  \end{example}

  \begin{figure}[ht]
  \centering
  \includegraphics[scale=1]{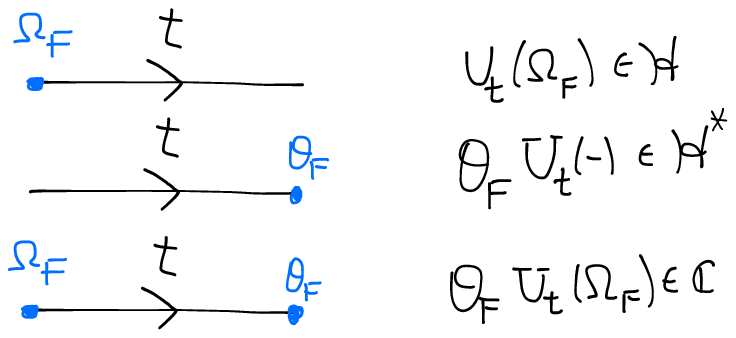}
  \caption{Quantum mechanics with boundary}\label{fig:5}
  \end{figure}

  \subsection{Global symmetries and equivariant extensions}\label{subsec:2.4}

\subsubsection{General discussion}\label{subsubsec:2.4.1}
 
Let $G$~be a Lie group.  For simplicity we discuss a non-extended field
theory~\eqref{eq:3}.  A global symmetry group may have an action on fields,
and it also may have an action on the category~$\Vect$.  In the simplest
cases those actions are trivial, and then $G$~is a \emph{global symmetry} if
the functor~\eqref{eq:3} lifts to a functor
  \begin{equation}\label{eq:11}
     F\:\bort(\sF)\longrightarrow \Rep_G 
  \end{equation}
into the category of representations of~$G$.  More plainly, the group~$G$
acts on the vector space~$F(Y)$ attached to each $(n-1)$-manifold and the
linear maps assigned to bordisms are $G$-invariant. 
 
In this situation we might ``gauge the symmetry'' or, in less ambiguous
terms, construct a \emph{$G$-equivariant extension} of the theory.  This
means that there is a new field which is a $G$-connection; if $G$~is finite,
then a $G$-connection is simply the underlying principal $G$-bundle.  Whence
the terminology: this is the gauge field.  There is a new set of fields
$\tsF$ which maps to the single field\footnote{$\BNG$ is most naturally a
simplicial sheaf on the category of smooth manifolds~\cite{FH}.}~$\BNG$ of a
$G$-connection, and the fiber over the trivial $G$-connection is the old
set~$\sF$ of fields.  The $G$-equivariant extension is a functor
  \begin{equation}\label{eq:12}
     \tF\:\bort(\tsF)\longrightarrow \Vect 
  \end{equation}
whose restriction to the trivial $G$-connection is the original
theory~\eqref{eq:11}.  This makes sense since the trivial $G$-connection has
the group~$G$ as its automorphism group.   

  \begin{example}[]\label{thm:9}
 A typical example in field theory is a $\sigma $-model into a Riemannian
manifold~$M$ with a group~$G$ of isometries.  The classical model makes sense
in any spacetime dimension~$n$.  The fields~$\sF(X)$ on an $n$-manifold
consist of a metric~$g$ and a map $\phi \:X\to M$.  In the $G$-equivariant
extension an element of~$\tsF(X)$ is a triple~$(g,\Theta ,\phi)$ where
$\Theta $~is a connection on a principal $G$-bundle $P\to X$ and now $\phi
$~is a $G$-equivariant map $P\to M$.  The map to~$\BNG$ sends~$(g,\Theta
,\phi )$ to~$\Theta $.  If $\Theta $~is the trivial $G$-connection then $\phi
$~is equivalent to a map $X\to M$, and the deck transformations of the
trivial bundle $X\times G\to X$ become the original $G$-action on the
fields. 
 
This example does \emph{not} fit our simplified description, since the global
symmetry group~$G$ \emph{does} act on the fields~$\sF$, but there is a
modification which covers this situation.  In the quantum $\sigma $-model we
integrate over the field~$\phi $, and as there is no $G$-action on the
remaining fields our description applies as is.
  \end{example}

\subsubsection{Anomalies}\label{subsubsec:2.4.2}

There may be obstructions to constructing this $G$-equivariant extension:
see~\cite{KT} for a recent discussion and examples.  One well-known example
is the gauged WZW model~\cite{W5}.  In good cases there is only a single
obstruction which can be interpreted as an \emph{anomaly}.  In these cases
the extension~\eqref{eq:12} does not exist but rather there is an invertible
$(n+1)$-spacetime dimensional theory~$\alpha $ and an extension~$\tF$ which
is a theory relative to~$\alpha $ in the sense of~\S\ref{subsec:2.3}.

  \begin{example}[]\label{thm:10}
 A standard example in~$n=4$ is quantum chromodynamics.  This theory has a
global $SU_N\times SU_N$ symmetry (for $N$~the number of flavors) which is
anomalous.
  \end{example}

We account for such anomalies in our proposals (\S\ref{subsec:5.2}). 
 
\subsubsection{Gauging antilinear symmetries}\label{subsubsec:2.4.3}

In quantum mechanics, due to Wigner's theorem, the global symmetry group~$G$
is equipped with a homomorphism\footnote{Notation: $\mu _k=\{\lambda \in
\CC:\lambda ^k=1\}$ is the group of $k^{\text{th}}$~roots of unity.}
  \begin{equation}\label{eq:98}
     \phi \:G\to \pmo=\{\pm1\} 
  \end{equation}
which tracks whether a given symmetry acts linearly or antilinearly.  Because
states are lines in a Hilbert space, rather than vectors, the group~$G$ acts
projectively and there is an extension by the group~$\TT$ of unit norm
scalars.  For theories tied to spacetime, as opposed to abstract theories,
one can also track whether or not symmetries reverse the orientation of time
by another homomorphism
  \begin{equation}\label{eq:99}
     t\:G\to\pmo .
  \end{equation}
In many situations~$t=\phi $, but that needn't be so in general.
See~\cite[\S\S1,3]{FM2} for a general discussion.

We handle the extension by simply replacing~$G$ with the extended symmetry
group.  In doing so we must take care that the group of scalars acts by
scalar multiplication on all vector spaces in the theory.  We discuss
time-reversal in~\S\ref{subsubsec:4.2.3}.  Here we explain how to gauge
antilinear symmetries.
 
Consider a 1-spacetime dimensional theory~$F$, so a quantum mechanical model
as in Figure~\ref{fig:1}.  Let $F(\pt_+)=\sH$ and suppose $G$~is a group of
global symmetries as above.  For simplicity, assume $G$~is a discrete group.
As explained in~\S\ref{subsubsec:2.4.1} a $G$-equivariant extension~$\tF$ is
a theory on manifolds equipped with a principal $G$-bundle.  Now $\sH$~is the
value of~$\tF$ on~$\pt_+$ equipped with the trivial (which means
\emph{trivialized}) $G$-bundle.  Consider the 1-manifold~$X=[0,1]$ with a
(necessarily trivial, but not trivialized) $G$-bundle, and suppose there are
trivializations over the endpoints~$\{0,1\}$.  Let $g\in G$~be the parallel
transport.  Then $\tF(X,g)\:\sH\to\sH$ is the action of the global
symmetry~$g$.  However, if $\phi (g)=-1$---i.e., if $g$~acts
antilinearly---then this does not fit into~\eqref{eq:12} since the
category~$\Vect$ has only \emph{linear} maps.  The way out is that $\phi
$~defines a 2-dimensional invertible anomaly theory~$\alpha $, and $\tF$~is
an anomalous theory with anomaly~$\alpha $, as in~\eqref{eq:9}.  The anomaly
theory assigns $\alpha(X,g)=\Cb$, the complex conjugate to the trivial
line of complex numbers, and then the relative theory gives a \emph{linear}
map 
  \begin{equation}\label{eq:100}
     \tF(X,g)\:\Cb\otimes \sH\longrightarrow \sH.
  \end{equation}
(A linear map $\Cb\otimes V\to W$ is equivalent to an antilinear map $V\to
W$.  See~\cite[(2.7)]{FT} for the analog of~\eqref{eq:100} in an arbitrary
relative theory.)
 
This discussion extends to higher dimensions.  For extended field theories
with values in a higher category~$\sC$ we would need to explain how complex
conjugation acts on~$\sC$.  In this paper we focus on invertible field
theories, and we will implement this ``antilinearity anomaly'' using twisted
cohomology; see~\S\ref{subsubsec:5.1.3}.

  \subsection{Unitarity}\label{subsec:2.8}

The formal setup of topological field theory described here is based on the
\emph{Euclidean} version of quantum field theory.  For Euclidean QFTs
unitarity is expressed by both a reality condition and a
\emph{reflection-positivity} condition.  (A standard reference is~\cite{GJ};
see~\cite[p.~690]{Detal} for a heuristic explanation.)  The unitarity
condition for a fully extended topological theory~\eqref{eq:5} implements
only the reality condition.  Namely, assuming the fields~$\sF$ include an
orientation there is an involution of~$\Bord_n(\sF)$ which reverses the
orientation.  Also, assuming that the target~$\sC$ is based on complex
numbers, then it has an involution of complex conjugation.  Unitarity is the
statement that $F\:\Bord_n(\sF)\to\sC$ is equivariant for these involutions.
A formal justification from the path integral stems from a basic fact:
orientation reversal conjugates the Euclidean action.
 
In this paper we indicate how to implement unitarity for invertible
topological field theories, which are maps of spectra.  The involutions are
quite explicit, and unitarity amounts to a twisted extension of the field
theory to unoriented manifolds.  One subtlety, which also occurs in
non-invertible theories, is that for spin theories there are two notions of
unitarity---the two different Euclidean pin groups lead to two
orientation-reversing involutions on spin manifolds.\footnote{I thank Kevin
Walker for emphasizing this point.  I do not know a physical argument
which distinguishes one of them as the preferred choice, though in specific
examples often one is preferred over the other.}

Because of the gaps in our understanding of unitarity, related to positivity
and the choice of pin group, we do not implement unitarity fully in our
proposal in~\S\ref{subsec:5.2}.  We discuss unitarity further
in~\S\S\ref{subsubsec:4.2.4}, \ref{subsubsec:4.2.5}.

  \subsection{Topological field theories and the cobordism hypothesis}\label{subsec:2.5}

One can debate which theories deserve the moniker `topological'.\footnote{For
example, with our definition classical Chern-Simons theory
(Example~\ref{thm:8}) is not topological if the gauge group is not discrete.}
We will say that a fully extended theory~\eqref{eq:5} is topological if the
fields~$\sF$ are topological, and the fields are topological if they satisfy
homotopy invariance: if $f_t\:X'\to X$ is a homotopy of local diffeomorphisms
of $n$-manifolds, then the pullbacks by~$f_0$ and~$f_1$ on fields are
equal.\footnote{Some fields, such as gauge fields, have internal symmetries
and then the pullbacks are not strictly `equal' but rather are `equivalent'
or `homotopic'.}  Thus orientations, spin structures, and $G$-bundles for
discrete groups~$G$ are all examples of topological fields.  Metrics,
conformal structures, and connections for positive dimensional Lie groups are
all examples of non-topological fields.  On the other hand, \emph{flat}
$G$-connections are topological fields for any Lie group~$G$.
 
Fully extended topological field theories are a topic of current interest in
topology and other parts of mathematics.  The \emph{cobordism hypothesis},
conjectured by Baez-Dolan~\cite{BD} and proved by Hopkins-Lurie in
dimensions~$\le 2$ and in general by Lurie~\cite{L}, is a powerful result
which determines the space of fully extended topological theories of a fixed
type.  The `type' refers to both the fields~$\sF$ and the target~$\sC$.  Thus
one speaks of ``oriented'' theories or ``framed'' theories, which tells about
the topological fields in the theory.  The theorem very roughly states that a
theory~$F$ is determined by its value~$F(\pt)$ on the 0-manifold consisting
of a single point.  One can intuitively think of this as the value on an
$n$-dimensional ball, and the idea is that any $n$-manifold is glued together
from balls, so that if the theory is fully local then its values can be
reconstructed from those on a point.  Furthermore, the value on a point is
constrained to satisfy strong finiteness conditions.  We refer the reader
to~\cite{L} and the expository account~\cite{F1}.

  \begin{remark}[]\label{thm:11}
 If $n=1$ and we let the field~$\sF$ be an orientation, then a theory 
  \begin{equation}\label{eq:13}
     F\:\Bord_{\langle 0,1 \rangle}(\sF)\longrightarrow \Vect 
  \end{equation}
is determined by the vector space~$F(\pt_+)$.  The finiteness condition is
that $F(\pt_+)$~is finite dimensional.  Of course, in usual quantum mechanics
the quantum Hilbert space is typically infinite dimensional; not so in this
topological version.
  \end{remark}

The cobordism hypothesis tells not just about individual theories, but rather
about the collection of theories with fixed~$\sF$ and~$\sC$.  The first
(easy) theorem is that this collection is an ordinary \emph{space} rather
than a more abstract category.  One should think of this space as
parametrizing families of theories, as in Example~\ref{thm:3} and
Example~\ref{thm:4}.  Yet in the homotopical setting for field theories, it
is only the \emph{homotopy type} of the space of theories which is
well-defined; see~\S\ref{subsec:2.7} for more discussion.  The theorem in
particular computes the set of path components of this space.  Two theories
lie in the same path component if and only if they can be continuously
connected.  This matches well the notion of a \emph{topological phase}, and
indeed the cobordism hypothesis is a powerful tool for distinguishing
topological phases of gapped theories.  In this paper we focus on
\emph{invertible} theories, which describe SRE phases, and the cobordism
hypothesis reduces to a much easier statement, as we explain
in~\S\ref{subsec:2.6}.  We emphasize that the cobordism hypothesis---for
invertible and non-invertible theories---determines the complete homotopy
type of the space of theories, not just the set of path components.

  \subsection{Invertible topological theories and maps of spectra}\label{subsec:2.6}

We begin with an analogy.  Let $\CC[x]$~be the ring of polynomials in a
variable~$x$ with complex coefficients and let $\CC$~be the ring of complex
numbers.  Define the ring homomorphism $F\:\CC[x]\to\CC$ which sends a
polynomial~$f(x)$ to its value~$f(0)$ at~$x=0$.  Let $S\subset \CC[x]$ denote
the subset of polynomials with nonzero constant term.  Note that $S$~is
closed under multiplication: if $f_1,f_2\in S$, then $f_1f_2\in S$.  Extend
the homomorphism~$F$ to ratios of polynomials $f/g$ where~$g\in S$.  This is
for the simple reason that $g(0)\not= 0$ if~$g\in S$, so $f(0)/g(0)$~makes
sense.  We write $S\inv \CC[x]$ for the ring of such ratios: we have inverted
elements in~$S$.  This inversion construction is easy in this case since
$\CC[x]$ ~is a \emph{commutative} ring; it is trickier in the noncommutative
case and in the categorical context to which we now turn.
 
Suppose 
  \begin{equation}\label{eq:14}
     F\:\Bord_n(\sF)\longrightarrow \sC
  \end{equation}
is an \emph{invertible} topological field theory.  By definition $F$~takes
values in the subset~$\sC^\times \subset \sC$ consisting of invertibles:
invertible objects, invertible 1-morphisms, and invertible morphisms at all
levels.  Now, as in the analogy, the fact that all values are
invertible\footnote{In our analogy, only some values are invertible; here
\emph{all} are.} means that the theory factors through the symmetric monoidal
$(\infty ,n)$-category $ |\bord n(\sF)|$ obtained by adjoining inverses of
all morphisms:
  \begin{equation}\label{eq:15}
     \xymatrix{\bord n(\sF)\ar[r]^(.6)F \ar[d]&\sC\\ |\bord
     n(\sF)|\ar[r]^(.6){\widetilde F }& \sC^\times \ar[u]} 
  \end{equation}
The map $\tF$ encodes all information about the theory~$F$. 
 
The domain and codomain of~$\tF$ are each a higher category, in fact an
$\infty $-category, in which all arrows are invertible.  Such categories are
called \emph{$\infty $-groupoids}.  The basic idea is that an $\infty
$-groupoid is equivalent to a space.  This is easiest to see in the opposite
direction: from a space~$S$ we can extract an $\infty $-groupoid~$\pi
_{\le\infty }S$.  Putting aside the~ `$\infty $' for a moment, we extract an
ordinary groupoid called the \emph{fundamental groupoid}~$\pi _{\le1}S$.  Its
objects are the points of~$S$ and a morphism $x_0\to x_1$ between points
$x_0,x_1\in S$ is a continuous path~$x\:[0,1]\to S$ from~$x_0$ to~$x_1$ up to
homotopy.  This is a groupoid because paths are invertible: reverse time.
The higher groupoids track homotopies of paths, homotopies of homotopies,
etc.  The conclusion is that $\tF$~may be considered as a continuous map of
spaces.  This already brings us into the realm of topology.  But more is
true.  The domain and codomain of~$\tF$ are \emph{symmetric monoidal} $\infty
$-groupoids, which induces more structure on the corresponding spaces.
(Recall that the monoidal product on the bordism category is disjoint union
and on the category~$\sC$ it is some sort of tensor product.)  Namely, those
spaces are \emph{infinite loop spaces}.  So for each of the spaces~$S$ there
exist a sequence of pointed spaces $S_0=S, S_1, S_2, \dots $ such that the
loop space of~$S_n$ is\footnote{The Clinton question: here best to take
`is'=`is homeomorphic to'.  See~\S\ref{subsec:4.1} for further
discussion.}~$S_{n-1}$.  Furthermore, the fact that $F$~preserves the
symmetric monoidal structure---a field theory takes compositions to
compositions and disjoint unions to tensor products---implies that $\tF$~is
an infinite loop map.  Such sequences of spaces are called \emph{spectra} and
an infinite loop map gives rise to a map of spectra.

The bottom line is that the space of invertible field theories (with
specified~$\sF,\sC$) is\footnote{The space is only determined up to homotopy
equivalence, i.e., there is only a well-defined homotopy type.} a space of
maps in homotopy theory.  We are interested in the abelian group of path
components, which houses an invariant of gapped topological phases, and that
is the group of homotopy classes of maps between spectra.  Therefore, the
computation of invertible field theories starts by recognizing the spectra in
the domain and codomain, which in turn depend on the choice of~$\sF$
and~$\sC$.  The domain spectra, obtained from bordism multicategories, will
be discussed in~\S\ref{sec:4}.  We remark that to study invertible field
theories one does not need the cobordism hypothesis; the power of the latter
is for more general non-invertible topological field theories.

  \begin{remark}[]\label{thm:12}
 The set of path components here has a natural abelian group structure.  In
fact, a spectrum~$\sX$ determines a collection $\{\pi _n\sX\}_{n\in \ZZ}$ of
abelian groups, its homotopy groups.  There is additional information in the
spectrum which binds these groups together, but a first heuristic is that a
spectrum is some topological version of a $\ZZ$-graded abelian group.  
  \end{remark}

  \begin{example}[]\label{thm:13}
 We continue with Example~\ref{thm:4}.  In this case~$n=2$ and the
field~$\sF$ is an orientation.  Now the target~$\sC$ is a 2-category, and we
need to determine the spectrum corresponding to the invertibles~$\sC^\times
$.  A typical choice is to take $\sC$~to be the 2-category of complex linear
categories.  This conforms to the usual picture that a 2-dimensional field
theory assigns a complex number to a closed 2-manifold, a complex vector
space to a closed 1-manifold, and a complex linear category to a 0-manifold.
In the invertible sub 2-groupoid~$\sC^\times $ all of these are
``1-dimensional''.  This means that the category is equivalent to~$\Vect$,
the category of vector spaces; the vector spaces are isomorphic to~$\CC$; and
the numbers we encounter are nonzero, so elements of~$\Cx$.  When we make the
corresponding spectrum~$\sX$, the fact that invertible categories are all
equivalent implies\footnote{We reprise the following argument
in~\S\ref{subsubsec:5.1.1}.} $\pi _0\sX=0$.  The fact that all invertible
1-morphisms are equivalent implies $\pi _1\sX=0$.  Finally we come to~$\pi
_2\sX$, which captures the 2-morphisms~$\Cx$.  Here we get two different
answers, depending on whether we consider~$\Cx$ to have the discrete topology
or the continuous topology.  In the discrete case we have $\pi
_2\sX_{\textnormal{discrete}}\cong \Cx$.  For the ordinary topology we use
$\pi _0\Cx=0,\;\pi _1\Cx\cong \ZZ$ to deduce $\pi
_2\sX_{\textnormal{continuous}}=0,\;\pi _3\sX_{\textnormal{continuous}}\cong \ZZ$.
Higher homotopy groups of~$\sX$ vanish.  Hence each of
~$\sX_{\textnormal{discrete}}$ and~$\sX_{\textnormal{continuous}}$ has only a
single nonzero homotopy group.  Such spectra are called
\emph{Eilenberg-MacLane spectra} and are basic building blocks.  The notation
is
  \begin{equation}\label{eq:16}
     \begin{aligned} \sX_{\textnormal{discrete}}&\simeq \Sigma ^2H\Cx, \\
      \sX_{\textnormal{continuous}}&\simeq \Sigma ^3H\ZZ.\end{aligned} 
  \end{equation}
The domain spectrum, to be explained in~\S\ref{sec:4}, is denoted~$\Sigma
^2MTSO_2$.  Thus the two sets of path components are: 
  \begin{equation}\label{eq:17}
     \begin{aligned} \ [\Sigma ^2MTSO_2,\Sigma ^2H\Cx]&\cong \Cx, \\ [\Sigma
      ^2MTSO_2,\Sigma ^3H\ZZ]&=0.\end{aligned} 
  \end{equation}
Resuming Example~\ref{thm:4} we see that the first of the computations
in~\eqref{eq:17} distinguishes the Euler theories, parametrized by the
base~$\lambda \in \Cx$ of the exponential in~\eqref{eq:18}.  The discreteness
in~$\sX_{\textnormal{discrete}}$ means that theories for distinct~$\lambda $
cannot be connected by a smooth path.  In the usual topology they can, and
this explains the second computation in~\eqref{eq:17}: the space of theories
is connected.
  \end{example}

  \begin{remark}[]\label{thm:42}
 For a general invertible bosonic theory in higher than 2-spacetime
dimensions we need more nonzero homotopy groups in the target spectrum.  This
surprise is discussed in~\S\ref{subsec:5.1}.
  \end{remark}

  \begin{remark}[]\label{thm:43}
 The distinction between $\sX_{\textnormal{discrete}}$
and~$\sX_{\textnormal{continuous}}$ is important and carries through to more
elaborate target spectra.  The discrete target, based on~$\Cx$ with the
discrete topology, is where we detect individual theories.  The other target,
based on~$\Cx$ with its usual topology and manifested as $\ZZ$~shifted up one
degree, is where we detect deformation classes of theories.  In the next
section we elaborate on the meaning of spaces of maps into targets such
as~$\sX_{\textnormal{continuous}}$.
  \end{remark}

  \subsection{An illuminating example}\label{subsec:2.7}

We have alluded several times to the \emph{homotopical} setting of the
cobordism hypothesis and so too the computations of deformation classes of
invertible field theories.  It often happens that the geometric
interpretation of a homotopical computation is subtle.  There are many
examples in topology, enumerative geometry, etc.  We illustrate in our
present context with a simple example.  In~\S\ref{sec:6} we encounter a more
sophisticated example of the same type in the classification of gapped
topological phases.
 
In Example~\ref{thm:6} we discussed an invertible field theory in
$n=1$~spacetime dimensions with fields 
  \begin{equation}\label{eq:19}
     \sFg=\{(o,L,\nabla )\} 
  \end{equation}
a triple consisting of an orientation, complex line bundle, and connection.
Then there is an obvious theory
  \begin{equation}\label{eq:20}
     \Fg\:\bord1(\sFg)\longrightarrow \Line 
  \end{equation}
with values in the category of complex lines.  Namely, to a point~$\pt_+$
with positive orientation and a line bundle $L\to\pt_+$ we attach the
line~$L$.  (A line bundle over a point is a single line.)  If we reverse the
orientation of~$\pt_+$, so consider~$\pt_-$, we take the dual line; if there
are several points we form the tensor product.  If $(L,\nabla )\to[0,1]$ is a
line bundle with covariant derivative over the interval with its usual
orientation, the field theory assigns to it the parallel transport
$\Fg\:L_0\to L_1$ from the fiber over the initial endpoint to the fiber at
the terminal endpoint.  If $(L,\nabla )\to\cir$ is a line bundle with
connection over an oriented circle, then $\Fg$~assigns to it the holonomy,
which is a number in~$\Cx$.  This is a well-defined theory, and it is not
``topological'' according to our definition, since the connection~$\nabla $
is not a homotopy-invariant field.
 
We can instead take the set of topological fields
  \begin{equation}\label{eq:21}
     \sFt=\{(o,L)\} 
  \end{equation}
consisting of an orientation and a complex line bundle but no connection.
Now we ask to \emph{classify} topological field theories
  \begin{equation}\label{eq:22}
     \Ft\:\bord1(\sFt)\longrightarrow \Line 
  \end{equation}
As emphasized in Example~\ref{thm:13} it is important to specify which
topology we use on linear isomorphisms in the category~$\Line$: the discrete
or continuous topology.  Here we use the continuous topology.  The
computation of equivalence classes of theories is
  \begin{equation}\label{eq:23}
     [\Sigma ^1MTSO_1\wedge B\Cx_+,\Sigma ^2\ZZ] = [S^0\wedge B\Cx_+,\Sigma
     ^2\ZZ] = H^2(B\Cx;\ZZ) \cong \ZZ.
  \end{equation}
According to the discussions in~\S\ref{subsec:2.5} we conclude that the space
of theories~\eqref{eq:22} is not connected, but rather there is an integer
invariant which distinguishes deformation classes of theories.  There is
always a trivial theory---it sends every 0-manifold with fields to the
trivial line~$\CC$ and every closed 1-manifold with fields to the
number~$1\in \Cx$---and it is in the deformation class of theories labeled by
the integer~0 in~\eqref{eq:23}.  So we are led to ask: 

  \begin{question}[]\label{thm:14}
 What theory~\eqref{eq:22} is labeled by the integer~$k$ in~\eqref{eq:23}?
Can we construct a single example of such a theory?

  \end{question}

It is clear what to do on 0-manifolds.  For example,
  \begin{equation}\label{eq:24}
     \Ft(L\to\pt_+) = L^{\otimes k}, 
  \end{equation}
the $k^{\textnormal{th}}$~tensor power of the line~$L$.  In other words, we
observe that the truncation~$\tau \mstrut _{\le0}\Fg$ of the geometric theory
above does not use the covariant derivative, and so we take $\Ft$ on
0-manifolds to be that theory to the $k^{\textnormal{th}}$~power.
 
What do we do on 1-manifolds?  The theory~ $\Fg$ uses the connection on
$L\to[0,1]$ to define a definite linear map---parallel transport---and the
connection on $L\to\cir$ to define a definite number---the holonomy.  But the
now the fields~\eqref{eq:21} do not include a connection and we have no
apparent way to determine these linear maps and numbers.

There are several ways out, and they illuminate what is is being computed
in~\eqref{eq:23} and in later computations.  Similar remarks apply to all
topological theories and to the cobordism hypothesis.
 
The first comment, already made in Remark~\ref{thm:1}, is that a field theory
gives invariants for \emph{families} of manifolds with fields parametrized by
a smooth manifold~$S$, not just for single manifolds.  So, for example, we
can consider a family of points~$\pt_+$ parametrized by the 2-sphere~$S=S^2$
endowed with a complex line bundle $L\to S^2$ on the total space.  Let $d$~be
the degree of this line bundle.   A field theory~$\Ft$ returns another
line bundle over~$S$, and according to~\eqref{eq:24} $\Ft$~of this family is
the $k^{\textnormal{th}}$~tensor power $L^{\otimes k}\to S$, which has
degree~$kd$.  The ratio $kd/d=k$ of the degrees is the integer
in~\eqref{eq:23}, and it is detected by this 2-parameter family of points.
Similarly, we can consider a 1-parameter family of circles $S\times \cir\to
S$ parametrized by $S=\cir$.  The fields are an orientation along the fibers
of this map and a line bundle $L\to S\times \cir$ on the total space, say of
degree~$d$.  A field theory~$\Ft$ returns a map $S\to\Cx$ and what the
computation~\eqref{eq:23} tells is that the winding number of this map around
the origin in~$\CC$ is ~$kd$.  In other words, the computation~\eqref{eq:23}
determines homotopical information in any particular field theory in the
deformation class, and that information may need to be measured in families.
 
This still does not construct a particular theory~\eqref{eq:22}; it only
constrains any such.  To construct a theory we need to make some
choices---which in itself is not surprising---but what may be surprising is
that we can choose to change the domain $\bord 1(\sFt)$ or the
codomain~$\Line$ in order to construct a particular theory.  What's more, we
may go outside the realm of \emph{topological} field theories and use more
general field theories.

For instance we can replace $\sFt$ with~$\sFg$.  There is an obvious map
$\sFg\to\sFt$ which forgets the connection~$\nabla $, and the important point
is that the fibers of this map are \emph{contractible}.  That is, the space
of covariant derivatives~$\nabla $ on a fixed line bundle $L\to X$ is a
contractible space.  (In fact, it is an infinite dimensional affine space.)
If we make this substitution, then we know how to construct a theory.
Namely, we take $\Fg^{\otimes k}$ where $\Fg$~is defined after~\eqref{eq:20}.
Another choice would be to fix a smooth model of the classifying space for
line bundles, say by fixing a complex Hilbert space~$\sH$ and taking the
classifying space to be the infinite dimensional projective space~$\PP(\sH)$.
Furthermore, we fix a smooth line bundle $H\to\PP(\sH)$ with covariant
derivative.   Augment the topological fields~\eqref{eq:21} by a
contractible choice: a classifying map $\gamma \:X\to\PP(\sH)$ for each line
bundle $L\to X$.  Again there is a map from triples~$(o,L,\gamma )$
to~$(o,L)$ and the fibers are contractible.  With these choices and
replacements it is easy to construct a particular theory using parallel
transport in the bundle $H^{\otimes k}\to\PP(\sH)$, pulled back via the
classifying map~$\gamma $.   

  \begin{remark}[]\label{thm:28}
 A different possibility is to keep the domain fields as is ~\eqref{eq:21}
but change the target category~$\Line$.  We can replace it by the
\emph{2-category} of $\ZZ$-gerbes; the automorphisms of any object comprise
the category of $\ZZ$-torsors, and the automorphism group of any automorphism
is~$\ZZ$.  This is the target of a \emph{2-dimensional} field theory with
fields~\eqref{eq:21}.  This theory assigns to a line bundle $L\to X$ over a
closed oriented surface its degree, which is an integer.  The information on
lower dimensional manifolds can be used to compute this degree locally;
see~\cite{F2}.
  \end{remark}

Another insight can be gleaned from~\eqref{eq:15}, which we unwrap:
  \begin{equation}\label{eq:25}
     \bord n(\sF)\longrightarrow |\bord
     n(\sF)|\xrightarrow{\;\;\widetilde{F}\;\;} \sC^\times 
  \end{equation}
Let double vertical bars around a higher category denote its \emph{geometric
realization}, which is a topological space.  The process of geometric
realization inverts all arrows in the category, so factors through the single
bar construction.  In these terms what we calculate in~\eqref{eq:23} is the
group of homotopy classes of maps of \emph{spaces}
  \begin{equation}\label{eq:26}
     \|\bord n(\sF)\|\longrightarrow \|\sC^\times \| 
  \end{equation}
whereas a particular field theory~$\Ft$ is a map of \emph{categories}
  \begin{equation}\label{eq:27}
     \bord n(\sF)\longrightarrow \sC^\times 
  \end{equation}
The functor `geometric realization' takes a map of categories~\eqref{eq:27}
to a map of spaces~\eqref{eq:26}.  Furthermore, Grothendieck's \emph{homotopy
hypothesis} asserts that the functor of geometric realization is an
\emph{equivalence} between higher \emph{groupoids} and spaces, so we expect
to be able to invert
  \begin{equation}\label{eq:168}
     \sCx\longmapsto \|\sCx\| 
  \end{equation}
since $\sCx$~is a higher groupoid.  Such an inversion would let us pass from
a map of spaces~\eqref{eq:26} to a field theory~\eqref{eq:27}, but in
practice we cannot explicitly construct an inverse without making choices.
In our example, $\sC^\times =\Line$ is the groupoid of complex lines and its
geometric realization as a space is $|\sC^\times |\simeq \CP^{\infty}$.  We
need to identify that geometric realization with an explicit model
of~$\CP^{\infty}$ to get back to one-dimensional vector spaces.

In summary, the most relevant interpretation we can offer of the
computation~\eqref{eq:23} for this paper is our first: it computes homotopy
classes of a class of theories obtained by augmenting the topological fields
with geometric fields which constitute a contractible choice.  This resonates
with the construction of quantum Chern-Simons theory~\cite{W1}, for example,
in which a Riemannian metric is introduced to evaluate the path integral.

  \begin{remark}[]\label{thm:44}
 The complication here arises since we seek to interpret \emph{nontorsion}
elements of the group of deformation classes, computed in~\eqref{eq:23}.
\emph{Torsion} elements, such as in footnote~\ref{foot:12}, lift to maps into
$\Cx_{\textnormal{discrete}}$ shifted down a degree, and so are realized by
definite topological theories.
  \end{remark}

   \section{The long-range effective topological field theory}\label{sec:3}

In this section we begin to apply the generalities about invertible
topological field theories to short-range entangled phases.  We start
in~\S\ref{subsec:3.1} with some general remarks about scaling, energy gaps,
and effective theories.  In~\S\ref{subsec:3.3} we give general arguments
about short-range entangled (SRE) phases in gapped condensed matter systems.
In the last subsection~\S\ref{subsec:3.4} we bring in local symmetries,
including time-reversal, and in particular define {symmetry protected
topological} (SPT) phases in terms of invertible topological field theories.

  \subsection{Low energy, long time}\label{subsec:3.1}

In classical nonrelativistic physics there are three basic
\emph{dimensions}:\footnote{There are others, such as temperature and
electric current, but they do not play a role in this discussion.}
length~($L$), time~($T$), and mass~($M$).  (For mathematical discussions of
dimensions and units, see~\cite[\S2.1]{Ta}, \cite[\S2.1]{DF}.)  The
dimensions of other physical quantities can be expressed in terms of these.
For example, energy has dimension
  \begin{equation}\label{eq:28}
     [E] = \frac{ML^2}{T^2}.
  \end{equation}
Universal physical constants also have dimensions, for example Planck's
constant~$\hbar$ and the speed of light~$c$: 
  \begin{equation}\label{eq:29}
     [\hbar] = \frac{ML^2}{T},\qquad [c]=\frac LT. 
  \end{equation}
The constant~$\hbar$ is present in any quantum system, the constant~$c$ in
any relativistic system, and both in a relativistic quantum system.  These
constants allow us to convert between dimensions, often silently by assuming
units in which $\hbar=1$ and~$c=1$.  Thus in a quantum system we have 
  \begin{equation}\label{eq:30}
     [E]=\frac 1T\qquad \textnormal{(assuming~$\hbar$)}. 
  \end{equation}
This dimensional analysis suggests that the \emph{low energy} behavior of a
quantum system is reflected in its \emph{long time behavior}.  In a
relativistic quantum system we have $L=T$ using~$c$, and so
  \begin{equation}\label{eq:31}
     [E]=M=\frac 1T=\frac 1L\qquad \textnormal{(assuming~$\hbar,c$)},
  \end{equation}
thereby relating low energy to low mass and long time to large distance.
 
Consider, then, a Riemannian manifold~$M$ with Riemannian metric~$g$.  Let
$\Delta $~be the Laplace operator on differential forms, which we take as
the Hamiltonian of a nonrelativistic quantum system.\footnote{The
dimensionally correct expression is $H=(\hbar^2/m)\Delta $, where
$m=\textnormal{mass}$; instead, we set $\hbar=1$ and~$m=1$.}  The eigenvalues
of~$\Delta $ are energies, so the low energy behavior involves the low lying
eigenvalues.  Equivalently, we can consider the long time evolution, which is
$e^{it\Delta }$ for $t$~large.  Nonzero eigenvalues lead to oscillations,
which tend to cancel out if $t$~is large, and once more we are led to the low
lying spectrum.\footnote{In Riemannian geometry the heat
operator~$e^{-t\Delta }$ is more familiar and leads to the same conclusion.}
Now we encounter a fundamental dichotomy.  If 0~is an isolated point of the
spectrum, then for $t$~large the kernel of~$\Delta $ dominates.  However, if
0~is not an isolated point of the spectrum, then the low lying continuous
spectrum mixes inextricably with the kernel.  Therefore, we isolate the
kernel at long time only if there is a \emph{gap} in the spectrum above~0.
This always happens if $M$~is compact.  In that case the Hodge and de Rham
theorems combine to prove that the kernel of the Laplace operator~$\Delta $
on differential forms has \emph{topological} significance: the dimension of
the kernel of~$\Delta $ on~$\Omega ^q(M)$ is the $q^{\textnormal{th}}$~ Betti
number of $M$.

A \emph{gapped quantum system} is one in which 0~is an isolated point of the
spectrum of the Hamiltonian~$H$.  By the dimensional analysis above, this
\emph{energy gap} is equivalent to a \emph{mass gap} in a relativistic
quantum system.  One is interested in the long time (and large distance)
behavior of a system since that is what we can observe, and often that
behavior is described by an \emph{effective} system.  We use the general
term `long-range' for either `long time' or `large distance'.  For example,
quantum field theories at large distance are often well-approximated by
another quantum field theory, and the approximating theory contains more than
the vacuum state if there is no mass gap; see~\cite[\S2.5]{W2} for a
discussion.  It is sometimes said that if there is a mass gap then the low
energy theory is trivial, but there is a more nuanced truism:\footnote{I do
not know where this idea originated; one older reference
is~\cite[p.~405]{W3}.}
  \begin{equation}\label{eq:32}
     \textit{The low energy behavior of a gapped system is approximated by a
     \emph{topological} field theory.} 
  \end{equation}
For this to make sense, we need to study the theory on arbitrary spacetimes,
so couple to background gravity.  That can be done for many field theories.
For example, 4-dimensional Yang-Mills theory makes sense on arbitrary
Riemannian manifolds, and conjecturally there is a mass gap, so one should
obtain a 4-dimensional topological field in the low energy limit.  (Here
$n=4$~is the spacetime dimension.)

  \begin{remark}[]\label{thm:15} 
 The $N=1$~supersymmetric Yang-Mills theory also has a conjectural mass gap,
and presumably the low energy effective topological theory is more
interesting in that case.  
  \end{remark}

We turn to condensed matter systems, which we treat \emph{very}
heuristically.  We assume the quantum Hilbert space~$\sH$ is the tensor
product over a set~$S$ of sites of finite dimensional Hilbert spaces
  \begin{equation}\label{eq:33}
     \sH=\bigotimes\limits_{s\in S}\sH_s 
  \end{equation}
and we may assume $\sH_s=V$ is a fixed Hilbert space independent of the
site~$s$; the sites are initially arranged in a regular pattern, such as a
lattice; and the Hamiltonian~$H$ is \emph{local} in that it is a sum
  \begin{equation}\label{eq:34}
     H=\sum\limits_{s\in S}H_s ,
  \end{equation}
where $H_s=T\otimes \id$ relative to a decomposition
$\sH=\bigl(\bigotimes_{s'}\sH_{s'}\bigr)\;\otimes
\;\bigl(\bigotimes_{s''}\sH_{s''}\bigr)$ in which $s'$~runs over sites in a
small vicinity of~$s$ and the operator~$T$ is independent of~$s$.  The
sites~$S$ are located in a background \emph{space}~$Y$ of dimension~$d$.  (We
emphasize that $d$~is the dimension of space and $d+1$~the dimension of
spacetime.)  This is a nonrelativistic system, and Galilean boosts have been
broken by the sites, which don't move in time.  Thus time is completely
separate from the geometry of space.  Initially $Y$~typically lies in flat
Euclidean space, and coupling to gravity in this context means that the
system can be studied on any curved $d$-dimensional manifold~$Y$.  Fix a
compact~$Y$ and imagine that the finite set of sites~$S$ becomes more and
more dense.  With no pretense of precision we assume that the system is
gapped in that 0~is the lowest eigenvalue of~$H$ and the spectrum of~$H$ has
a fixed size gap above~0 which persists in the limit that $S$~becomes dense.
Furthermore, we assume the kernel of~$H$ stabilizes to a finite dimensional
vector space~$F(Y)$.  In this situation we would like to
assert~\eqref{eq:32}: there is an effective $d$-space dimensional topological
field theory~$F$ which approximates the low energy/long time behavior of the
system.  The vector spaces~$F(Y)$ are part of that theory.  The deformation
class of the theory~$F$ is then a \emph{topological invariant} of the
original system, much the same way that the Betti numbers are topological
invariants of the Laplace operator on a compact Riemannian manifold.

  \begin{remark}[]\label{thm:18}
 Some parts of the gluing laws of a topological field theory are clearly
going to hold.  For example, if $Y=Y'\amalg Y''$ is a disjoint union, and we
envision a system as described after Remark~\ref{thm:15}, then the quantum
Hilbert space~\eqref{eq:33} of~$Y$ is the tensor product of those on~$Y'$
and~$Y''$ and the Hamiltonian~\eqref{eq:34} decomposes accordingly.
Therefore, so too does the kernel~$F(Y)$.  One expects a more subtle
decomposition emerges from a continuum limit process.
  \end{remark}

  \begin{remark}[]\label{thm:17}
 This procedure does not obviously give numerical invariants on all compact
$(d+1)$-dimensional manifolds.  We expect invariants corresponding to time
evolution, and since the field theory is defined for compact manifolds we
must take circular time~$\cir$.  When we return to the initial time the
manifold and its fields can undergo a symmetry, so the manifold we obtain is
a fiber bundle $X^{d+1}\to\cir$.  We do expect invariants for these
\emph{mapping cylinders}.  This kind of impoverished $(d+1)$-dimensional
field theory, perhaps with variations, goes under many names: a
\emph{$(d+\epsilon )$-dimensional theory}, a theory of \emph{$H$-type},
\dots.  Many, in fact most, of the effective theories of $H$-type that we
encounter do extend to full $(d+1)$-dimensional theories, and it may be that
additional considerations in the microscopic theory can imply that property
of the long-range topological theory.  Even more is possible.  A 1-parameter
family of $(d+1)$-manifolds parametrized by~$S$ is the total space of a fiber
bundle $M^{d+2}\to\cir$, and the theory gives a map~$\cir\to\Cx$.  Its
winding number is an integer invariant associated to~$M$.  It may happen that
these integer invariants are defined for \emph{arbitrary} closed
$(d+2)$-manifolds, not just mapping cylinders.  This would impose a more
severe constraint on the low energy effective theory, as we illustrate
in~\S\ref{subsec:6.3}.   Geometrically, an $H$-type theory gives invariants
for manifolds equipped with a rank~$d$ bundle stably equivalent to the
tangent bundle.
  \end{remark}

  \begin{remark}[]\label{thm:16}
 If one starts with a gapped quantum field theory, at first glance one can
imagine its low energy behavior giving rise to a specific topological field
theory, or at least a contractible space of theories depending on some mild
choices in the approximation.  For a condensed matter system there seem to be
more choices as one must, in addition to any cutoffs, take a continuum limit.
This leads to the expectation that we obtain a space of theories---again
presumably contractible, but hopefully at least connected---and makes it more
plausible that one could encounter the phenomenon highlighted
in~\S\ref{subsec:2.7}.  Indeed, we will in ~\S\ref{subsec:6.3}.
  \end{remark}

  \subsection{Short-range entanglement hypothesis and its
consequences}\label{subsec:3.3} 

We make explicit the assumptions beyond~\eqref{eq:32} which underlie our
proposal in~\S\ref{subsec:5.2}.  The most drastic of these is 
invertibility, which is an expression of short-range entanglement;
see~(\ref{srehyp}) below.

  \begin{enumerate}[(1)]
\narrower\itemsep1em

 \item We assume that the low energy effective topological theory~$F$ is
\emph{fully extended}, in the sense discussed after Remark~\ref{thm:2}.  This
is a strong form of locality, and it seems reasonable since the continuum
limit theory is obtained from local discrete systems, as in~\eqref{eq:33}
and~\eqref{eq:34}.  To define the extended theory we must specify what sorts
of (higher categorical) invariants $F$~computes on low dimensional manifolds.
In other words, we must specify the target~$\sC$ for the field theory;
see~\eqref{eq:5}.  Since we restrict to invertible theories, we need only
specify a spectrum.  That is one of the key choices to be made;
see~\S\ref{subsec:5.1} for a full discussion.

 \item We assume that $F$~is \emph{unitary}, since microscopic condensed matter
systems are typically unitary.

 \item One important consideration which affects the choice of~$\sC$ is the
\emph{topology} on the space of theories, as already mentioned and
illustrated in Remark~\ref{thm:2}, Example~\ref{thm:4}, and
Example~\ref{thm:13}.  To the extent that the theory~$F$ has numerical
invariants of closed $(d+1)$-manifolds, they lie in~$\CC$; the invariants of
closed $d$-manifolds are $\CC$-vector spaces.  In both cases we use the usual
topology on~$\CC$ to allow the theory~$F$ to deform by continuously varying
the numerical invariants, the linear maps between vector spaces, etc.  Two
microscopic gapped systems related by a continuous deformation are considered
to define the same topological phase, and they should give rise to effective
topological field theories which are deformation equivalent.  On the other
hand, we also consider anomaly theories~$\alpha $, which can arise when
gauging an anomalous global symmetry.  (See~\S\ref{subsec:2.3} and the text
before Example~\ref{thm:10}.)  In that case we do not allow continuous
deformations of the anomaly theory: anomalous theories relative to distinct
anomaly theories should not be viewed as the same topological phase.  Thus,
for the classification of anomalies we take the \emph{discrete} topology
on~$\CC$.\footnote{I thank Xiao-Gang Wen for explaining this point to me.}
 
 \item Another important choice is of the \emph{background fields}~$\sF$ in
the low energy effective topological theory~\eqref{eq:5}.  We expect only
topological fields.  Our choice here is based on limited experience, and is
more or less a guess.  Naively one may think a lattice system on a manifold
induces a framing, but we hope there is rotational invariance which allows us
to formulate the theory on more general manifolds.  As orientations are
typically used to define local Hamiltonians~\eqref{eq:34}, we require
invariance only under rotations connected to the identity.  Therefore, for
purely bosonic systems we choose $\sF$~to include an orientation, and if the
theory includes fermions then we augment this to a spin structure.  In the
fermionic case one can think that the coupling of such a system to background
gravity involves spacetime spinor fields, whence the necessity of a spin
structure.  There is an additional field---a background principal bundle---if
we gauge a global symmetry, as we elaborate below in~\S\ref{subsec:3.4}.

 \item As already discussed in Remark~\ref{thm:17} we take $F$ to be a theory
defined in $d$~space dimensions which defines complex numbers only on special
closed $(d+1)$-manifolds.  In Remark~\ref{thm:17} we explain then that
integers are obtained for special closed $(d+2)$-manifolds.  `Special' in
both cases means the manifold is equipped with a rank~ $d$ vector bundle and
a stable isomorphism with the tangent bundle.  

 \item \label{srehyp} Finally, we assume $F$~is \emph{invertible}, which is
the hypothesis of \emph{no long-range entanglement}, referred to as
\emph{short-range entanglement}.  The macroscopic definition given in some
literature (for example~\cite{VS}) is that $\dim F(Y)=1$ for all closed
$d$-manifolds.  This \emph{is} precisely invertibility at this
level---dimension~$d$---of the field theory.  It is not too much of a stretch
to extrapolate that to invertibility at all levels.  The theorem mentioned in
footnote~\ref{foot} supports this extrapolation.
  \end{enumerate}

  \subsection{Symmetries and SPT phases}\label{subsec:3.4}

Suppose that a Lie group~$G$ acts as global symmetries of a condensed matter
system.  This can include internal symmetries which act on the local Hilbert
space~$\sH_s=V$ in~\eqref{eq:33}.  It can also include time-reversal
symmetry, since we consider time as external to space and so time-reversal
symmetries fix the points of space.  We do not, however, consider symmetries
which move points of space since we want to consider the theory on an arbitrary
$d$-manifold~$Y$.  For example, this rules out rotation and reflection
symmetries of theories on Euclidean space.
 
Just as we study the condensed matter system on arbitrary space manifolds~$Y$
to explore its low energy behavior (coupling to gravity), to explore the
effect of the global symmetry we attempt to construct an equivariant
extension, as in~\S\ref{subsec:2.4}.  Recall that this means that we augment
the set of fields to include a $G$-connection on a principal $G$-bundle.  As
mentioned after Example~\ref{thm:9} there may be obstructions to constructing
a $G$-invariant extension.  We assume that any obstruction, if it exists, can
be expressed as an anomaly theory in one higher dimension.  Passing to the
low energy approximation, we stipulate that there is a topological anomaly
theory~$\alpha $ and the original long-range effective theory~$F$ is extended
to be an anomalous theory~$\tF$ with anomaly~$\alpha $.  The topological
theories~$\alpha $ and~$\tF$ have a \emph{principal $G$-bundle} as an
additional background field---the choice of connection does not appear in the
classification of effective topological theories.  (See the discussion
in~\S\ref{subsec:2.7}.)  Also, both theories are truncated only on manifolds
of dimension~$\le d$, as in \S\ref{subsec:3.3}(4).  As for any anomaly
theory, $\alpha $ ~is invertible.  We make the hypothesis that for
short-range entanglement the $G$-equivariant long-range topological
theory~$\tF$ is also invertible.  When we come to classify anomalies
in~\S\ref{subsec:5.2} we use the considerations of~ \S\ref{subsec:3.3}(2) to
guide the choice of topology on the space of anomaly theories.
 
In summary, then, we will assume one of two cases.  Either the original
long-range effective theory~$F$ extends to a $G$-equivariant theory~$\tF$, or
there is an anomaly~$\alpha $ and there is an anomalous $G$-equivariant
extension~$\tF$.  In both cases the original set of fields~$\sF$ in~$F$ is
augmented to
  \begin{equation}\label{eq:35}
    \tsF=\sF\cup \{\textnormal{$G$-bundle}\}. 
  \end{equation}

There is an embedding $\sF\to\tsF$ which chooses the trivial $G$-bundle.
Restriction along this map allows us to study the effective theory~$\tF$
ignoring the symmetry~$G$.  This restriction is simply the theory~$F$.  A
topological phase is said to be \emph{symmetry protected} if this
restriction~$F$ is the trivial theory.  (Notice that the anomaly
theory~$\alpha $ is trivialized under this restriction since the original
long-range effective theory~$F$ is an absolute theory---it has no anomaly.)

   \section{Bordism and homotopy theory}\label{sec:4}

In this section we recall the \emph{Madsen-Tillmann} spectra which appear in
the study of invertible topological field theories.  In the notation
of~\S\ref{subsec:2.6} the simplest of these is the spectrum\footnote{The
`$\Sigma ^n$' denotes a shift, or suspension, and is present because of an
unfortunate indexing convention.  What should appear on the left hand side
of~\eqref{eq:36} is the 0-space of the Madsen-Tillmann spectrum, which is
standardly denoted~`$\Omega ^{\infty}\Sigma ^nMTO_n$', but we blur the
notation.}
  \begin{equation}\label{eq:36}
     \Sigma ^nMTO_n\simeq  |\bord n| 
  \end{equation}
which is the geometric realization of the bordism category, obtained by
inverting all morphisms.  The Madsen-Tillmann spectra were introduced
in~\cite{MT} and versions of~\eqref{eq:36}---including nontrivial
fields~$\sF$---are proved in~\cite{MaWe}, \cite{GMTW}, \cite{Ay}.  Those
theorems treat the geometric realization of a topological 1-category, whereas
the right hand side of~\eqref{eq:36} is the geometric realization of an
$(\infty ,n)$-category.  The $(\infty ,n)$ version is stated
in~\cite[\S2.5]{L} and the techniques to prove it are most likely contained
in the cited references.  We proceed to use the $(\infty ,n)$~statement as
the basis for our proposal in~\S\ref{subsec:5.2}.  We give a brief
introduction in~\S\ref{subsec:4.1}; the class notes~\cite{F3} contain much
more detail.  In~\S\ref{subsec:4.2} we note the modifications to include
symmetry and the modifications for theories of $H$-type discussed in
Remark~\ref{thm:17}.  We also explain orientation-reversal and unitarity in
this context.

  \subsection{Madsen-Tillmann spectra}\label{subsec:4.1}

Pontrjagin studied smooth maps 
  \begin{equation}\label{eq:38}
     f\:S^{n+q}\to S^q 
  \end{equation}
by choosing a regular value~$p\in S^q$ and focusing on the inverse
image~$M=f\inv (p)\subset S^{n+q}$, an $n$-dimensional closed submanifold.
It carries an additional topological structure.  Namely, for all~$m\in M$ the
differential
  \begin{equation}\label{eq:37}
      df_m\:T_mS^{n+q}/T_mM\to T_pS^q 
  \end{equation}
is an isomorphism from the normal space to~$M$ at~$m$ to a fixed vector
space, the tangent space of~$S^q$ at~$p$.  This is a \emph{framing} of the
normal bundle.  We emphasize that it is a \emph{normal}, rather than
tangential, framing.  The submanifolds for different regular values and
homotopic maps are all \emph{framed bordant} in the sense that for any
two~$M_0,M_1$ there exists a compact $(n+1)$-dimensional manifold $N\subset
[0,1]\times S^{n+q}$ with boundary $\{0\}\times M_0\;\amalg\; \{1\}\times
M_1$; the manifold~$N$ carries a normal framing which agrees with that
of~$M_0$ and~$M_1$ on the boundary.  This is the fundamental connection
between bordism and homotopy theory~\cite[\S7]{Mi}, which was elaborated and
given computational punch by Thom in his PhD thesis~\cite{Th}.  If we seek to
understand not $n$-dimensional closed submanifolds of a fixed
sphere~$S^{n+q}$ but rather $n$-dimensional abstract closed framed manifolds,
then a basic theorem of Whitney tells that every abstract manifold embeds in
some sphere, so it suffices to take $q$~ large.  We increase~$q$ by
suspending~\eqref{eq:38}; large~$q$ is realized by iterated suspension.  The
suspension of~\eqref{eq:38} is a map between the suspension of the domain and
codomain spheres, and the suspension of a sphere is a sphere of dimension one
greater.  A sequence of spaces related by suspension is a \emph{spectrum},
and Thom introduced special spectra to compute bordism groups.  The manifolds
classified by these bordism groups may carry geometric structures (framings,
orientations, spin structures, \dots ) on their \emph{stable normal} bundle.

The bordism question for invertible topological field theory differs in a
fundamental way.  A field theory has a definite spacetime dimension~$n$, and
the geometric structures live on the $n$-dimensional tangent bundle.  So
whereas Thom's theory is \emph{stable normal}, the Madsen-Tillmann spectra
which arise encode \emph{unstable\footnote{$MT$~spectra are still part of
stable homotopy theory---they \emph{are} spectra---but the dimension is
fixed, not stabilized.}  tangential} bordism.  We sketch the basic
example~$MTO_n$ and indicate the modifications~$MTSO_n$ for oriented
$n$-manifolds and $MT\Spin_n$ for spin $n$-manifolds.

Formally, a \emph{prespectrum}~$T_{\bullet }$ is a sequence~$\{T_q\}_{q\in
\ZZ^{}}$ of pointed spaces and pointed maps $s_q\:\Sigma T_q\to T_{q+1}$,
where `$\Sigma $'~denotes suspension.  It is a \emph{spectrum} if the induced
maps $t_q\:T_{q}\to\Omega T_{q+1}$ are homeomorphisms, where `$\Omega
$'~denotes based loops.  Any prespectrum has an associated spectrum, and
furthermore it suffices to define~$T_q$ for $q\ge q_0$ for some~$q_0\in
\ZZ^{}$.  The simplest example, indicated above, is the sphere prespectrum
with $T_q=S^q$.  We now construct a prespectrum whose associated spectrum
is~$MTO_n$ for a fixed~$n\in \ZZ^{>0}$.
 
Let $Gr_n(W)$ be the Grassmannian of $n$-dimensional subspaces of the real
vector space~$W$.   A point of~$[V]\in Gr_n(W)$ is an $n$-dimensional
subspace~$V\subset W$.  The Grassmannian is a smooth manifold, and there is a
tautological rank~$n$ \emph{universal subbundle}
  \begin{equation}\label{eq:47}
     S\longrightarrow  Gr_n(W) 
  \end{equation}
whose  fiber at~$[V]$  is~$V$.  Even  better, there  is a  tautological exact
sequence
  \begin{equation}\label{eq:39}
     0\longrightarrow S\longrightarrow \vtriv W\longrightarrow
     Q\longrightarrow 0 
  \end{equation}
of vector bundles; the fiber of the \emph{universal quotient bundle}~$Q\to
Gr_n(W)$ at~$V\in Gr_n(W)$ is the quotient vector space~$W/V$, and the vector
bundle $\vtriv W\to Gr_n(W)$ has constant fiber~$W$.  For any integer~$q>0$
we define the space~$T_{n+q}$ to be the \emph{Thom space} of the universal
quotient bundle
  \begin{equation}\label{eq:40}
     Q(q)\longrightarrow Gr_n(\RR^{n+q}). 
  \end{equation}
The Thom space is obtained from the manifold~$Q(q)$ by introducing an inner
product on the vector bundle~\eqref{eq:40} and collapsing the subspace of all
vectors of norm~$\ge R$ to a point, where $R>0$ is any real number, as
indicated in Figure~\ref{fig:6}.  The structure map~$s_{n+q}$ is obtained by
applying the Thom space construction to the map
  \begin{equation}\label{eq:41}
 \begin{gathered}
    \xymatrix{\underline{\RR}\oplus Q(q)\ar[r]^{} \ar[d]_{} & Q(q+1)\ar[d]^{}
    \\ Gr_n(\RR^{n+q})\ar[r]^{} & Gr_n(\RR^{n+q+1})} 
 \end{gathered} 
  \end{equation}
The bottom arrow takes a subspace~$V\subset \RR^{n+q}$ and regards it as a
subspace of~$\RR^{n+q+1}$ all of whose vectors have first coordinate zero.

  \begin{figure}[ht]
  \centering
  \includegraphics[scale=.8]{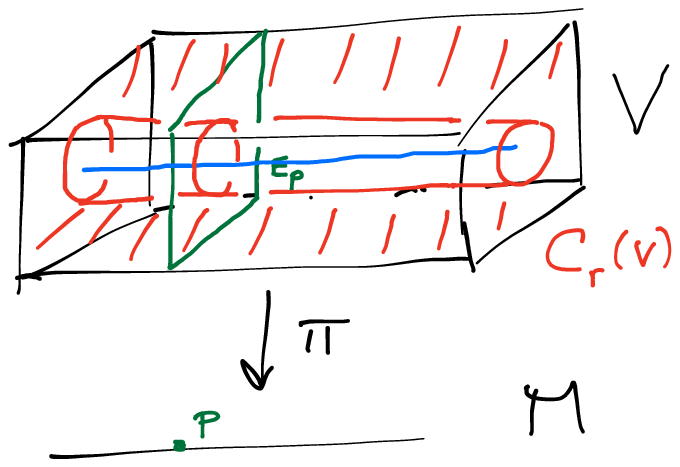}
  \caption{The Thom space of a vector bundle $V\to M$}\label{fig:6}
  \end{figure}
 
More useful to us is the $n^{\textnormal{th}}$~suspension $\Sigma ^nMTO_n$,
which is represented by the shift of this prespectrum with $(\Sigma
^nT)_{q}=T_{n+q}$.  A map $S^m\to\Sigma ^nMTO_n$ is represented by a pointed
map $S^{m+q}\to T_{n+q}$ for some~$q$ large, so a map $f\:S^{m+q}\to Q(q)$
which sends the basepoint of~$S^{m+q}$ to a vector of norm~$\ge R$.  Suppose
the map is transverse to the zero section~$Z\subset Q(q)$.  Then the pullback
$M:=f\inv (Z)\subset S^{m+q}$ is a submanifold of dimension~$m$.  Its normal
bundle is identified via~$df$ with the normal bundle to~$Z$ in~$Q(q)$, which
in turn is identified with the vector bundle~\eqref{eq:40}.  But because of
the degrees it is more natural to consider the virtual
bundle~$Q(q)-\vtriv{\RR^{n+q}}\approx -S(n)$ whose pullback is then
identified with the negative of the \emph{tangent} bundle to~$M$, stabilized
to have rank~$n$.  (Recall that $S(n)\to Gr_n(\RR^{n+q})$ is the universal
subbundle~\eqref{eq:47}.)  The choice of embedding into a sphere~$S^{m+q}$
disappears in the limit~$q\to\infty $ and, as in Thom's bordism theory, we
obtain abstract manifolds rather than embedded ones.

  \begin{remark}[]\label{thm:53}
 The pullback of $S(n)\to Gr_n(\RR^{n+q})$ to~$M$ is a rank~$n$ bundle which
is equipped with a stable isomorphism to the tangent bundle~$TM$. 
  \end{remark}

More generally, for any smooth manifold~$S$, a map $S\times S^m\to \Sigma
^nMTO_n$ leads to a smooth proper map $\pi \:N\to S$ with $\dim N-\dim S=m$
and with a rank~$n$ vector bundle~$V\to N$ equipped with a stable isomorphism
with $TN-\pi ^*TS$.  The Galatius-Madsen-Tillmann-Weiss theorem~\cite{GMTW}
asserts that the spectrum $\Sigma ^nMTO_n$ classifies a bordism theory of
proper \emph{fiber bundles}, rather than arbitrary proper maps~$\pi $.

  \begin{remark}[]\label{thm:19}
 There is a sequence of maps 
  \begin{equation}\label{eq:44}
     \Sigma ^1MTO_1\longrightarrow \Sigma ^2MTO_2\longrightarrow \Sigma
     ^3MTO_3\longrightarrow \cdots 
  \end{equation}
whose ``limit'' is the Thom spectrum~$MO$ which classifies unoriented
manifolds.  In this precise sense the Madsen-Tillmann spectra approximate
Thom spectra.
  \end{remark}

To construct the Madsen-Tillmann spectra ~$MTSO_n$ ($MT\Spin_n$) use the
Grassmannian of oriented (spin) subspaces. 
 
Quite generally, a map $S\to T$ of spectra represents a $T$-cohomology class
on~$S$.  The $T$-cohomology of a Madsen-Tillmann spectrum $S=\Sigma ^nMTSO_n$
is isomorphic to the $T$-cohomology of the space~$BSO_n$ by the \emph{Thom
isomorphism}.  For example, if $T=\Sigma ^qH\ZZ$ then there is a Thom
isomorphism 
  \begin{equation}\label{eq:134}
     [\Sigma ^nMTSO_n,\Sigma ^qH\ZZ]\cong H^q(BSO_n;\ZZ). 
  \end{equation}
There is an orientation condition which is satisfied here because we consider
the Madsen-Tillmann spectrum with group~$SO_n$.  If instead we use~$O_n$,
then we obtain cohomology with twisted coefficients: 
  \begin{equation}\label{eq:135}
     [\Sigma ^nMTO_n,\Sigma ^qH\ZZ]\cong H^q(BO_n;\ZZ^w), 
  \end{equation}
where $\ZZ^w\to BO_n$ is the nontrivial local system with holonomy~$-1$
around the nontrivial loop in~$BO_n$.  The general form of the twisted Thom
isomorphism is 
  \begin{equation}\label{eq:136}
     [\Sigma ^nMTO_n,\Sigma ^qT]\cong T^{\tau +q}(BO_n), 
  \end{equation}
where $\tau $~is the twisting of $T$-cohomology defined by the virtual vector
bundle $\vtriv {\RR^n}-S(n)\to BO_n$.  There are similar statements for the
groups~$SO_n$ and~$\Spin_n$.  We refer to~\cite{ABGHR} and references therein
for a modern treatment of twisted cohomology, orientations, and the Thom
isomorphism.

  \subsection{Variations}\label{subsec:4.2}

 \subsubsection{Global symmetry groups}\label{subsubsec:4.2.1}
 Let $G$~be a Lie group.  As explained in~\S\ref{subsec:3.4} the long-range
effective topological theory approximating a condensed matter system with
global symmetry group~$G$ includes a $G$-bundle among its background fields;
see~\eqref{eq:35}.  Such a bundle can be obtained from a universal $G$-bundle
$EG\to BG$ by pullback.  To incorporate the bundle into the Madsen-Tillmann
spectrum~$MTO_n$ replace~\eqref{eq:40} with the pullback bundle
  \begin{equation}\label{eq:42}
     Q(q)\longrightarrow Gr_n(\RR^{n+q})\times BG 
  \end{equation}
over the Cartesian product.  A map $f\:S^{m+q}\to Q(q)$ which is transverse
to the zero section of~\eqref{eq:42} gives a manifold~$M\subset S^{m+q}$ and
a $G$-bundle $P\to S^{m+q}$.  The structure maps defined from~\eqref{eq:41}
extend to incorporate the $BG$~factor.  The spectrum so obtained is denoted
  \begin{equation}\label{eq:43}
     MTO_n\wedge BG_+ 
  \end{equation}
The wedge, pronounced ``smash'', is the appropriate product for pointed
spaces; the~`$+$' denotes a disjoint basepoint.\footnote{The classifying
space~$BG$ does have a basepoint which represents the trivial $G$-bundle, and
we use it at the end of~\S\ref{subsec:5.2} to implement the embedding
described after~\eqref{eq:35} in the discussion of SPT phases.  The basepoint
of~$BG$ is ignored in the construction of~$BG_+=BG\amalg \pt$, the disjoint
union of~$BG$ and a point.}  Of course, there are similar spectra
to~\eqref{eq:43} for oriented and spin manifolds.

 \subsubsection{Theories of $H$-type}\label{subsubsec:4.2.2}
 In a $d$-space dimensional theory we obtain numerical invariants only for
$(d+1)$-manifolds which are essentially products: time is not mixed with
space; see Remark~\ref{thm:17}.  In terms of~\eqref{eq:5} an ``$H$-type''
theory is defined on the subcategory of~$\bord{d+1}(\sF)$ which only contains
top dimensional bordisms which are fibered over 1-manifolds:
  \begin{equation}\label{eq:45}
     F\:\Bord_d(\sF)\longrightarrow \sC 
  \end{equation}
An invertible theory of that form, if $\sF$~is empty, is a map out of the
spectrum~ $\Sigma ^dMTO_d$.  In terms of the explicit prespectrum described
in~\S\ref{subsec:4.1}, the $q^{th}$-space of~$\Sigma ^dMTO_d$ is the Thom
space of 
  \begin{equation}\label{eq:50}
     Q(q)\longrightarrow  Gr_d(\RR^{d+q}). 
  \end{equation}
We  explained in the  paragraph after~\eqref{eq:41}  that manifolds  in the
usual $(d+1)$-spacetime dimensional theory have tangent bundles stabilized to
rank~$d+1$.  Here in the space theory they are stabilized to rank~$d$.  So to
get  a  bundle  of  rank~$d+1$  we  simply add  an  extra  rank  one  trivial
bundle~$\underline{\RR}$;  the  fiber  represent  time, which  is  visibly  a
product and does not mix with space.
 
 \subsubsection{Time-reversal symmetries}\label{subsubsec:4.2.3}

The general picture of symmetries in quantum mechanics
(\S\ref{subsubsec:2.4.3}) distinguishes time-preserving
vs.~time-reversing~\eqref{eq:99} from linear vs.~antilinear~\eqref{eq:98}.
Often, of course, they are equal dichotomies.  In that sense an accounting of
antilinearity suffices to account for time-reversal.  But more to the point,
as we consider theories of $H$-type (\S\ref{subsubsec:4.2.2}) in which there
is no explicit time, there is no need to track~\eqref{eq:99} other than
through antilinearity.

 \subsubsection{Orientation-reversal and unitarity}\label{subsubsec:4.2.4} 
 
This is to implement the unitarity~(\S\ref{subsec:2.8}).  The geometric
realization of the \emph{oriented} bordism category~$\bord n(\sF)$ is~$\Sigma
^nMTSO_n$, if $\sF$~consists of an orientation.  Thus orientation-reversal
induces an involution on~$\Sigma ^nMTSO_n$.  Introducing the notation $M^V$
for the Thom spectrum of the virtual vector bundle $V\to M$, we can summarize
~\eqref{eq:41} as 
  \begin{equation}\label{eq:129}
     \Sigma ^nMTSO_n = BSO_n{}_{\phantom{M}}^{\vtriv{\RR^n}-S(n)}, 
  \end{equation}
where $S(n)\to BSO_n$ is the universal rank~$n$ bundle.  Orientation-reversal
on rank~$n$ bundles is represented as the deck transformation of the double
cover 
  \begin{equation}\label{eq:130}
     BSO_n\longrightarrow  BO_n. 
  \end{equation}
There is an induced involution on the Thom spectrum~\eqref{eq:129}.   

In a unitary theory orientation-reversal maps to complex conjugation.
Complex conjugation on~$\Cx$ corresponds to the involution
  \begin{equation}\label{eq:132}
     z\longmapsto -\bar z 
  \end{equation}
on~$\CZ$ via exponentiation.  Let $H\Cx$ be the target of an invertible
unitary theory $\Sigma ^nMTSO_n\to \Sigma ^nH\Cx$.  This represents a twisted
cohomology class of~$\Sigma ^nMTO_n$.  The twisted Thom isomorphism
theorem~\eqref{eq:135} identifies this twisted cohomology group with 
  \begin{equation}\label{eq:137}
      H^n(BO_n;\widetilde{\CZ}), 
  \end{equation}
where the coefficients are twisted by the complex conjugation action (no
sign).  Note the short exact sequence
  \begin{equation}\label{eq:133}
     0\longrightarrow \ZZ\longrightarrow \CC\longrightarrow
     \CZ\longrightarrow 0 ,
  \end{equation}
which is equivariant for the $(-1)$-action on~$\ZZ$ and the
action~\eqref{eq:132} on the other two groups.  So when we map to~$H\ZZ$
or~$I\ZZ$ we use the $(-1)$-action in place of the action~\eqref{eq:132}
on~$H\CZ$ and~$I\CZ$.  We remark that the twisted Thom isomorphism involves
more complicated twistings for~$I\CZ$.

An analogous construction works if we replace~$SO_n$ with~$\Spin_n$ and
$O_n$ with~$\Pin_n$.  However, there are two choices for~$\Pin_n$.  If we
view~ $\Pin_n$ as embedded in the real Clifford algebra~$\Cliff_n$, then the
choices depend on the sign in the relation $\gamma _i^2=\pm1$ which
defines~$\Cliff_n$; see~\cite{ABS}.  So it seems there are two distinct
notions of unitarity for spin theories.  We will simply write~`$\Pin$' and
not investigate the distinction further in this paper.

We now sketch precisely how we implement the $(-1)$-involution on spectra.

  \begin{construction}[$(-1)$-action on spectra]\label{thm:48}
 Let $H\to\RP^{\infty}$ denote the real Hopf line bundle.  The Thom spectrum
$(\RP^{\infty})^{H-\otriv}$ of the reduced Hopf bundle is a bundle, or sheaf,
of spectra over~$\RP^{\infty}$ whose typical fiber is the sphere spectrum.
The holonomy around the nontrivial loop acts on the homotopy groups~$\pi
_{\bullet }S^0$ of the sphere as multiplication by~$-1$.  In other words, for
each~$q$ the homotopy groups make a local system over~$\RP^{\infty}$ with
fiber~$\pi _qS^0$ and holonomy~$-1$.  Any other spectrum~$T$ is a module
over~$S^0$ and there is an induced bundle of spectra over~$\RP^{\infty}$ with
fiber~$T$; the holonomy acts on~$\pi _{\bullet }T$ as multiplication by~$-1$.

See~\cite{ABGHR} for a precise definitions, statements, and proofs of
these assertions about bundles of spectra and the twisted Thom isomorphism.
  \end{construction}

 \subsubsection{Interlude on unitarity}\label{subsubsec:4.2.5} 
 
See~\S\ref{subsec:2.8} for a general discussion of unitarity,
and~\S\ref{subsubsec:4.2.4} for its implementation in invertible theories.
Here we compute and interpret groups of low-dimensional oriented unitary
theories, using~\eqref{eq:137}. 
 
The group of $n=1$~spacetime dimensional oriented theories with
Eilenberg-MacLane target is 
  \begin{equation}\label{eq:139}
     [\Sigma ^1MTSO_1,\Sigma ^1H\CZ]\cong H^1(BSO_1;\CZ)=0; 
  \end{equation}
there is only the trivial theory~$F$ which assigns the trivial line~$L=\CC$ to
the oriented point~$\pt_+$ and the number~$1$ to the oriented circle.  This
theory is clearly unitariz\emph{able}.  A unitary structure is data, and
according to~\eqref{eq:137} the group of isomorphism classes of unitarity
data is
  \begin{equation}\label{eq:138}
     H^1(BO_1;\widetilde{\CZ})\cong H^1(\RP^{\infty};\widetilde{\CZ})\cong
     \zt, 
  \end{equation}
where the local system~$\widetilde{\CZ}\to BO_1$ has holonomy~$-1$ around the
nontrivial loop.\footnote{Here are two methods to compute~\eqref{eq:138}:
(1)~use a cell structure on~$\RP^{\infty}$ with a single cell in each
dimension, trivialize the local system over each cell, and derive the cochain
complex
  \begin{equation}\label{eq:140}
      \CZ\xrightarrow{\;1-c\;}\CZ\xrightarrow{\;1+c\;}\CZ\xrightarrow{\;1-c\;}\cdots 
  \end{equation}
in which $c$~is complex conjugation; (2) use the short exact coefficient
sequence ~\eqref{eq:133} where $\CC$, $\CZ$~are local systems with holonomy
$-1$ and $\ZZ$~is untwisted.}  The two equivalence classes of unitarity data
on~$F$ are easily explained.  Writing~$F(\pt_+)=L$ the unitarity data
provides an isomorphism $F(\pt_-)\xrightarrow{\;\cong \;}\overline{L}$.  The
oriented interval with both endpoints incoming is a bordism $\pt_-\amalg
\pt_+\to \emptyset ^0$ and applying~$F$ and the unitarity isomorphism we
obtain a nondegenerate hermitian form $h\:\overline{L}\otimes L\to\CC$.  Such
a form is either positive or negative, which accounts for~\eqref{eq:138}. 

  \begin{remark}[]\label{thm:49}
 The \emph{positivity} condition in a unitary theory means we should exclude
the negative form, so only allow a unique isomorphism class of unitarity
data.  The group~\eqref{eq:137} only implements the correct action of
orientation-reversal, not the positivity, but I do not know how to pick out
the subgroup corresponding to positive unitarity data in higher dimensions. 
  \end{remark}

As a further illustration we consider $n=2$~spacetime dimensional oriented
theories with Eilenberg-MacLane target.  These theories are discussed in
Example~\ref{thm:4} and Example~\ref{thm:13}.  The group of oriented theories
is  
  \begin{equation}\label{eq:141}
     [\Sigma ^2MTSO_2,\Sigma ^2H\CZ]\cong H^2(BSO_2;\CZ)\cong \CZ, 
  \end{equation}
and the theory~$F_z$ corresponding to~$z\in \CZ$ has 
  \begin{equation}\label{eq:144}
     F_z(X)=\lambda ^{\Euler(X)} 
  \end{equation}
for a closed oriented 2-manifold~$X$, where $\lambda =e^{2\pi iz}\in \Cx$.
Since $X$~has an orientation-reversing involution, only those theories with
$F_z(X)$~real can possibly be unitarizable, which forces $z\in i\RR$.  That
is consistent with the computation of ~\eqref{eq:137} for~$n=2$:
  \begin{equation}\label{eq:142}
     H^2(BO_2;\widetilde{\CZ})\cong \zt\oplus i\RR , 
  \end{equation}
together with the computation of the map
  \begin{equation}\label{eq:143}
     H^2(BO_2;\widetilde{\CZ})\longrightarrow H^2(BSO_2;\CZ) ,
  \end{equation}
which includes $\frac 12\ZZ/\ZZ\oplus i\RR\hookrightarrow \CZ$.  For
$z=1/2-ix$ we have~\eqref{eq:144} with $\lambda =-e^{2\pi x}$, but since the
Euler number is even the numerical invariants are positive.  The line
$F_z(\cir)=L$ attached to the oriented circle has a real structure, since
$\cir$~has an orientation-reversing involution (reflection), and the cylinder
with both boundaries incoming gives a nondegenerate real symmetric bilinear
form $h\:L\otimes L\to\CC$.  Glue 2-disks to each incoming boundary component
to deduce that $h(\ell ,\ell )=\lambda ^2$~is positive.  The conclusion is
that all theories with parameter~$\lambda \in \RR^\times \subset \Cx$ are
uniquely unitarizable.

\section{SRE phases}\label{sec:5}

There is one more preliminary before we can state our proposed topological
invariant of short-range entangled (SRE) phases: we must specify our choice
of target spectrum for classifying long-range effective topological theories.
Recall that in general to define a topological field theory~\eqref{eq:5} we
need to specify a target higher category~$\sC$, but for an invertible theory
we need the much weaker information of the sub-groupoid~$\sC^{\times }$ of
invertibles; see~\eqref{eq:15}.  While concrete arguments give information
about the highest few homotopy groups of~$\sC^\times $, we can only guess at
the structure lower down, which is increasingly relevant as the dimension of
the theory increases.  In the general fermionic case we take a universal
choice, the \emph{Brown-Comenetz dual to the sphere spectrum}.  In the
bosonic case we have solid arguments to determine the ``top part'' of the
spectrum, but after that only sparse data points.  For that reason we take
the same target spectrum as in the fermionic case---after all, theories with
only bosons are special cases of general theories---though in space
dimension~$d=1$ the ``top part'', an \emph{Eilenberg-MacLane spectrum}, is
all that is relevant and so we use it instead.  We explain these choices
in~\S\ref{subsec:5.1}.  Our main proposal is stated precisely
in~\S\ref{subsec:5.2}.  In~\S\ref{subsec:5.3} we work out the relationship to
the group cohomology classifications in the literature~\cite{CGLW},
\cite{GW}.

  \subsection{Target spectra}\label{subsec:5.1}

The discussions in Example~\ref{thm:13} and especially \S\ref{subsec:3.3}(2)
are relevant here.  

\subsubsection{Preliminary: duals to the sphere
spectrum}\label{subsubsec:5.1.4} 
 
Let $A$~be an abelian group, which might be discrete or have a topology.  (In
the latter case there is an additional hypothesis: $A$~is locally compact.)
Examples: $\zmod n$, $\ZZ$, $\RZ$.  There are two notions of dual group we
might consider.  The first is the \emph{Pontrjagin dual}, which is the group
of continuous homomorphisms $A\to\RZ$.  The Pontrjagin dual of~$\zmod n$ is
isomorphic to~$\zmod n$, the Pontrjagin dual of~$\ZZ$ is isomorphic to~$\RZ$,
and the Pontrjagin dual of~$\RZ$ is isomorphic to~$\ZZ$.  On the other hand,
we can consider an integral dual $\Hom(A,\ZZ)$.  However, the naive
interpretation is not so well behaved.  For example, if $A=\zmod n$ there are
no nonzero homomorphisms $A\to\ZZ$.  The resolution is to consider
$\Hom(-,\ZZ)$ in the \emph{derived} sense, which means that we replace~$A$ by
a free chain complex whose homology in degree~0 is~$A$ and then
compute~$\Hom$.  The result is called~$\Ext^{\bullet }(A,\ZZ)$ and is the
``correct'' integral dual.  For $A=\zmod n$ the only nonzero group is in
degree~$-1$, whereas for~$A=\ZZ$ the only nonzero group is in degree~0.

Both notions of duality exist in the world of spectra; see~\cite{An},
\cite[Appendix~B]{HS}, \cite[Appendix~B]{FMS}, \cite{HeSt} for precise
definitions and discussion..  Recall that a spectrum~$T_{\bullet }$ has an
associated $\ZZ$-graded abelian group~$\pi _{\bullet }T$.  For the sphere
spectrum~$S^0$ the first several homotopy groups are
  \begin{equation}\label{eq:114}
     \pi \mstrut _{\{0,1,2,\dots \}}S^0 \cong
     \{\ZZ\,,\,\zmod2\,,\,\zmod2\,,\,\zmod{24}\,,\,0\,,\,0\,,\,\dots \} 
  \end{equation}
The analog of the Pontrjagin dual is the \emph{Brown-Comenetz dual}.  For the
sphere we denote it as~$I\CZ$.  (We replace~$\RZ$ with~$\CZ$; in topology it
is often $\QQ/\ZZ$ instead.)  Its homotopy groups are the Pontrjagin dual
groups to~\eqref{eq:114}: 
  \begin{equation}\label{eq:115}
     \pi \mstrut _{\{0,-1,-2,\dots \}}I\CZ \cong
     \{\CZ\,,\,\zmod2\,,\,\zmod2\,,\,\zmod{24}\,,\,0\,,\,0\,,\,\dots \} 
  \end{equation}
Note carefully the indexing difference on the left hand side
between~\eqref{eq:115} and~\eqref{eq:114}.  The analog of the integral dual
is the \emph{Anderson dual}.  For the sphere we denote it as~$I\ZZ$.  Its
homotopy groups are the integral dual to~\eqref{eq:114}; the torsion groups
are degree shifted:
  \begin{equation}\label{eq:116}
     \pi \mstrut _{\{0,-1,-2,\dots \}}I\ZZ \cong
     \{\ZZ\,,\,0\,,\,\zmod2\,,\,\zmod2\,,\,\zmod{24}\,,\,0\,,\,0\,,\,\dots \} 
  \end{equation}
There is a fiber sequence $I\ZZ\to I\CC\to I\CZ$ and a corresponding long
exact sequence of homotopy groups.  Note $I\CC\simeq H\CC$ is an
Eilenberg-MacLane spectrum with a single nonzero homotopy group.

These spectra enjoy universal properties which make them universal targets
for invertible field theories.  Recall that the notation `$[X,X']$' is used
for the abelian group of homotopy classes of maps $X\to X'$ between spectra.
If $X$~is any spectrum, then
  \begin{equation}\label{eq:121}
     [X,\Sigma ^nI\CZ] = I\CZ^n(X)\cong \Hom(\pi _nX,\CZ). 
  \end{equation}
(To compute the $I\CZ$~cohomology of a space~$Y$, set $X=\Sigma ^{\infty}Y$
the suspension spectrum, so $\pi _nX$~is the
$n^{\textnormal{th}}$~\emph{stable} homotopy group of~$Y$.)  Thus an
invertible $n$-spacetime dimensional theory whose values on closed
$n$-manifolds are numbers in~$\CZ$ pushes uniquely to a field theory with
values in~$I\CZ$, and any theory with target~$I\CZ$ is determined by its
numerical values on $n$-manifolds.  For maps from any spectrum~$X$ into~$I\ZZ$ 
there is a short exact sequence
  \begin{equation}\label{eq:122}
     0\longrightarrow \Ext^1(\pi _{n-1}X,\ZZ)\longrightarrow
     I\ZZ^n(X)\longrightarrow \Hom(\pi _nX,\ZZ)\longrightarrow 0 
  \end{equation}
which is split, but not canonically.

\subsubsection{Bosonic theories}\label{subsubsec:5.1.1}

We continue with theories~$F$ defined in $d$-space dimensions and assume
first that $F$~ is a bosonic theory.  The target is a
$(d+1)$-category~$\sC_{\textnormal{bose}}$, but since $F$~is invertible we
need only the groupoid~$\sC_{\textnormal{bose}}^\times $.  Standard quantum
mechanics dictates that $F(Y)$~is a complex vector space for any closed
$d$-manifold~$Y$, and since $F$~is invertible it must be 1-dimensional.
Furthermore, a diffeomorphism of~$Y$ acts as an invertible map $F(Y)\to
F(Y)$, which is multiplication by a scalar in\footnote{In a unitary theory
this scalar has unit norm, but the partition function of a general
$(d+1)$-manifold need not.}~$\Cx$.  Thus the truncation~$\Omega
^d\sC_{\textnormal{bose}}^\times $ of the target groupoid to the top two
levels is the groupoid~$\Line$ of complex lines.  A crucial point, discussed
in~\S\ref{subsec:3.3}(2), is that we use the \emph{continuous} topology on
the morphism spaces, in particular on~$\Cx$.  This gives information about
some homotopy groups of~$\sCx_{\textnormal{bose}}$:
  \begin{equation}\label{eq:52}
     \pi _d\sCx_{\textnormal{bose}}=0,\qquad \pi _{d+1}\sCx_{\textnormal{bose}}=0,\qquad \pi _{d+2}\sCx_{\textnormal{bose}}\cong \ZZ. 
  \end{equation}
The first equality expresses the fact that any two lines are isomorphic; the
latter two that $\Cx$~is connected with infinite cyclic fundamental group.
Furthermore, $\pi _k\sCx_{\textnormal{bose}}=0$ for~$k>d+2$ since there are
no non-identity $k$-morphisms for~$k>d+1$.  

It remains to determine the lower homotopy groups.  We have quite solid
information about the next homotopy group down, as we know the nature of the
higher categorical invariant usually attached to manifolds of dimension~$d-1$
by the theory~$F$.  For example, a standard choice is to assign a linear
category to a closed $(d-1)$-manifold, and since all invertible linear
categories are isomorphic deduce
  \begin{equation}\label{eq:54}
     \pi _{d-1}\sCx_{\textnormal{bose}}=0. 
  \end{equation}
An alternative choice is to assign an invertible complex algebra to a closed
$(d-1)$-manifold, and the triviality of the Brauer group of~$\CC$ leads to
the same conclusion~\eqref{eq:54}.  At this stage we have reproduced the
first part of Example~\ref{thm:13}, in which the target is a 2-category and
there are no further homotopy groups.  Note that \eqref{eq:52}~and
\eqref{eq:54}~amount to an Eilenberg MacLane spectrum~$\Sigma ^{d+2}H\ZZ$: a
single nonzero homotopy group.  This is the full story for~$d\le 1$.

What is perhaps surprising is that we postulate a nonzero homotopy group at
the next level down: 
  \begin{equation}\label{eq:117}
     \pi _{d-2}\sCx_{\textnormal{bose}}\not= 0. 
  \end{equation}
One rationale for ~\eqref{eq:117} is the 4-dimensional integral invertible
oriented extended topological field theory
  \begin{equation}\label{eq:119}
     \sigma \:\bord 4(\textnormal{orientation})\longrightarrow
     \sCx_{\textnormal{bose}}(d=2)
  \end{equation}
for which the integer invariant~$\sigma (W)$ of a closed oriented
4-manifold~$W$ is its signature~$\Sign(W)$.  In its analytic incarnation this
invariant involves only differential forms on~$W$, not spinor fields, so in
that sense is bosonic.  The theory factors through a map
  \begin{equation}\label{eq:120}
     \Sigma ^4MTSO_4\longrightarrow \Sigma ^4KO 
  \end{equation}
which can be viewed as the universal symbol~\cite[\S3]{FHT} of the signature
operator in dimension~4.  The relevant stretch of homotopy groups of the
target spectrum in~\eqref{eq:120} is 
  \begin{equation}\label{eq:131}
     \pi \mstrut _{\{0,1,2,3,4\}}\cong \{\ZZ\,,\,0\,,\,0\,,\,0\,,\,\ZZ\, \} .
  \end{equation}
This theory does \emph{not} factor through $\Sigma ^4H\ZZ$, so the
bottom~$\ZZ$ in~\eqref{eq:131} cannot be replaced by~$0$.  Another piece of
evidence for at least a nonzero homotopy group in this spot ($\pi
_{d-2}\sCx_{\textnormal{bose}}$) comes from \emph{conformal nets}.  These are
a possible target for 3-dimensional theories with numerical invariants
in~$\CZ$, and conjecturally not all invertible conformal nets are isomorphic.
(See~\cite{DH} for a discussion of conformal nets in this context.)

We do not have any information about lower homotopy groups.  So we could take
the shifted Postnikov truncation $\Sigma ^{d-6}KO\langle 4,...,8 \rangle $ as
a reasonable choice of target spectrum~$\sCx_{\textnormal{bose}}$.  But the
$\ZZ$~at the bottom of the spectrum is overkill---a smaller torsion group
will do---and so instead for theories in dimension~$d\ge2$ we use the
universal choice, the Anderson dual of the sphere.

  \begin{hypothesis}[]\label{thm:20}
 The target spectrum for classifying long-range effective theories of a
$d$-space dimensional bosonic system is~$\Sigma ^{3}H\ZZ$ for~$d=1$
and~$\Sigma ^{d+2}I\ZZ$ for~$d\ge2$.
  \end{hypothesis}

  \begin{remark}[]\label{thm:21}
 The classification of anomalies has similar target spectra, but as indicated
in~\S\ref{subsec:3.3}(2) we use instead the discrete topology on~$\Cx$.  Thus
for~$d\ge2$ the target which classifies anomaly theories is $\Sigma
^{d+2}I\CZ$; for~$d=1$ it is $\Sigma ^3H\CZ$.
  \end{remark}

\subsubsection{Fermionic theories}\label{subsubsec:5.1.2}

The hypothesis for theories with fermions is different.  Namely, the
dichotomy between bosonic and fermionic states in quantum mechanical systems
is encoded by stipulating that the quantum Hilbert space be $\zt$-graded:
states with even grading are bosonic and states with odd grading are
fermionic.  That persists in the long-range effective theory: vacua are
either bosonic or fermionic.  So the complex lines~$F(Y)$ in an invertible
long-rang theory are either even or odd.  The existence of distinct
isomorphism classes of $\zt$-graded lines modifies~\eqref{eq:52}:
  \begin{equation}\label{eq:53}
     \pi _d\sCfx\cong \zt,\qquad \pi _{d+1}\sCfx=0,\qquad \pi
     _{d+2}\sCfx\cong \ZZ.  
  \end{equation}
The $k$-invariant is nonzero; it is the composition~$\beta _{\ZZ}\circ Sq^2$
of the integer Bockstein and the Steenrod square.

As in the bosonic case higher homotopy groups vanish.  But now we expect many
nontrivial lower homotopy groups.  For example, we expect \eqref{eq:54}~to be
replaced by
  \begin{equation}\label{eq:55}
     \pi _{d-1}\sCfx\cong \zt 
  \end{equation}
since the super Brauer group of invertible $\zt$-graded complex algebras has
two elements (represented by even and odd complex Clifford algebras).  The
idea that modules over Clifford algebras may be used as the state space of a
quantum system is familiar in condensed matter theory; see also~\cite{F5}.
From the super Brauer category we compute how the three Eilenberg-MacLane
spectra fit together, though we do not do so here.  We can gather information
about~$\pi _{d-2}\sCfx$ by arguments analogous to those
in~\S\ref{subsubsec:5.1.1}.  For example, there is an invertible theory
  \begin{equation}\label{eq:126}
     \Sigma ^4MT\Spin_4\longrightarrow KO 
  \end{equation}
which assigns the $\Ahat$-genus to a closed spin 4-manifold.  (The universal
symbol is quaternionic, which explains the absence of shift in~$KO$.)  This
suggests that $\pi _{d-2}\sCfx\not= 0$.  Invertible fermionic conformal nets
are another clue, but the only knowledge is conjectural.

Therefore, by fiat really, we make a \emph{universal} choice for the target
spectrum, namely the {Anderson spectrum}~$\IZ$.  Notice that its first four
homotopy groups ~\eqref{eq:116} agree with ~\eqref{eq:53} and~\eqref{eq:55}
and more precisely the Postnikov truncations are equivalent.

  \begin{hypothesis}[]\label{thm:22}
  The target spectrum for classifying long-range effective theories of a
$d$-space dimensional fermionic system is~$\Sigma ^{d+2}\IZ$.
  \end{hypothesis}

  \begin{remark}[]\label{thm:23}
 We use the shift~$\Sigma ^{d+2}\ICZ$ as a target to classify anomaly
theories.
  \end{remark}

\subsubsection{Antilinear symmetries}\label{subsubsec:5.1.3} 
 
As explained in~\S\ref{subsubsec:2.4.3} antilinear symmetries lead to
anomalous field theories after gauging.  For invertible theories we account
for the anomaly using twisted cohomology.   
 
We implement complex conjugation on the target spectra by the universal
$(-1)$-action defined in Construction~\ref{thm:48}.  Consider, for example,
the Eilenberg-MacLane spectrum ~$H\ZZ$.  In degree~1 a cohomology class is
represented by a map to~$\Cx$, and the $(-1)$-action is complex conjugation
on~$\Cx$.  The classifying space for a degree~1 class is ~$\Cx$, and
Construction~\ref{thm:48} gives a fiber bundle (not principal!)  with
fiber~$\Cx$ over~$\RP^{\infty}$ whose holonomy acts as complex conjugation
on~$\Cx$.  In degree~2, as just explained, a cohomology class is represented
by a complex line bundle and this action is complex conjugation on complex
line bundles.  In degree~3 there is a similar story with bundles of complex
algebras, and it is reasonable to extend this picture to all degrees.  The
same story applies to~$H\CZ$, which one can view as flat elements in~$H\ZZ$
(with a degree shift).  For example, an element in degree~0 in~$H\CZ$ is a
locally constant map to~$\Cx$.  Similar considerations apply to other target
spectra.  The top nonzero homotopy group of ~$\IZ$ (respectively~$\ICZ$) is
the same as that of~$H\ZZ$ (respectively~$H\CZ$), and the $(-1)$-action is
trivial on the $\zt$ homotopy groups in~\eqref{eq:53} and~\eqref{eq:55} is
trivial.  Our tentative grasp on lower homotopy groups puts further
justification beyond reach.

  \subsection{Topological invariants of short-range entangled
  phases}\label{subsec:5.2}

We propose a home for long-range effective topological field theories of
gapped systems with short-range entanglement.  We do not know if the map from
the microscopic phases to the deformation classes of field theories is either
injective or surjective.  Regardless, the evidence presented in the remainder
of the paper suggests that it is a very effective invariant.

Assume the theory is $d$-space dimensional and has a global symmetry
group~$G$, which is a Lie group equipped with a smooth homomorphism
  \begin{equation}\label{eq:56}
     \phi \:G\longrightarrow \pmo=\{\pm1\}
  \end{equation}
which encodes linearity vs.~antilinearity: an element ~$g\in G$ with~$\phi
(g)=1$ acts linearly and an element~ $g\in G$ with~$\phi (g)=-1$ acts
antilinearly.  Recall our assumption, stated before~\eqref{eq:35}, that there
is a $G$-equivariant extension of the long-range effective theory
(\emph{non-anomalous case}) or that there is an anomaly and an anomalous
extension (\emph{anomalous case}).

The hypotheses underlying the proposal are stated in~\S\ref{subsec:3.3}
and~\S\ref{subsec:5.1}.  We use the twisted Thom isomorphism~\eqref{eq:136}
(and its variations for $SO$~and $\Spin$ replacing~$O$).

\subsubsection{Bosonic theories}\label{subsubsec:5.2.1}

According to Hypothesis~\ref{thm:20} the target spectrum is~$\Sigma
^{d+2}H\ZZ$ for~$d=1$ and $\Sigma ^{d+2}I\ZZ$ for~$d\ge2$.  We use the
notation ~`$\Sigma ^{d+2}T\ZZ$' for this target spectrum.  Note that
$T\ZZ$~is a ring spectrum.  Let $\tTZ$~denote the Thom twisting of the
ring spectrum~$T\ZZ$ associated to the virtual bundle (see~\eqref{eq:129})
  \begin{equation}\label{eq:145}
     \vtriv{\RR^d}-S(d)\longrightarrow BSO_d,
  \end{equation}
and $\otTZ$~the analogous Thom twisting for $O_d$~in place of~$SO_d$.  Let
$\wTZ$~denote the $(-1)$-twist of~$T\ZZ$ (Construction~\ref{thm:48})
associated to the double cover $BSO_d\to BO_d$.  The
homomorphism~\eqref{eq:56} determines a $(-1)$-twist of~$T\ZZ$ associated to
the double cover\footnote{If $\phi $~is identically~$+1$, then $G_0=G$ and
$\fTZ=0$~is trivial.} $BG_0\to BG$, where $G_0=\ker\phi $; we denote it~$\fTZ$.
Degree shifts are Thom twistings of trivial bundles, whence sums of twistings
and degrees are defined.

  \begin{proposal}[bosonic theories]\label{thm:24}
  Short-range entangled phases of a $d$-space dimensional bosonic theory with
global symmetry group~$G$ map to the abelian group
  \begin{equation}\label{eq:57}
     \SREb=T\ZZ\mlstrut^{\tTZ+\fTZ +d+2}(BSO_d\times BG)
  \end{equation}
in the non-anomalous case.  The unitarizable theories lie in the image of the
map\footnote{A more precise proposal would identify the subgroup of the
domain of~\eqref{eq:146} representing unitary theories which satisfy
positivity.  As I do not know how to do
this---see~\S\ref{subsubsec:4.2.5}---we settle for a weaker formulation.}
  \begin{equation}\label{eq:146}
     T\ZZ\mlstrut^{\otTZ+\wTZ+\fTZ +d+2}(BO_d\times BG)\longrightarrow
     T\ZZ\mlstrut^{\tTZ+\fTZ +d+2}(BSO_d\times BG). 
  \end{equation}
Anomalies are classified by the abelian group
  \begin{equation}\label{eq:58}
     \Anob=T\CZ\mlstrut^{\tTZ+\fTZ +d+2}(BSO_d\times BG)
  \end{equation}
and the anomalous theories with fixed anomaly map to a torsor for the abelian
group~\eqref{eq:57}.
  \end{proposal}

\noindent
 An anomaly theory has spacetime dimension~$d+2$, but here is only defined on
manifolds of dimension~$\le d$; this explains the degrees in~\eqref{eq:58}.
The last assertion is that any two choices of anomalous theories with the
same anomaly are related by tensoring with a non-anomalous theory.

\subsubsection{Fermionic theories}\label{subsubsec:5.2.2}

The proposal for fermionic theories is similar: we simply swap out the
Eilenberg-MacLane spectra we used in the $d=1$~bosonic case for the Anderson
and Brown-Comenetz spectra (\S\ref{subsubsec:5.1.2}) and assume all manifolds
are spin in addition to being oriented.  The corresponding Thom twistings
of~\eqref{eq:145} are denoted~$\tIZ$ and~$\otIZ$; the $(-1)$-twisting
of~$\IZ$ associated to the double cover $BSO_d\to BO_d$ is~$\wIZ$; and the
$(-1)$-twisting of~$\IZ$ associated to the double cover $BG_0\to BG$
is~$\fIZ$.

  \begin{proposal}[fermionic theories]\label{thm:25}
  Short-range entangled phases of a $d$-space dimensional fermionic theory
with global symmetry group~$G$ map to the abelian group
  \begin{equation}\label{eq:61}
     \SREf=\IZ\mlstrut^{\tIZ+\fIZ +d+2}(B\Spin_d\times BG)
  \end{equation}
in the non-anomalous case.  The unitarizable theories lie in the image of the
map  
  \begin{equation}\label{eq:147}
     \IZ\mlstrut^{\otIZ+\wIZ+\fIZ +d+2}(B\Pin_d\times BG)\longrightarrow
     \IZ\mlstrut^{\tIZ+\fIZ +d+2}(B\Spin_d\times BG). 
  \end{equation}
Anomalies are classified by the abelian group
  \begin{equation}\label{eq:62}
     \Anof=\ICZ\mlstrut^{\tIZ+\fIZ +d+2}(B\Spin_d\times BG)
  \end{equation}
and the anomalous theories with fixed anomaly map to a torsor for the abelian
group~\eqref{eq:61}.
  \end{proposal}

\subsubsection{Symmetry protected topological phases}\label{subsubsec:5.2.3}

Now we address the question of \emph{symmetry protected topological} (SPT)
phases.  The ``symmetry protection'' means that the effective topological
field theory~$F$ is trivial when the symmetry is ignored.  As explained at
the end of~\S\ref{subsec:3.4} this means that the $G$-extension~$\tF$ is
trivial when restricted to the trivial $G$-bundle.  The basepoint of~$BG$
determines an embedding $X\hookrightarrow X\wedge BG_+$ for any spectrum~$X$,
and so a restriction map
  \begin{equation}\label{eq:63}
 [X\wedge BG_+,X']\longrightarrow  [X,X']
  \end{equation}
for any spectrum~$X'$.  When $X$~is a Madsen-Tillmann spectrum we can
rewrite~\eqref{eq:63} using the twisted Thom isomorphism; then the map is
pullback along the inclusion 
  \begin{equation}\label{eq:169}
     BO_d\hookrightarrow BO_d\times BG 
  \end{equation}
defined by the basepoint of~$BG$.

  \begin{proposal}[]\label{thm:26}
 \
  \begin{enumerate}[(i)]

 \item Symmetry protected topological phases of a $d$-space dimensional
bosonic theory with global symmetry group~$G$ map to the kernel
  \begin{equation}\label{eq:64}
     \SPTb=\ker\Bigl(T\ZZ\mlstrut^{\tTZ+\fTZ +d+2}(BSO_d\times BG)
     \longrightarrow T\ZZ\mlstrut^{\tTZ+d+2}(BSO_d) \Bigr) 
  \end{equation}
of the indicated restriction map constructed from~\eqref{eq:169}.

 \item Symmetry protected topological phases of a $d$-space dimensional
fermionic theory with global symmetry group~$G$ map to the kernel
  \begin{equation}\label{eq:65}
     \SPTf=\ker\Bigl(\IZ\mlstrut^{\tIZ+\fIZ +d+2}(B\Spin_d\times BG) 
     \longrightarrow  \IZ\mlstrut^{\tIZ +d+2}(B\Spin_d)\Bigr) 
  \end{equation}
of the indicated restriction map.  
 
  \end{enumerate}
  \end{proposal}

  \subsection{Relation to group (super) cohomology}\label{subsec:5.3}

 From the definition of $T\ZZ$ at the beginning of~\S\ref{subsubsec:5.2.1} we
construct a map of spectra
  \begin{equation}\label{eq:123}
     \Sigma ^{d+2}H\ZZ\longrightarrow \Sigma ^{d+2}\TZ.
  \end{equation}
It induces the second map in the composition
  \begin{equation}\label{eq:67}
     H^{d+2}(BG;\ZZ_\phi )\longrightarrow H^{d+2}(BSO_d\times
     BG;\ZZ_\phi ) \longrightarrow T\ZZ\mlstrut ^{\tTZ+\fTZ+d+2}(BSO_d\times BG);
  \end{equation}
the first is induced from the projection $BSO_d\times BG\to BG$.  Here
$\ZZ_\phi \to BG$ is the local system defined by~\eqref{eq:56}.  Note that
ordinary cohomology is oriented for oriented vector bundles, which explains
why the twisting~$\tTZ$ is trivialized when restricted under~\eqref{eq:123}.
The homomorphism~\eqref{eq:67} maps the group cohomology phases
discussed\footnote{To compare it helps to observe that $H^{d+2}(BG;\ZZ_\phi
)\cong H^{d+1}(BG;\underline{U(1)}_\phi )$, where $\underline{U(1)}$~has its
continuous topology.} in~\cite{CGLW} to~$\SREb$.  Furthermore, the image
of~\eqref{eq:67} lies in the subgroup~$\SPTb$ of symmetry protected phases;
see~\eqref{eq:64}.

  \begin{remark}[]\label{thm:46}
 If $G$~is discrete, then the topological cohomology of~$BG$ is isomorphic to
the group cohomology of the group~$G$.  More generally, if $G$~is a (finite
dimensional) Lie group, then a theorem of D.~Wigner~\cite{Wi} states that the
topological cohomology of~$BG$ with coefficients in a discrete $G$-module is
isomorphic to the Borel cohomology of~$G$; see~\cite{St} for more on group
cohomology.
  \end{remark}

For theories with fermions Gu and Wen~\cite{GW} introduced a group ``super''
cohomology theory.  In fact, it can be identified with a certain generalized
cohomology theory of the classifying space~$BG$.  This generalized cohomology
theory, which we simply call~$E$, had already appeared in at least a few
contexts in theoretical physics: (1)~in spin Chern-Simons theories~\cite{J}
and (2)~in QCD, in the Wess-Zumino term of the long-range effective theory of
pions~\cite{F4}.  The spectrum~$E$ has two nonzero homotopy groups: 
  \begin{equation}\label{eq:71}
     \pi _0E\cong \ZZ, \qquad \pi _{-2}E\cong \zt.  
  \end{equation}
The $k$-invariant which relates them is nonzero.  We defer to~\cite[\S1]{F4}
for generalities on this cohomology theory.

The two nontrivial homotopy groups in~$E$ occur in~\eqref{eq:53}, shifted up
by degree~$d+2$, and as the $k$-invariant match there is a map $E\to \IZ$ of
spectra.  The theory~ $E$ is oriented for spin bundles, as proved in
\cite[Proposition~4.4]{F4}.  Therefore, there is a homomorphism\footnote{The
map~\eqref{eq:68} is part of a long exact sequence; the terms which come
before and after are maps into the spectrum which is the cofiber of~$E\to\IZ$
so has vanishing homotopy group in degrees~$\ge-2$.  The spin orientation
of~$E$ explains why $\tIZ$~does not appear in the domain of~\eqref{eq:68}.}
  \begin{equation}\label{eq:68}
     E\mlstrut ^{\fE+d+2}(B\Spin_d\times BG)
     \longrightarrow \SREf=\IZ\mlstrut^{\tIZ+\fIZ +d+2}(B\Spin_d\times BG) . 
  \end{equation}
The projection $B\Spin_d\times BG\to BG$ induces an inclusion
  \begin{equation}\label{eq:70}
     E^{\fE +d+2}(BG)\longrightarrow E^{\fE +d+2}(B\Spin_d\times BG) 
  \end{equation}
which, after composition with~\eqref{eq:68}, induces a homomorphism of the
group ``super'' cohomology into $\SREf$, as expected.  Note that the image
of~~$E^{\fE +d+2}(BG)$ in~$\SREf$ lies in the subgroup~$\SPTf$ of symmetry
protected phases; see~\eqref{eq:65}.

   \section{Computations and special cases}\label{sec:6}

We illustrate how the proposed invariants of gapped short-range entangled
(SRE) phases in~\S\ref{subsec:5.2} detect phases not covered by the group
cohomology classification discussed in~\S\ref{subsec:5.3}.  We organize the
discussion by space dimension~$d$ and by whether or not the theory includes
fermions.  The examples treated here are non-anomalous.  Clearly there are
many more computations and analyses which can be carried out.

  \subsection{$d=1$ bosonic theories: group cohomology}\label{subsec:6.1}

This is the one case in which there is nothing beyond the group cohomology
classification.  There are two reasons for this: (1) the group~$SO_1$ which
governs the spatial tangential structure is trivial, and (2)~for $d=1$ we
have $T\ZZ=H\ZZ$.  More formally, from~\eqref{eq:57} and~\eqref{eq:64} we
deduce
  \begin{equation}\label{eq:72}
   \begin{split}
     \SREba 1G\phi &=\SPTba 1G\phi \\&=H\ZZ\mlstrut ^{\fHZ+3}(BSO_1\times BG)
     \\&\cong H^3(BG;\ZZ_\phi ).
  \end{split}\end{equation}

  \subsection{$d=1$ fermionic theories}\label{subsec:6.2}

Since the shifted Madsen-Tillmann spectra are connective (in this case the
relevant spectrum is $\Sigma ^{1}MT\Spin_1\wedge BG_+$), we can replace the
codomain~$\Sigma ^3\IZ$ of an invertible topological field theory by its
connective cover.  That connective cover is a module for $ko$-theory, which
is connective real $K$-theory.  (See~\cite[\S4]{F5}, for example, where that
connective cover is the theory called~`$R\inv $'.)  In particular, the
connective cover is Spin-oriented so the Thom twisting is trivial.  Hence
  \begin{equation}\label{eq:73}
     \SREfa 1G\phi \cong \IZ\mlstrut^{\fIZ +3}(B\Spin_1\times BG). 
  \end{equation}
The subgroup coming from the $BG$~factor is 
  \begin{equation}\label{eq:74}
     \IZ\mlstrut^{\fIZ +3}(BG). 
  \end{equation}

  \begin{remark}[]\label{thm:30}
 This already goes beyond the group ``super'' cohomology theory~$E$, since
$E$~ has two nonzero cohomology groups~\eqref{eq:71}, whereas the truncation
of~$\IZ$ we are using here has a third nonzero homotopy group~$\pi _{-3}\cong
\zt$.
  \end{remark}

Consider the special case~$G=\pmo$ with $\phi $~nontrivial; this is the case
of a time-reversal symmetry which squares to the identity.  Then
\eqref{eq:74}~is cyclic of order~8:
  \begin{equation}\label{eq:75}
     \IZ\mlstrut^{\fIZ +3}(B\pmo)\cong \zmod8 .
  \end{equation}
One proof of~\eqref{eq:75} is \cite[Theorem~3.13]{DFM}.   

  \begin{remark}[]\label{thm:27}
 One interpretation of the left hand side of~\eqref{eq:75} is the group of
degree shifts of $KO$-theory, which is the Brauer group of \emph{real}
$\zt$-graded central simple algebras.  This is surely very closely related to
the classification in~\cite[\S V]{FK}.
  \end{remark}

Because the group~$\Spin_1\cong \zt$ is nontrivial, there are additional
SRE~phases~\eqref{eq:73} not captured by~\eqref{eq:74}.  One example is the
truncation to 1-space dimension of the 2-spacetime dimensional ``Arf
theory''.  The invariant of a 2-dimensional closed spin manifold is~$\pm1$
according to the Arf invariant of the quadratic form defining the spin
structure: even spin structures have invariant~$+1$ and odd spin structures
have invariant~$-1$.  The invariants on 1- and 0-dimensional manifolds are
also explicit.  The invariant of~$\cir$ is the trivial even line for the
bounding spin structure and the odd line for the nonbounding spin structure.
The invariant of~$\pt_+$ with the standard spin structure is the complex
Clifford algebra~$\Cliff^{\CC}_1$.  (For simplicity we take the target
2-category ~$\sC$ of the theory to be algebras-bimodules-intertwiners.
Equivalently, we can assign to~$\pt_+$ the category of $\zt$-graded
$\Cliff^{\CC}_1$-modules.)  For more on the Arf theory, and a beautiful
geometric application, see~\cite{G}.   
 
It appears that the Arf theory is realized as the long-range effective
topological theory of the Majorana chain~\cite{K6} in its nontrivial phase.
For example, the description of the effective theory~\cite[(15)]{K6} does
assign the Clifford algebra~$\Cliff^{\CC}_1$ to a point.

In the remainder of this section we illustrate some relevant computational
techniques.  More elaborate techniques are needed in higher dimensions, as we
illustrate in the appendix.  The Arf theory is predicted by
  \begin{equation}\label{eq:91}
     \IZ^3(\Sigma ^1MT\Spin_1)\cong \IZ^3(B\Spin_1)\cong \zt. 
  \end{equation}
In fact, since the Arf theory extends to a 2-spacetime dimensional theory, we
can detect it in the group $\IZ^3(\Sigma ^2MT\Spin_2)$.  It is illuminating
to first compute
  \begin{equation}\label{eq:92}
  \begin{split}
     \ICZ^2(\Sigma ^2MT\Spin_2)&\cong \ICZ^2(B\Spin_2)\\&\cong
     \ICZ^2(\CP^{\infty})\\&\cong \Hom(\pi _2^s\CP^{\infty}_+,\CZ)\\&\cong
     \Hom(\pi _2^s\CP^{\infty}\times \pi _2^sS^0,\CZ)\\&\cong \CZ\times \zt. 
  \end{split}
  \end{equation}
The first line is the Thom isomorphism; the third the defining property
of~$\ICZ$; the fourth the general fact $\Sigma ^{\infty}(X_+)\simeq \Sigma
^{\infty}X\vee S^0$; and the last line the results $\pi
_2^s(\CP^{\infty})\cong \ZZ$ (computed in~\cite{Li}, for example) and $\pi
_2^sS^0\cong \zt$.  This is the set of 2-dimensional invertible spin theories
with target~$\ICZ$.  It includes the family of Euler theories
(Example~\ref{thm:4}) parametrized by~$ \CZ$ and the Arf theory.  The group
of path components of~\eqref{eq:92} is
  \begin{equation}\label{eq:95}
      \IZ^3(\Sigma ^2MT\Spin_2)\cong \IZ^3(\CP^{\infty})\cong \zt. 
  \end{equation}
One computation of this group uses the split short exact
sequence~\eqref{eq:122}
  \begin{equation}\label{eq:93}
     0 \longrightarrow \Ext(\pi _2\Sigma ^2MT\Spin_2)\longrightarrow
     \IZ^3(\Sigma ^2MT\Spin_2) \longrightarrow \Hom(\pi _3\Sigma
     ^2MT\Spin_2,\ZZ)\longrightarrow 0 
  \end{equation}
and the computations $\pi _3^s(\CP^{\infty})=0$ and $\pi _3^s(S^0)\cong
\zmod{24}$.  

  \begin{figure}[ht]
\bigskip
  \centering
  \includegraphics[scale=.8]{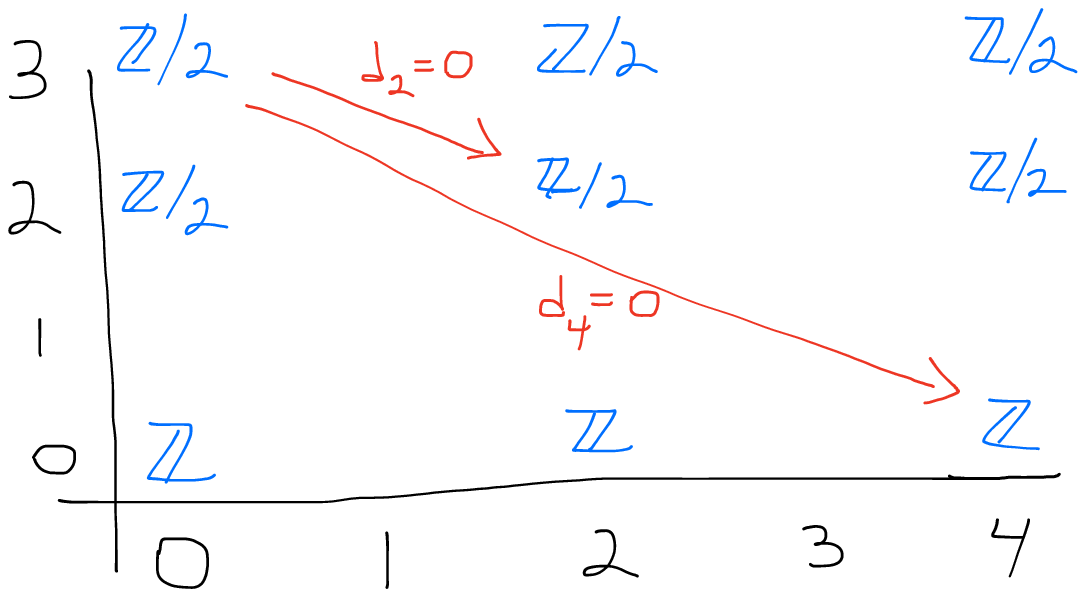}\phantom{MMMMMMMMMM}
  \caption{Computation of $\IZ^3(\CP^{\infty})$}\label{fig:7}
  \end{figure}
\bigskip

The Atiyah-Hirzebruch spectral sequence provides another means to compute the
generalized cohomology groups~\eqref{eq:92} and~\eqref{eq:95}.  The relevant
portion of the $E_2$~page of the spectral sequence 
  \begin{equation}\label{eq:94}
      E_2^{pq}=H^p\bigl(\CP^{\infty};\IZ^q(\pt) \bigr)\Longrightarrow
     \IZ^{p+q}(\CP^{\infty}) 
  \end{equation}
for computing~\eqref{eq:95} is shown in Figure~\ref{fig:7}.  The
differentials emanating from the initial column vanish, as can be seen from
the splitting $\pt\to \CP^{\infty}\to \pt$ of the projection to a point.
That reasoning also applies to the spectral sequence
  \begin{equation}\label{eq:96}
     E_2^{pq}=H^p\bigl(\CP^{\infty};\ICZ^q(\pt) \bigr)\Longrightarrow
     \ICZ^{p+q}(\CP^{\infty}),
  \end{equation}
a portion of which is shown in Figure~\ref{fig:8}, and it also applies to
prove that the short exact sequence 
  \begin{equation}\label{eq:97}
     0\longrightarrow \CZ\longrightarrow \ICZ^2(\CP^{\infty})\longrightarrow
     \zt\longrightarrow 0 
  \end{equation}
splits.  (One reads off~\eqref{eq:97} from the $E_{\infty}$~page.)

  \begin{figure}[ht]
\bigskip
  \centering
  \includegraphics[scale=.8]{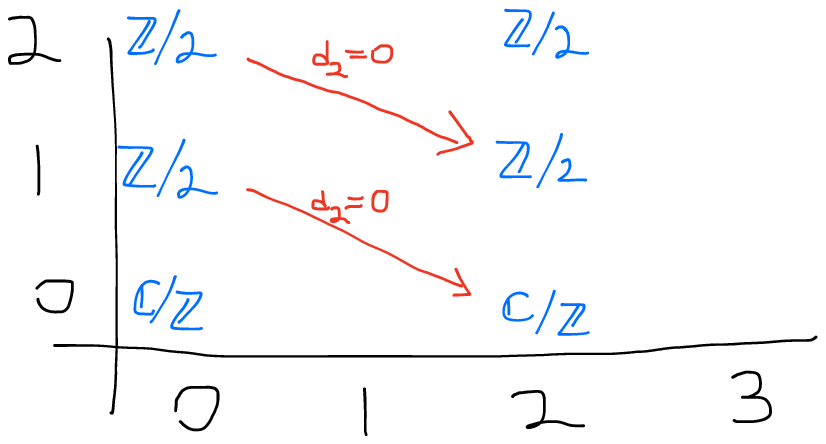}\phantom{MMMMMMMMMM}
  \caption{Computation of $\ICZ^2(\CP^{\infty})\cong \CZ\times
  \zt$}\label{fig:8} 
  \end{figure}
\bigskip

\subsection{$d=2$ bosonic theories: Kitaev  $E_8$ phase and chiral central
charge}\label{subsec:6.3}

To investigate SRE phases at the other extreme from those detected by group
cohomology, we set~$G=\{1\}$ to be the trivial group.  Then necessarily $\phi
\:G\to\pmo$ is the trivial homomorphism.  Unwinding Proposal~\ref{thm:24} we
have
  \begin{equation}\label{eq:76}
     \SREba 2{\{1\}}1=[\Sigma ^2MTSO_2,\Sigma ^4\IZ],
  \end{equation}
the group of $H$-type oriented theories with target the Anderson dual of the
sphere.  The following computations are carried out in Appendix~\ref{sec:8}.
Recall that the group~$\pi _4MSO$ is Thom's oriented bordism group of
4-manifolds, which is isomorphic to~$\ZZ$ via homomorphism which attaches to
each closed oriented 4-manifold~$M$ its signature~$\Sign(M)$.

  \begin{proposition}[]\label{thm:50}
 \ 

  \begin{enumerate}[{\textnormal(}i{\textnormal)}]

 \item $\pi _3\Sigma ^2MTSO_2=0$.

 \item $\pi _4\Sigma ^2MTSO_2\cong \ZZ$ and the composition 
  \begin{equation}\label{eq:170}
     \pi _4\Sigma ^3MTSO_2\longrightarrow \pi _4MSO \xrightarrow
     {\;\;\Sign\;\;}\ZZ  
  \end{equation}
maps the generators to~$\pm4$.

 \item $\pi _3\Sigma ^3MTSO_3=0$. 

 \item $\pi _4\Sigma ^3MTSO_3\cong \ZZ$ and the composition
  \begin{equation}\label{eq:148}
     \pi _4\Sigma ^3MTSO_3\longrightarrow \pi _4MSO \xrightarrow
     {\;\;\Sign\;\;}\ZZ   
  \end{equation}
maps the generators to~$\pm2$.

  \end{enumerate}
  \end{proposition}

We interpret these computations using~\eqref{eq:122}.  First, (i)~and
(ii)~imply 
  \begin{equation}\label{eq:171}
     [\Sigma ^2MTSO_2,\Sigma ^4I\ZZ]\cong \ZZ. 
  \end{equation}
Furthermore, the generating field theory extends to a $\ZZ$-valued invertible
4-dimensional oriented theory whose numerical invariant is \emph{4~times} the
signature.  Assertions~(iii) and~(iv) together imply 
  \begin{equation}\label{eq:149}
     [\Sigma ^3MTSO_3,\Sigma ^4\IZ]\cong \ZZ 
  \end{equation}
and the generating field theory extends to a $\ZZ$-valued invertible
4-dimensional theory whose numerical invariant is \emph{2~times} the
signature.  A 2-dimensional theory which generates the group~\eqref{eq:171}
does not extend to 3~dimensions, and a 3-dimensional theory which generates
the group~\eqref{eq:149} does not extend to 4~dimensions.  We relate these
factors to known geometric facts about $\eta $-invariants and determinant
lines.

The interpretation of the homotopical computations~\eqref{eq:76},
\eqref{eq:149} is not obvious at first glance; indeed that mystery was the
motivation for~\S\ref{subsec:2.7}.  Furthermore, the example discussed
there---in particular, the contractible choice replacing~$\Ft$ with~
$\Fg$---has a clear analog in our current situation.  If we include a
Riemannian metric as a field, then there is a 3-spacetime dimensional theory
whose invariant of a closed oriented Riemannian 3-manifold is the
exponentiated Atiyah-Patodi-Singer $\eta
$-invariant~\cite{APS}.\footnote{More precisely, for a self-adjoint
operator~$B$ the invariant is~$\exp\bigl(\pi i (\eta_B +h_B)\bigr)$, where
$h_B=\dim\ker B$.}  As in Remark~\ref{thm:28} there is a related
\emph{4-dimensional} theory (called~$\tia _8$ below) whose value on a closed
oriented 4-manifold~$M$ is a multiple of the signature~$\Sign(M)$.  If we use
the $\eta $-invariant associated to the signature operator, then the multiple
is~1.  But we can use instead the $\eta $-invariant associated to the
self-duality operator; the corresponding boundary operator~$B$ is $1/2$~that
of the signature operator.  (This $\eta $-invariant appears in quantum
Chern-Simons theory~\cite[(1.27)]{FG}.)  This theory represents a generator
of~\eqref{eq:149}; the multiple of the signature is~1/2.  The
invariants of 2-manifolds are determinant lines, and the determinant line of
the 2-dimensional signature operator on a closed oriented Riemannian manifold
has a natural $4^{\textnormal{th}}$~root: the determinant line of the
$\dbar$-operator.  Determinant lines of~$\dbar$ provide a generator
of~\eqref{eq:171}; the multiple of the signature is~1/4.

  \begin{remark}[]\label{thm:54}
 The discussion in~\S\ref{subsec:4.1}, especially Remark~\ref{thm:53}, and
also Remark~\ref{thm:17} are relevant here.  A theory classified
by~\eqref{eq:149} gives integer invariants of 4-manifolds whose stable
tangent bundle can be represented by a rank~3 vector bundle.  Such manifolds
have even signature.  (For example, the $4^{\textnormal{th}}$~Stiefel-Whitney
number vanishes.  It is the reduction modulo~2 of the Euler number which is
equal to the signature modulo~2.)  An example of such a 4-manifold is the
mapping cylinder of a 3-manifold, which fibers over the circle, but then the
signature vanishes.  A nontrivial example is the connected sum
$M=(\CP^2)^{\#2} \;\#\; (S^1\times S^3)^{\#3}$, which has signature~2 and
Euler number~0.  The vanishing Euler number implies that $M$~admits a nonzero
vector field, which splits an oriented line bundle off of~$TM$.  Similarly,
4~divides the signature of a compact oriented 4-manifold whose stable tangent
bundle is 2-dimensional.  In this case there are nontrivial examples which
are fiber bundles; the base and fiber are both compact oriented 2-manifolds.
The first examples are due to Atiyah~\cite{A3}; an example with signature
exactly~4 is constructed in~\cite[Theorem~1]{EKKOS}. 
  \end{remark}

We now argue that the SRE phase referred to in the literature as ``Kitaev's
$E_8$~phase'' or ``Kitaev's $E_8$~state'' (see~\cite{K5}, \cite{K2},
\cite{LV}) has the field theory whose invariant is exactly the signature as
its low energy approximation.

To make a first connection to $E_8$~Chern-Simons, we recall that the
gravitational Chern-Simons invariant enters into the quantization of
classical Chern-Simons theory as a counterterm~\cite[(2.20)]{W1}.  Its
appearance means that in general quantum Chern-Simons theory is
\emph{anomalous} as an oriented theory.  The anomaly is an invertible
4-dimensional theory
  \begin{equation}\label{eq:78}
     \alpha _{\bc}\:\Sigma ^4MTSO_4\longrightarrow \Sigma ^4I\Cx
  \end{equation}
whose invariant on a closed oriented 4-manifold~$M$ is 
  \begin{equation}\label{eq:79}
     e^{2\pi i\,c\Sign(M)/8}=e^{2\pi i\,c\,p_1(M)/24}. 
  \end{equation}
The anomaly depends only on the mod~8 reduction~$\bc$ of the \emph{chiral
central charge}~$c\in \RR$ of the corresponding 2-dimensional chiral
Wess-Zumino-Witten model.

  \begin{remark}[]\label{thm:29}
 Walker's approach~\cite{Wa} to quantum Chern-Simons theory uses bounding
4-manifolds to control the framing dependence.  In joint work with Constantin
Teleman (so far unpublished) we prove that a modular tensor category is
invertible as an object in the 4-category of braided tensor categories, and
we use it to define an invertible oriented 4-dimensional topological field
theory which is precisely~$\alpha _{\bc}$.  Note that the modular tensor
category determines~$\bc=c\pmod8$---see~\cite[(172)]{K5}, for example---but
it does not determine~$c\in \RR$.  The usual approach to quantum Chern-Simons
theory in the mathematics literature is to lift to a theory of manifolds with
a $(w_1,p_1)$-structure.  A $w_1$-structure is a trivialization of~$w_1$: an
orientation.  A $p_1$-structure is similar~\cite{BHMV}, but its geometric
avatars are not as simple as an orientation.  For example, a $p_1$-structure
on a 3-manifold can be given by a ``2-framing''~\cite{A2}.  In this way one
obtains Chern-Simons as an extended theory of 1-, 2-, and 3-dimensional
manifolds, but to do so one needs to lift $c\pmod8$ to $c\pmod{24}$.  Of
course, given $c\in \RR$ there is a preferred choice, but starting from a
modular tensor category there are 3~choices.
  \end{remark}

We observe that given a chiral central charge~$c\in \RR$ there is a
4-dimensional invertible theory 
  \begin{equation}\label{eq:80}
     \tia_c\:\Sigma ^4MTSO_4\longrightarrow \Sigma ^4H\RR 
  \end{equation}
whose invariant on a closed oriented 4-manifold~$M$ is the real number
  \begin{equation}\label{eq:81}
     c\,\Sign(M)/8. 
  \end{equation}
The anomaly theory~$\alpha _{\bc}$ in~\eqref{eq:78} is obtained by
composing~\eqref{eq:80} with the map $\Sigma ^4H\RR\to\Sigma ^4I\Cx$ induced
by the exponential map $e^{2\pi i(-)}\:\RR\to\Cx$; see~\eqref{eq:121}.  If
$c=8n$ for some~$n\in \ZZ$, then \eqref{eq:80}~ factors through an integral
theory
  \begin{equation}\label{eq:82}
     \tia_{8n}\:\Sigma ^4MTSO_4\longrightarrow \Sigma ^4\IZ
  \end{equation}
These integral topological theories are not part of the usual quantum
Chern-Simons theory: only the exponential~\eqref{eq:78} of~\eqref{eq:80}
occurs (as the ``framing anomaly'' theory).  For the theories
in~\eqref{eq:82} the framing anomaly is trivial.  What does occur in
Chern-Simons is the invertible metric 3-dimensional theory whose partition
function is the exponentiated $\eta $-invariant to a suitable
power~\cite{W1}.  $E_8$~Chern-Simons at level~1, or Chern-Simons theory for
the maximal torus of~$E_8$ (with its Cartan matrix specifying the level, or
``$K$-matrix''), has chiral central charge~$c=8$.  The $\eta $-invariant
which occurs is associated to the signature operator, and this is the class
of theories we associate to Kitaev's $E_8$~phase. 

  \begin{remark}[]\label{thm:55}
 Proposition~\ref{thm:50} implies that there are additional possibilities for
an $H$-type 2-dimensional theory: a ``$4^{\textnormal{th}}$~root'' of the
effective theory of Kitaev's $E_8$~phase.  Such theories are associated with
chiral central charge~$c=2$, and in general the 2-dimensional $H$-type
theories are associated with chiral central charge divisible by~2.  This
matches the conformal anomaly in 2-dimensional conformal field
theory---see~\cite[(5.9)]{S1}, for example: if the chiral central charge is
not divisible by~2, then a $p_1$-structure is needed to define the theory.
Even if we require the theory to extend to 3-manifolds, there is still the
possibility of dividing by~2, so having chiral central charge divisible by~4.
So it appears that our proposal allows for more effective topological
theories than have been seen so far by SRE phases.
  \end{remark}

  \subsection{$d=2$ bosonic theories: mixed gravitational/gauge phases}\label{subsec:6.4}

Continuing with $d=2$ space dimensional theories, we give an example to
illustrate the unitarizability restriction in~\eqref{eq:146}.  Now we allow a
global symmetry group~$G$.

We focus on the Kunneth component
  \begin{multline}\label{eq:83}
     H^2(BSO_2;\ZZ)\otimes H^2(BG;\ZZ_\phi )\subset H^4(BSO_2\times
     BG;\ZZ_\phi ) \\\longrightarrow \IZ^{\tau \mstrut _{\IZ}+\phi\mstrut 
     _{\IZ}+4}(BSO_2\times BG) = \SPTba 2G\phi . 
  \end{multline}
The group~$H^2(BSO_2;\ZZ)$ is infinite cyclic with generator the Euler
class~$e$.  For simplicity let $G$~be finite and $\phi $~the trivial
homomorphism.  Then $H^2(BG;\ZZ)\cong H^1(BG;\Cx)$ is isomorphic to the group
of abelian characters $\chi \:G\to\Cx$.  Fix one and let $\lambda _\chi \in
H^2(BG;\ZZ)$ be the corresponding class.  Let $F$~be the theory which
corresponds to~$e\otimes \lambda _\chi $ in~\eqref{eq:83}.  We remark in
passing that $F$~does \emph{not} extend to a 3-spacetime dimensional theory:
the Euler class~$e$ is not the restriction of a class in~$H^2(BSO_3;\ZZ)=0$.
The main point: this theory is \emph{not} unitarizable.  To see this it
suffices to restrict along $H\ZZ\longrightarrow I\ZZ$ in~\eqref{eq:146},
since we are trying to hit~$e\otimes \lambda _\chi $ which lies in ordinary
cohomology.  There is an isomorphism~${\bar\tau \mstrut _{H\ZZ}}\cong w
\mstrut _{H\ZZ}$ of twistings in the domain of~\eqref{eq:146}, so the
relevant group is $H^4(BO_2\times BG;\ZZ_\phi )$ and the relevant Kunneth
component is $H^2(BO_2;\ZZ)\otimes H^2(BG;\ZZ_\phi )$.  The Euler class~$e$
does \emph{not}\footnote{It does drop to a class on~$BO_2$ with twisted
coefficients, but that is not relevant here.} drop to a class
in~$H^2(BO_2;\ZZ)$.

  \begin{remark}[]\label{thm:51}
 It is instructive to compute something nontrivial in the theory~$F$.  Let
$Y$~be a closed oriented 2-manifold and $P\to Y$ a principal $G$-bundle.
Then $F(P\to Y)$ is a complex line.  Suppose $\varphi $~is an automorphism of
$P\to Y$, so a map
  \begin{equation}\label{eq:84}
     \begin{gathered} \xymatrix{P\ar[r]^{\varphi } \ar[d]_{} & P\ar[d]^{} \\
     Y\ar[r]^{\bar\varphi } & Y} \end{gathered} 
  \end{equation}
of principal $G$-bundles covering an orientation-preserving diffeomorphism
of~$Y$.  Gluing the ends of $[0,1]\times P\to[0,1]\times Y$ using~$\varphi $
we obtain a principal $G$-bundle $Q_\varphi \to X_\varphi $ over the mapping
cylinder~$X_\varphi $.  Note that the mapping cylinder is a 3-manifold which
is the total space of a fiber bundle $X_\varphi \to \cir$ with typical
fiber~$Y$.  The rank~2 relative tangent bundle $T(X_\varphi /\cir)\to
X_{\varphi}$ is oriented so has an Euler class $e(X_\varphi /\cir)\in
H^2(X_\varphi ;\ZZ)$.  The $G$-bundle $Q_\varphi \to X_\varphi $ has a
characteristic class $\lambda _\chi (Q_\varphi )\in H^1(X_\varphi ;\Cx)$.
Then the action of~$\varphi $ on the line $F(P\to Y)$ is multiplication by
  \begin{equation}\label{eq:85}
     \langle e(X_\varphi /\cir)\smile \lambda _\chi (Q_\varphi ),[X_\varphi ]
     \rangle\in \Cx. 
  \end{equation}
In~\eqref{eq:85} we pair the cup product of the characteristic classes with
the fundamental class of the oriented 3-manifold~$X_\varphi $.  A special case
of note: $\bar\varphi $~is the identity diffeomorphism, $P\to Y$ is the trivial
bundle, and the gauge transformation~ $\varphi $ is given by an element~$g\in G$
(assuming $Y$~is connected).  Then \eqref{eq:85}~reduces to $\chi
(g)^{\Euler(Y)}$, where $\Euler(Y)\in \ZZ$ is the Euler number. 
  \end{remark}

  \subsection{$d=3$ bosonic theories: time-reversal symmetry}\label{subsec:6.5}

Set $G=\pmo$ and $\phi \:\pmo\to\pmo$ the identity map.  The group cohomology
captures a subgroup of the group of SRE phases:
  \begin{equation}\label{eq:86}
     H^5(B\pmo;\ZZ_\phi )=H^5(\RP^{\infty};\ZZ_\phi )\cong \zt, 
  \end{equation}
as appears in~\cite{CGLW}.  Another nontrivial SRE phase of order~2 was
introduced in~\cite{VS} where it was emphasized that this goes beyond the
group theory computation.  This SRE phase is predicted by
Proposal~\ref{thm:24}.

  \begin{proposition}[]\label{thm:52}
We have 
  \begin{equation}\label{eq:87}
     \SREba3{\pmo}{\id} = \IZ^{\tIZ+\fIZ+5}(BSO_3\times \RP^{\infty}) \cong
     \zt\times \zt.
  \end{equation} 
Furthermore, the map 
  \begin{equation}\label{eq:150}
     i \:H^5(BSO_3\times \RP^{\infty};\ZZ_\phi )\longrightarrow
     \IZ^{\tIZ+\fIZ+5}(BSO_3\times \RP^{\infty}) 
  \end{equation}
is surjective and 
  \begin{equation}\label{eq:151}
     H^5(BSO_3\times \RP^{\infty};\ZZ_\phi )\cong \zt\times \zt\times \zt. 
  \end{equation}
  \end{proposition}

\noindent 
 We defer most of the proof to the appendix; here we briefly sketch two ways
to compute~\eqref{eq:151}.
 
The first is a direct approach.  Use the chain complex\footnote{This is the
\emph{minimal chain complex}~\cite[Proposition~3E.3]{Ha} derived from the
homology of~$BSO_3$.  }
  \begin{equation}\label{eq:88}
     \ZZ\xleftarrow{\;\;0\;\;} 0\xleftarrow{\;\;0\;\;}
     \ZZ\xleftarrow{\;\;2\;\;} \ZZ\xleftarrow{\;\;0\;\;}
     \ZZ\xleftarrow{\;\;0\;\;} \ZZ\xleftarrow{\;\;2\;\;} \ZZ \cdots 
  \end{equation}
for~$BSO_3$ and the chain complex 
  \begin{equation}\label{eq:89}
     \ZZ\xleftarrow{\;\;2\;\;} \ZZ\xleftarrow{\;\;0\;\;}
     \ZZ\xleftarrow{\;\;2\;\;} \ZZ\xleftarrow{\;\;0\;\;}
     \ZZ\xleftarrow{\;\;2\;\;} \ZZ\xleftarrow{\;\;0\;\;} \ZZ \cdots 
  \end{equation}
for $\RP^{\infty}$ with the nontrivial local system $\ZZ_\phi
\to\RP^{\infty}$.  Compute the cohomology of the cochain complex obtained by
applying $\Hom(-,\ZZ)$ to the tensor product of~\eqref{eq:88}
and~\eqref{eq:89}.  An alternative approach is to apply the Kunneth formula
for cohomology~\cite[\S5.5]{Sp}, which in this case gives a split short exact
sequence
  \begin{multline}\label{eq:90}
     0\longrightarrow \bigl[H^{\bullet }(BSO_3;\ZZ)\otimes H^{\bullet
     }(\RP^{\infty};\Zf) \bigr]^5 \longrightarrow H^5(BSO_3\times
     \RP^{\infty};\Zf) \\\longrightarrow \bigl[H^{\bullet
     }(BSO_3;\ZZ)*H^{\bullet }(\RP^{\infty};\Zf) \bigr]^6 \longrightarrow 0 
  \end{multline}
Here `$*$'~denotes the torsion product of abelian groups.  The tensor product
in the kernel of~\eqref{eq:90} is isomorphic to~$\zt\times \zt$ generated by
the nonzero class in~$H^5(\RP^{\infty};\Zf)\cong \zt$ and the tensor product
of $p_1\in H^4(BSO_3;\ZZ)$ and the nonzero class $a\in
H^1(\RP^{\infty};\Zf)$.  The quotient group in~\eqref{eq:90} is isomorphic
to~$\zt$, which is the torsion product
$H^3(BSO_3;\ZZ)*H^3(\RP^{\infty};\Zf)$.

  \begin{claim}[]\label{thm:32}
 The image~$i(p_1\otimes a)$ of $p_1\otimes a$ under~\eqref{eq:150} ~is the
long-range effective topological theory of the SRE phase identified
in~\cite{VS}.   
  \end{claim}

\noindent
 We argue for this claim in~\S\ref{sec:7}.

  \begin{remark}[]\label{thm:56}
 Since $i(p_1\otimes a)$~is torsion we can lift it to the group
$\ICZ^{\tIZ+\fIZ+4}(BSO_3\times \RP^{\infty}) $ where it is easier to
identify a particular theory in this class; see Remark~\ref{thm:44}. 
  \end{remark}

   \section{Boundary conditions and long-range topological field theories}\label{sec:7}

We begin in~\S\ref{subsec:7.1} with a general discussion of spatial boundary
conditions for field theories and condensed matter systems.  We specialize to
the case at hand: the bulk theory is gapped and the long-range topological
theory is invertible.  We apply these general ideas in~\S\ref{subsec:7.2} to
argue for Claim~\ref{thm:32}, which locates the 3d $E_8$~phase with
half-quantized surface thermal Hall effect introduced in~\cite{VS} and
further investigated in~\cite{BCFV}.  We recover some key aspects of that
theory from our topological viewpoint.

  \subsection{Boundary terminations}\label{subsec:7.1}

Suppose we are given a theory~$F$ in $n$ spacetime dimensions.  It may be a
quantum field theory or a condensed matter theory.  Then to a compact
$(n-1)$-manifold~$Y$ \emph{with empty boundary} we obtain a complex vector
space of states.  We would like to extend to allow compact
$(n-1)$-dimensional manifolds~$Y$ which have nonempty boundary.  These
boundaries are \emph{spatial}, not \emph{temporal}.  In this case we expect
to impose boundary conditions~$\beta $ which essentially close off the
boundary.  In other words, the pair~$(Y,\beta )$ behaves as a closed
$(n-1)$-manifold for the pair of theories~$(F,\beta )$, and the $(F,\beta
)$~theory attaches to it a vector space of states.  (Without the boundary
condition~$\beta $ we expect instead a module for an algebra more complicated
than~$\CC$, or an object in a category more complicated than~$\Vect$.)
Furthermore, we expect~$\beta $ to be \emph{local}.  In classical physics a
spatial boundary condition is typically a \emph{local constraint} on fields:
a boundary condition for a system of classical partial differential
equations.  In quantum physics a spatial boundary condition is a
\emph{relative field theory} (\S\ref{subsec:2.3})
  \begin{equation}\label{eq:103}
     \beta \:\tau \mstrut _{\le \,n-1}F\longrightarrow\bo . 
  \end{equation}
The theory~$F$ evaluates on~$Y$ to a map $F(Y)\:\Vect\to F(\partial Y)$ and
the boundary condition~$\beta $ evaluates on~$\partial Y$ to a map $\beta
(\partial Y)\:F(\partial Y)\to\Vect$.  The composition $\beta (\partial
Y)\circ F(Y)\:\Vect\to\Vect$ is tensor product with the vector space
associated to~$Y$ in the theory~$(F,\beta )$.  If $F$~is an invertible field
theory, then $\beta $~is an anomalous theory with anomaly ~$F$.  We call
~$(F,\beta )$ a \emph{bulk-boundary pair}.

  \begin{remark}[Vocabulary]\label{thm:36}
 A quantum boundary condition~$\beta $ in quantum field theory is sometimes
called a \emph{D-brane}, a term most appropriate in the context of
2-spacetime dimensional conformal field theories.  In condensed matter
physics $\beta $~ goes by a name like \emph{edge termination} or
\emph{surface termination}, depending on the dimension of the theory.
Sometimes the word `excitation' is used in place of `termination'.
  \end{remark}

  \begin{remark}[]\label{thm:38}
 If $F$~is a \emph{topological} field theory with values in an $(\infty
,n)$-category~$\sC$, then a boundary condition is a 1-morphism $F(\pt)\to
\bo$ in~$\sC$.  The dual\footnote{The choice of direction of the arrow
in~\eqref{eq:103} reflects our choice that $\partial Y$~is outgoing rather
than incoming.  There is an equivalent exposition with the other choice, and
we would not need the dual here.} map $\bo\to F(\pt)$ may be considered as an
``object in~$F(\pt)$''.  For example, if $n=2$ and $\sC$~is a 2-category of
categories, then $\beta $~is literally an object in the 1-category~$F(\pt)$.
Boundary conditions in topological theories are a special case of a much more
general construction~\cite[Example 4.3.22]{L}.
  \end{remark}

  \begin{remark}[]\label{thm:37}
 If $F$~ is a $d$-space dimensional theory of $H$-type, then \eqref{eq:103}~
is replaced by
  \begin{equation}\label{eq:104}
     \beta \:\tau \mstrut _{\le \,d-1}F\longrightarrow\bo . 
  \end{equation}
  \end{remark}

There is not a unique boundary condition for a given~$F$, but rather
$F$~determines a collection of boundary conditions.  These formal
considerations can lead to physical consequences.  One important example in
condensed matter physics is the \emph{integer quantum Hall effect};
see~\cite{W4} for an account aimed at mathematicians.  

Suppose the theory~$F$ is the long-range topological approximation to a
gapped $d$-space dimensional system of $H$-type.  Then if a boundary
condition produces a combined bulk-boundary pair which is still gapped, we
expect that the long-range topological approximation is a bulk-boundary
pair~$(F,\beta )$ of topological theories.  If $F$~describes a short-range
entangled phase---that is, $F$~is invertible---then $\beta $~is a
$(d-1)$-space dimensional anomalous theory with anomaly~$F$.  We implicitly
assume that the truncation~$\tau \mstrut _{\le\,d-1}F$ of~$F$ is nontrivial.
If, furthermore, $F$~describes an SPT~phase, then we arrive at the following
trichotomy.
  \begin{equation}\label{eq:105}
  \begin{aligned}
     \textit{A long-range}&\textit{  effective boundary
     condition:}\\\textrm{(i)}&\textit{ produces a 
      gapless bulk-boundary pair which preserves the symmetry,}\\
     \textrm{(ii)}&\textit{ is non-anomalous and 
     breaks the symmetry, or}\\ \textrm{(iii)}&\textit{ is anomalous,
     symmetric, and 
     exhibits long-range entanglement.}
  \end{aligned}
  \end{equation}
Possibility~(ii) arises since by the definition of an SPT phase $F$~restricts
to a trivial theory when the symmetry is broken, and a theory relative to the
trivial theory is non-anomalous.  For~(iii) we observe that an invertible
theory relative to an invertible theory~$\trF$ is a trivialization of~$\trF$,
so if $\trF$~is not trivial then any relative theory must not be invertible.
In the physics lingo it is not short-range entangled but rather is long-range
entangled---it ``exhibits topological order''.  The trichotomy~\eqref{eq:105}
is a restatement of an assertion in the introduction to~\cite{VS}.

\newpage
  \subsection{The invertible field theory of an exotic $d=3$ bosonic phase}\label{subsec:7.2}

We turn now to the SPT phase identified as $i(p_1\otimes a)$ at the end of
\S\ref{subsec:6.5}.  

The argument that $i(p_1\otimes a)$ corresponds to the 3d $E_8$~phase with
half-quantized surface thermal Hall effect is based on~(iii) in the
trichotomy~\eqref{eq:105}.  Consider a long-range effective boundary
condition---surface termination---which is time-reversal symmetric and
anomalous with anomaly~$i(p_1\otimes a)$.  This boundary condition is a
relative 3-spacetime dimensional topological theory.  We claim that any
Chern-Simons theory with chiral central charge
  \begin{equation}\label{eq:111}
     c\equiv 4\pmod8 
  \end{equation}
is such an effective boundary condition: it satisfies~(iii) in the
trichotomy.  (See~\S\ref{subsec:6.3} for a topological discussion of chiral
central charge.)  One of the simplest examples is Chern-Simons theory for the
maximal torus of~$SO_8$, which was proposed for this role in~\cite[\S
VII]{VS} and was realized as a boundary condition in the exactly soluble
Hamiltonian constructed in~\cite{BCFV}.
 
To justify the claim begin with the short exact coefficient sequence
  \begin{equation}\label{eq:112}
     0\longrightarrow \Zt\longrightarrow \tfrac 12\Zt\longrightarrow \tfrac
     12\Zt/\Zt \longrightarrow 0 
  \end{equation}
in which the first map is the inclusion.  Identify $\tfrac 12\Zt/\Zt\cong
\tfrac 12\ZZ/\ZZ\cong \zt$, which in particular is untwisted, and so
write\footnote{This maneuver allows us to write the torsion class~$p_1\otimes
a$ as a class in one lower degree, and so identify it with particular field
theories.  This circumvents the issues raised in~\S\ref{subsec:2.7} about
nontorsion classes; see Remark~\ref{thm:44}.} $p_1\otimes a$ as the image of
  \begin{equation}\label{eq:113}
     \tfrac 12 p_1\!\!\pmod \ZZ\; \in H^4(BSO_3;\tfrac 12\ZZ/\ZZ) \subset
     H^4(BSO_3\times \RP^{\infty};\tfrac 12\ZZ/\ZZ) 
  \end{equation}
under the connecting homomorphism in the long exact sequence deduced
from~\eqref{eq:112}.  The image~$i$ in the Brown-Comenetz dual of the sphere
is computed via the sequence of maps 
  \begin{multline}\label{eq:152}
     H^4(BSO_3;\tfrac 12\ZZ/\ZZ)\cong H^4(\Sigma ^4MTSO_3;\tfrac
     12\ZZ/\ZZ)\\\longrightarrow H^4(\Sigma ^4MTSO_3;\CZ)\cong [\Sigma
     ^4MTSO_3,\Sigma ^4H\CZ]\\\longrightarrow [\Sigma ^4MTSO_3,\Sigma
     ^4I\CZ]\cong \CZ. 
  \end{multline}
The last isomorphism follows from Proposition~\ref{thm:50}(iv).  Since
\eqref{eq:113}~has order~2, so does its image, and checking
against~\eqref{eq:79} we identify it with the anomaly theory~$\alpha _{\bar
c}$ with~$\bar c \equiv 4\pmod 8$.  As explained in~\S\ref{subsec:6.3},
$\alpha _{\bar c=4}$~ is the (framing) anomaly of any quantum Chern-Simons
theory whose chiral central charge satisfies~\eqref{eq:111}.

\appendix

   \section{Some homotopy groups of Madsen-Tillmann spectra}\label{sec:8}
 
We prove Proposition~\ref{thm:50} and Proposition~\ref{thm:52}.  I thank
Oscar Randal-Williams for sharing his expertise, for correcting a mistake in
a previous version of Proposition~\ref{thm:50}, and for providing a few
arguments in the proof.

First recall some facts about Madsen-Tillmann spectra.  Set $X_n=\Sigma
^nMTSO_n$.  Then $X_1\simeq S^0$.  Let $\Sigma ^{\infty}Y$~denote the
suspension spectrum of a pointed space~$Y$.  The fibration
  \begin{equation}\label{eq:153}
     X_{n-1}\longrightarrow X_n\longrightarrow \Sigma
     ^n\Sigma ^{\infty}(BSO_n)_+  
  \end{equation}
is proved in \cite[Proposition~3.1]{GMTW} and \cite[Lemma~3.8]{FHT}.  For any
space~$Y$ we have
  \begin{equation}\label{eq:154}
     \Sigma ^{\infty}(Y_+)\simeq S^0\vee \Sigma ^{\infty}Y. 
  \end{equation}
The stable homotopy groups of a pointed space~$Z$ are $\pi _j^sZ=\pi \mstrut
_j\Sigma ^{\infty}Z$.  The oriented version of~\eqref{eq:44} expresses the
Thom spectrum~$MSO$ as the colimit of a sequence of maps $X_1\to X_2\to
X_3\to\cdots$.  The homotopy groups of~$MSO$ are Thom's oriented bordism
groups.  The long exact sequence of homotopy groups deduced
from~\eqref{eq:153} implies that
  \begin{equation}\label{eq:155}
     \pi _jX_n\xrightarrow{\;\;\cong \;\;}\pi _jMSO,\qquad j<n, 
  \end{equation}
is an isomorphism and that there is an exact sequence 
  \begin{equation}\label{eq:156}
     \pi _{n+1}X_{n+1}\xrightarrow{\;\;\chi \;\;}\ZZ\longrightarrow \pi
     _nX_n\longrightarrow \pi _nMSO\longrightarrow 0. 
  \end{equation}
An element of~$\pi _{n+1}X_{n+1}$ is represented by a closed oriented
$(n+1)$-manifold~$W$, and its image under $\chi$ is the Euler number
of~$W$.  For $n$~even the map~$\chi $ is zero.  For~$n=3$ the map~$\chi $ is
surjective, since $\chi (\CP^2\;\#\;S^1\times S^3)=1$.  

  \begin{proof}[Proof of Proposition~\ref{thm:50}]
 First apply~\eqref{eq:156} with~$n=3$ to derive the exact sequence 
  \begin{equation}\label{eq:172}
     \pi _4X_3\longrightarrow \pi _4X_4\xrightarrow{\;\;\chi \;\;} \ZZ
     \longrightarrow \pi_3X_3\longrightarrow \pi _3MSO 
  \end{equation}
As remarked above the Euler characteristic map~$\chi $ is onto, and since
$\pi _3MSO=0$ we deduce $\pi _3X_3=0$, which is (iii) in the proposition.
Next, apply~\eqref{eq:156} with~$n=4$ to deduce $\pi _4X_4\cong \ZZ\times
\ZZ$ and the composition $\pi _4X_4\to \pi _4MSO\xrightarrow{\;\Sign\;}\ZZ$
is surjective.  Then a stretch of the long exact sequence of homotopy
groups deduced from~\eqref{eq:153} with~$n=4$ is
  \begin{equation}\label{eq:159}
     \begin{gathered} \xymatrix@R=.7em{\pi ^s_1(BSO_4)_+\ar[r]\ar@{=}[d] & 
     \pi _4X_3\ar[r]& \pi _4X_4\ar[r]^<<<<<\chi \ar@{=}[d] & \pi _0^s
     (BSO_4)_+\ar[r] \ar@{=}[d] & \pi \mstrut _3X\mstrut _3\ar@{=}[d]\\ \zt
     && \ZZ\times \ZZ & \ZZ & 0} \end{gathered} 
  \end{equation}
from which $\pi _4X_3 \cong \ZZ\textnormal{ or }\ZZ\times \zt$.  To see that
it is the former, let $F$~be the fiber of the spectrum map $X_3\to H\ZZ$
which represents the generator of~$H^0(X_3;\ZZ)\cong \ZZ$.  So there is a
cofiber sequence 
  \begin{equation}\label{eq:175}
     F\longrightarrow  X_3\longrightarrow  H\ZZ. 
  \end{equation}
Then \eqref{eq:155}~and the vanishing of~$\pi _3X_3$ imply that $\pi
_jF=0,\,j\le3$, whence the Hurewicz map $\pi _4F\to H_4F$ is an isomorphism.
In addition, the map $\pi _4F\to \pi _4X_3$ is an isomorphism.
Figure~\ref{fig:11} is a schematic depiction of the long exact sequence of
  \begin{figure}[ht]
  \centering
  \includegraphics[scale=.8]{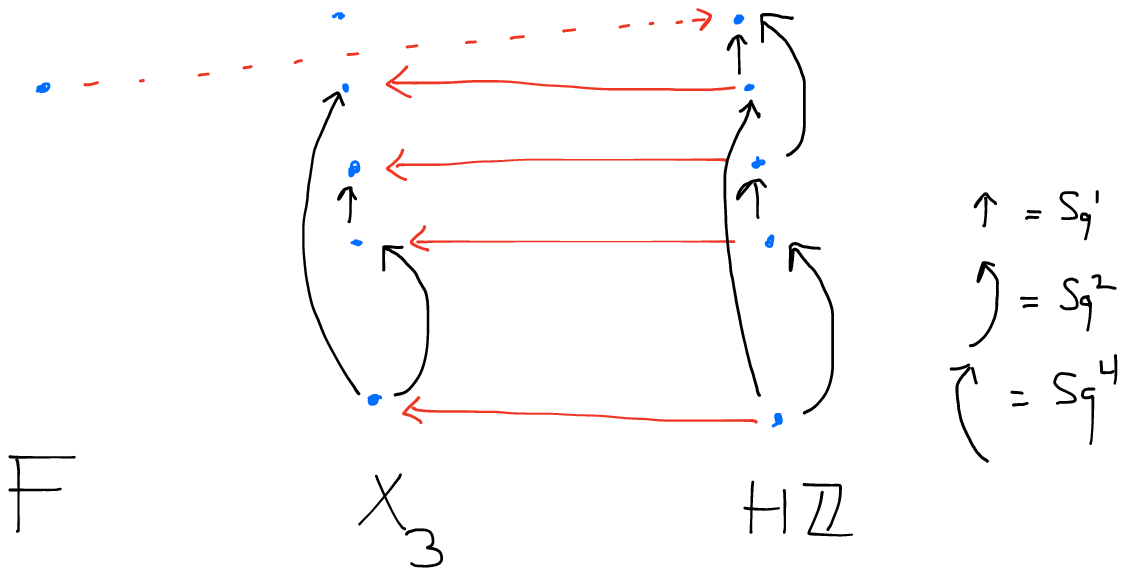}
  \caption{Long exact sequence in $\ft$-cohomology induced by $F\to X_3\to
  H\ZZ$}\label{fig:11}
  \end{figure}
$\ft$-cohomology groups induced by the cofibration~\eqref{eq:175}, where
$\ft=\zt$ is the field of 2~elements.  The $\ft$-cohomology of a spectrum is
a $\ZZ$-graded $\ft$-vector space which is a module for the Steenrod algebra.
The dots indicate basis elements and the vertical arrows the action of the
Steenrod operations~$Sq^1,Sq^2$.  The degrees in the figure ascend from~0
to~5.  The $\ft$-cohomology of~$H\ZZ$ is isomorphic to the group of
cohomology operations $H\ZZ\to H\zt$ and is computed by a theorem of Serre.
The generators are the operations~$Sq^2,Sq^3, Sq^4, Sq^5$ (preceded by
reduction modulo~2), and the action of the Steenrod operations is given by
the Adem relations.  For~$X_3$ the generators are $w_2u, w_3u, w_2^2u,
w_2w_3u$ where $u\in H^0(X_3;\ft)$ is the mod~2 Thom class of the virtual
bundle~\eqref{eq:145}, which stably is \emph{minus} the canonical rank~3
bundle ~$S(3)\to BSO_3$.  Its total Stiefel-Whitney class is  
  \begin{equation}\label{eq:180}
     \tilde w=\frac{1}{1+w_2+w_3}= 1 + w_2 + w_3 + w_2^2 + \cdots, 
  \end{equation}
the inverse of the Stiefel-Whitney class of $S(3)\to BSO_3$.  The action of
the total Steenrod operation~$Sq=1 + Sq^1 + Sq^2 +\cdots$ is $Sq(u)=\tilde
wu$.  The horizontal arrow in degree~0 follows from the definition of $X_3\to
H\ZZ$, and the arrows in degrees~2,3,4 from the module structure, as does the
lack of a horizontal arrow in degree~5.  Exactness then implies the existence
of a class in~$H^4(F;\ft)$ which maps to~$Sq^5\in H^5(H\ZZ;\ft)$.  Thus
  \begin{equation}\label{eq:176}
     \ft\cong H^4(F;\ft)\cong \Hom(H_4F,\ft)\cong \Hom(\pi _4F,\ft)\cong
     \Hom(\pi _4X_3,\ft) 
  \end{equation}
and we conclude $\pi _4X_3\not\cong \ZZ\times \zt$, whence $\pi _4X_3\cong
\ZZ$.  For the last claim in~(iv) we revisit~\eqref{eq:172}.  A 4-manifold
which represents a class in~$\pi _4X_3$ has vanishing~$w_4$, since $w_4$~is a
stable characteristic class and vanishes for rank~3 bundles.  Thus its Euler
number is even, and since the Euler number and signature are congruent
modulo~2 its signature is also even.  Then observe that the 4-manifold
$(\CP^2)^{\#2} \;\#\; (S^1\times S^3)^{\#3}$ represents an element of~$\pi
_4X_3$ (since it has vanishing Euler characteristic, so a nonvanishing vector
field which splits a line bundle off its tangent bundle) and has signature~2.

Similar techniques prove~(i) and~(ii).  (An alternative is to use the
Madsen-Weiss theorem~\cite{MaWe} and known facts about the stable homology of
mapping class groups of surfaces.)  First 
  \begin{equation}\label{eq:177}
     \pi \mstrut _{\{0,1,2\}}X_2\cong \{\ZZ\,,\,0\,,\,\ZZ\} 
  \end{equation}
from~\eqref{eq:155} and \eqref{eq:156} with~$n=2$.  There is a nontrivial
$k$-invariant connecting these homotopy groups; if not, then
$H^2(X_2;\ft)\cong H^2(BSO_2;\ft)$~would be 2-dimensional.  Let $C$~denote
the spectrum with these two nonzero homotopy groups and nontrivial
$k$-invariant.  Its $\ft$-cohomology 
  \begin{figure}[ht]
  \centering
  \includegraphics[scale=.8]{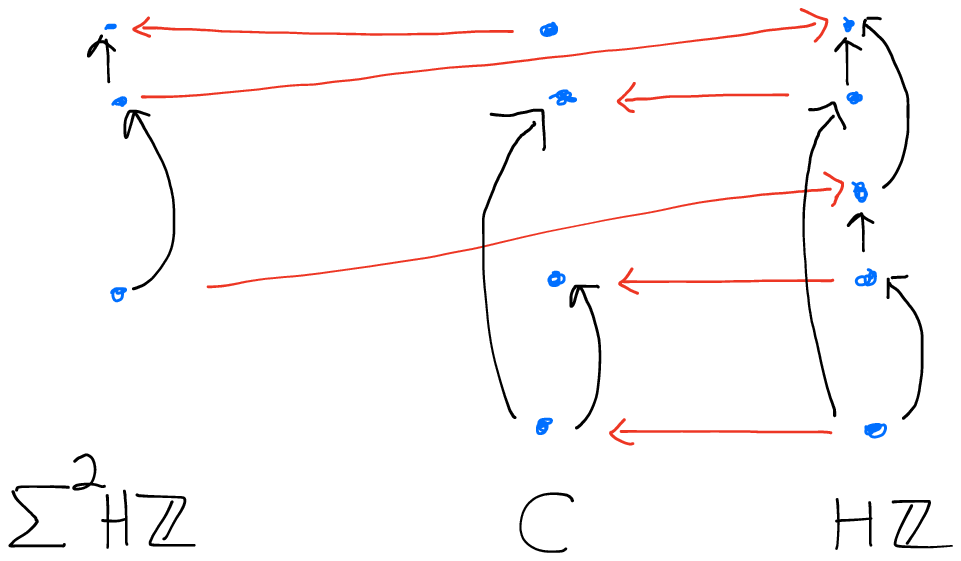}
  \caption{Long exact sequence in $\ft$-cohomology induced by $\Sigma ^2H\ZZ\to C\to
  H\ZZ$}\label{fig:12} 
  \end{figure}
is worked out in Figure~\ref{fig:12} using the cofiber sequence $\Sigma
^2H\ZZ\to C\to H\ZZ$.  All cohomology in degree~1 vanishes.  Also, the
nontrivial connecting map $H^2(\Sigma ^2H\ZZ;\ft)\to H^3(H\ZZ;\ft)$ is the
$k$-invariant.  Let $F'$ be the fiber of the Postnikov map $X_2\to C$, and
consider the cofiber sequence $F'\to X_2\to C$.
  \begin{figure}[ht]
  \centering
  \includegraphics[scale=.6]{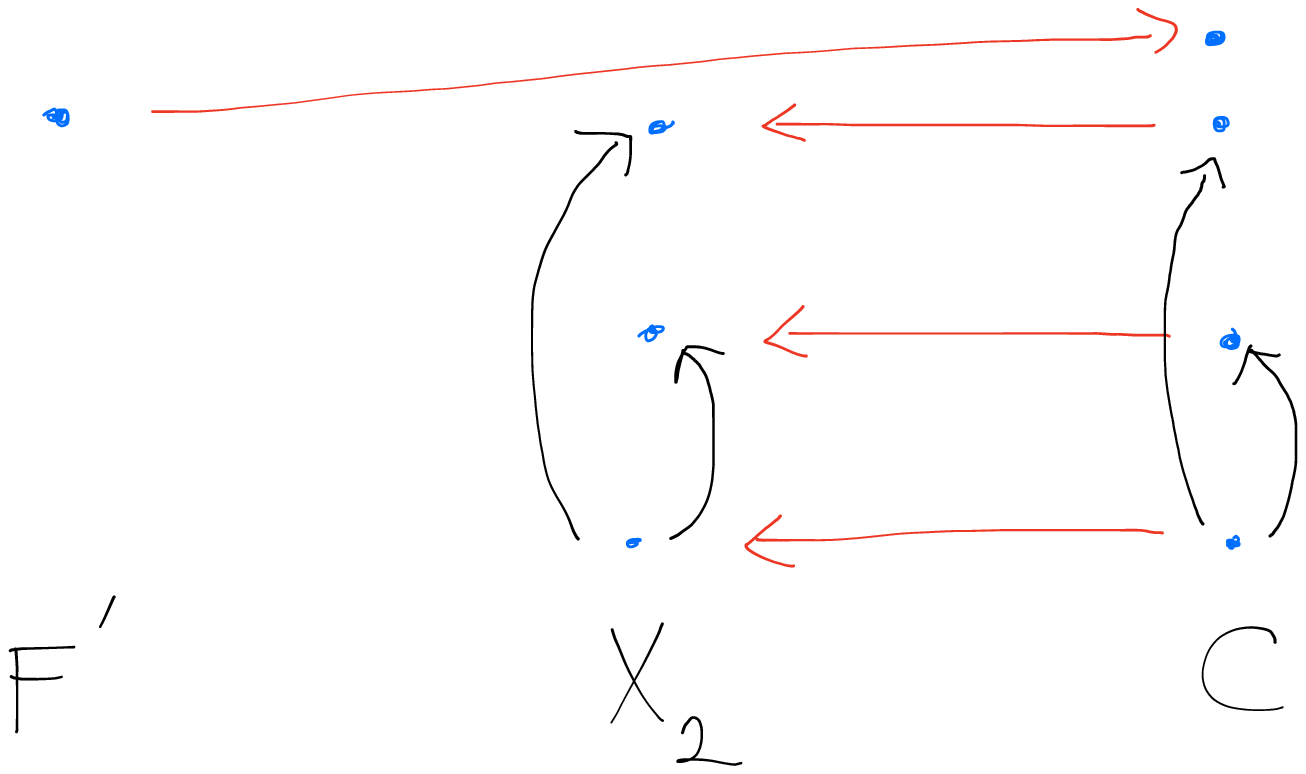}
  \caption{Long exact sequence in $\ft$-cohomology induced by $F'\to X_2\to
  C$}\label{fig:13}
  \end{figure}
The induced maps on $\ft$-cohomology are worked out in Figure~\ref{fig:13};
the nonzero cohomology is in degrees~0,2,4,5.  We deduce 
  \begin{equation}\label{eq:178}
     H^4(F';\ft)\cong \ft, 
  \end{equation}
and from Hurewicz $\pi _{\le3}F'=0$ and $\pi _4F'\to H_4F'$ is an
isomorphism.  Now the long exact sequence of homotopy groups induced from ~
\eqref{eq:153} with~$n=2$ includes the stretch
  \begin{equation}\label{eq:157}
     \begin{gathered} \xymatrix@R=.7em{\pi _4S^0\ar[r]\ar@{=}[d]&\pi
     _4X_2\ar[r] &  
     \pi _2^s(BSO_2)_+\ar[r]\ar@{=}[d]& \pi _3S^0 \ar@{=}[d] \\ 
     0&&\ZZ\times \zt& \zmod{24}} \end{gathered} 
  \end{equation}
This implies $\pi _4X_2\cong \ZZ \textnormal{ or }\ZZ\times \zt$;
\eqref{eq:178}~rules out the latter since $\pi _4F'\cong \pi _4X_2$ and
$H^4(F';\ft) \cong \Hom(H _4F',\ft)\cong \Hom(\pi _4F',\ft)$.  For the last
statement in~(ii) consider the long exact sequence of homotopy groups induced
from~\eqref{eq:153} with~$n=3$:
  \begin{equation}\label{eq:179}
     \begin{gathered} \xymatrix@R=.7em{\pi _4X_2\ar[r]\ar@{=}[d]&\pi
     _4X_3\ar[r] \ar@{=}[d]& \pi _1^s(BSO_3)_+\ar[r]\ar@{=}[d]& \pi _3X_2
     \ar@{=}[d] \\ \ZZ&\ZZ&\zt& 0} \end{gathered} 
  \end{equation}
and so the first homomorphism is multiplication by~$\pm2$ on generators.

  \end{proof}

  \begin{proof}[Proof of Proposition~\ref{thm:52}]
 As a preliminary we prove that $\pi _5X_3$~is finite.  Since this group is
finitely generated, an equivalent assertion is $\pi _5X_3\otimes \QQ=0$.  To
prove this observe that $X_3\to MSO$ induces an isomorphism on rational
homology in degrees~$\le 7$, whence also on rational homotopy groups in that
range.  Collating with~\eqref{eq:155} and facts in the previous proof we have
  \begin{equation}\label{eq:162}
     \pi \mstrut _{\{0,1,2,3,4,5\}}X_3\cong
     \{\ZZ,0,0,0,\ZZ,\textnormal{finite}\}.  
  \end{equation}
 
Introduce the mapping spectrum 
  \begin{equation}\label{eq:163}
     A=\Map(X_3,\Sigma ^5I\ZZ). 
  \end{equation}
From\footnote{More simply, the $\ZZ$-graded homotopy group of the Anderson
dual to~$X_3$ is the derived $R\Hom(\pi _{\bullet }X_3,\ZZ)$; there is a
shift of~5 in~\eqref{eq:164}}~\eqref{eq:122} and~\eqref{eq:162} we deduce
  \begin{equation}\label{eq:164}
     \pi\mstrut  _{\{0,1,2,3,4,5\}}A\cong \{0,\ZZ,0,0,0,\ZZ\},
  \end{equation}
and $\pi \mstrut _{\ge6}A=0$.  The cohomology group in~\eqref{eq:87}
is~$A^{\fA}(\RP^{\infty})$, which we compute using the Atiyah-Hirzebruch
spectral sequence.  The rows in the $E_2$~page, shown in Figure~\ref{fig:9},
  \begin{figure}[ht]
  \centering
  \includegraphics[scale=.8]{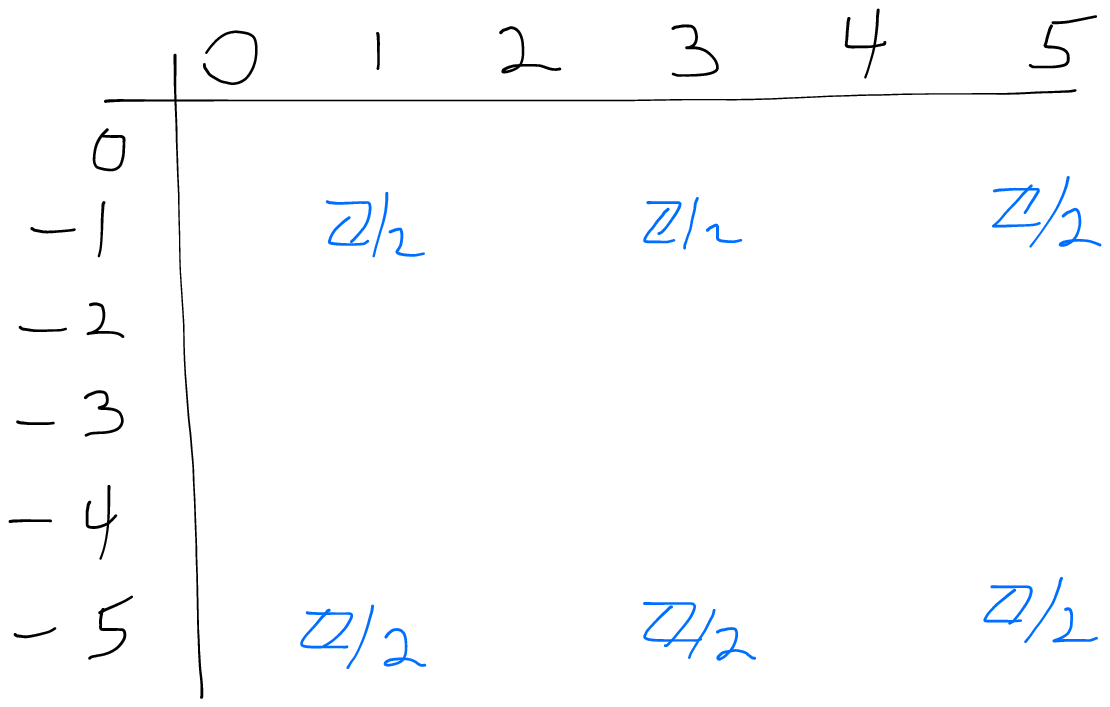}
  \caption{Computation of $A^{\fA}(\RP^{\infty})$}\label{fig:9}
  \end{figure}
are twisted cohomology groups of~$\RP^{\infty}$.  All differentials vanish in
this range, for degree reasons, whence $A^{\fA}(\RP^{\infty})$~is isomorphic
to $\zt\times \zt$ or~$\zmod4$, depending on whether there is a group
extension.
 
Define 
  \begin{equation}\label{eq:165}
     B=\Map(X_3,\Sigma ^5H\ZZ). 
  \end{equation}
Then 
  \begin{equation}\label{eq:166}
     \pi _jB = [\Sigma ^jX_3,\Sigma ^5H\ZZ]\cong H^{5-j}(X_3;\ZZ)\cong
     H^{5-j}(BSO_3;\ZZ), 
  \end{equation}
where the last step is the Thom isomorphism.  Hence 
  \begin{equation}\label{eq:167}
     \pi\mstrut _{\{0,1,2,3,4,5\}}B\cong \{0,\ZZ,\zt,0,0,\ZZ\}. 
  \end{equation}
The map $H\ZZ\to\IZ$ induces a map $B\to A$ and so a map of Atiyah-Hirzebruch
spectral sequences.  The $E_2$~page of the spectral sequence
for~$B^{\fB}(\RP^{\infty})$ is shown in Figure~\ref{fig:10}.  The group
  \begin{figure}[ht]
  \centering
  \includegraphics[scale=.8]{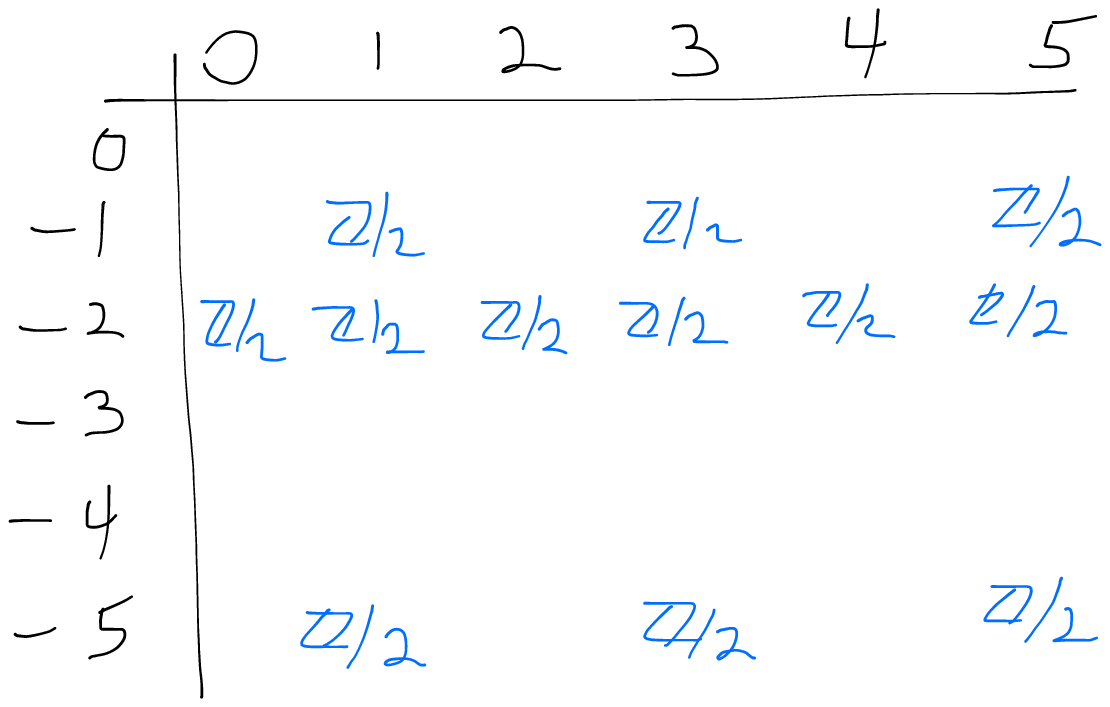}
  \caption{Computation of $B^{\fB}(\RP^{\infty})$}\label{fig:10}
  \end{figure}
$B^{\fB}(\RP^{\infty})\cong (\zt)^{\times 3}$ was computed
after~\eqref{eq:151}, and it follows that $d_2\:E_2^{1,-1}\to E_2^{3,-2}$
vanishes and there is no group extension passing from the degree~0 part of
the $E_\infty $~page to~$B^{\fB}(\RP^{\infty})$.  The map of spectral
sequences now implies that there is no group extension in the $A$-spectral
sequence either, that $A^{\fA}(\RP^{\infty})\cong (\zt)^{\times 2}$, and that
$i\:B^{\fB}(\RP^{\infty})\to A^{\fA}(\RP^{\infty})$ is surjective. 
  \end{proof}

\newcommand{\etalchar}[1]{$^{#1}$}
\providecommand{\bysame}{\leavevmode\hbox to3em{\hrulefill}\thinspace}
\providecommand{\MR}{\relax\ifhmode\unskip\space\fi MR }
\providecommand{\MRhref}[2]{%
  \href{http://www.ams.org/mathscinet-getitem?mr=#1}{#2}
}
\providecommand{\href}[2]{#2}

  \end{document}